\documentclass[11pt]{article}
\usepackage{amsmath}
\usepackage{amsfonts}
\usepackage{amssymb}
\usepackage{latexsym}
\usepackage{amsthm}
\usepackage{graphicx}
\usepackage{verbatim}
\usepackage[pagewise]{lineno}
\usepackage{cases}
\usepackage[arrow, matrix, curve]{xy}
\usepackage{pdflscape}

\newtheorem{theorem}{Theorem}
\newtheorem{proposition}[theorem]{Proposition}
\newtheorem{corollary}[theorem]{Corollary}
\newtheorem{remark}[theorem]{Remark}
\allowdisplaybreaks
\newcommand{\A}{\mathsf{A}}
\newcommand{\B}{\mathsf{B}}
\newcommand{\X}{\mathsf{X}}
\newcommand{\Y}{\mathsf{Y}}

\newcommand{\E}{\mathsf{E}}
\newcommand{\F}{\mathsf{F}}
\newcommand{\M}{\mathsf{M}}
\newcommand{\N}{\mathsf{N}}
\newcommand{\CC}{\mathsf{C}}
\newcommand{\AB}{\mathsf{AB}}

\newcommand{\PP}{\mathsf{P}}
\newcommand{\AN}{\mathsf{AN}}
\newcommand{\NB}{\mathsf{NB}}
\newcommand{\ANB}{\mathsf{ANB}}
\newcommand{\Finf}{{\mathsf{F}_\infty}}
\newcommand{\Fnull}{{\mathsf{F}_0}}
\newcommand{\Frho}{{\mathsf{F}_\rho}}
\newcommand{\PAo}{{\mathsf{P}_{\A,1}}}
\newcommand{\PAt}{{\mathsf{P}_{\A,2}}}
\newcommand{\PBo}{{\mathsf{P}_{\B,1}}}
\newcommand{\PBt}{{\mathsf{P}_{\B,2}}}
\newcommand{\mA}{m^{\A}}
\newcommand{\mB}{m^{\B}}

\newcommand{\mFnull}{m^{\Fnull}}
\newcommand{\mMo}{m^{\M_1}}
\newcommand{\mMf}{m^{\M_4}}
\newcommand{\mast}{m^{\ast}}
\newcommand{\mtilde}{{\tilde m}}
\newcommand{\mmax}{m_{\max}}
\newcommand{\meff}{m_{\mathsf{eff}}}
\newcommand{\phA}{\phi^{\A}}
\newcommand{\phB}{\phi^{\B}}
\newcommand{\phAB}{\phi^{\AB}}
\newcommand{\phABt}{\tilde\phi^{\AB}}
\newcommand{\phMo}{\phi^{\M_1}}
\newcommand{\phMot}{\tilde\phi^{\M_1}}
\newcommand{\phMf}{\phi^{\M_4}}
\newcommand{\phMft}{\tilde\phi^{\M_4}}
\newcommand{\phFnull}{\phi^{\Fnull}}
\newcommand{\phAFnull}{\phi^{\A \Fnull}}
\newcommand{\phBFnull}{\phi^{\B \Fnull}}
\newcommand{\phinv}{\phi_{\rm inv}}
\newcommand{\rhMo}{\rho^{\M_1}}
\newcommand{\rhMot}{\tilde{\rho}^{\M_1}}
\newcommand{\rhMf}{\rho^{\M_4}}
\newcommand{\rhMft}{\tilde{\rho}^{\M_4}}
\newcommand{\Fst}{F_{\rm ST}}
\newcommand{\Fstij}{F_{{\rm ST},ij}}
\newcommand{\Fstii}{F_{{\rm ST},ii}}

\newcommand{\tht}{\tilde{\theta}}
\def\a{\alpha}
\newcommand\be{\beta}

\def\th{\theta}

\newcommand\et{\eta}

\newcommand\ga{\gamma}

\newcommand\ka{\kappa}
\newcommand\la{\lambda}
\newcommand\ph{\phi}
\newcommand\rh{\rho}
\newcommand\si{\sigma}
\newcommand\ta{\tau}

\begin{document}
\normalbaselineskip16pt
\baselineskip16pt
\global\hoffset=-15truemm
\global\voffset=-10truemm
\numberwithin{equation}{section}
\numberwithin{theorem}{section}
\renewcommand{\topfraction}{.8}

\parindent15pt

\title{The consequences of gene flow for local adaptation and differentiation: A two-locus two-deme model}

\date{}

\maketitle

{\large
Ada Akerman and Reinhard B\"urger

\vspace{2cm}
Department of Mathematics, University of Vienna, Austria
}
\vspace{5cm}

{
\normalbaselineskip15pt
\baselineskip15pt{
{\obeylines 
{\bf Address for correspondence:}
\bigskip 
Ada Akerman
Institut f\"ur Mathematik 
Universit\"at Wien 
Nordbergstrasse 15
A-1090 Wien 
Austria
E-mail: ada.akerman@univie.ac.at
Phone: +43 1 4277 50784
Fax: +43 1 4277 9506}
}
}

\vspace{1.5cm}

\newpage
%\linenumbers
%\setpagewiselinenumbers
\section*{Abstract}

We consider a population subdivided into two demes connected by migration in which selection acts in opposite direction. We explore the effects of recombination and migration on the maintenance of multilocus polymorphism, on local adaptation, and on differentiation by employing a deterministic model with genic selection on two linked diallelic loci (i.e., no dominance or epistasis). For the following cases, we characterize explicitly the possible equilibrium configurations: weak, strong, highly asymmetric, and super-symmetric migration, no or weak recombination, and independent or strongly recombining loci. For independent loci (linkage equilibrium) and for completely linked loci, we derive the possible bifurcation patterns as functions of the total migration rate, assuming all other parameters are fixed but arbitrary. For these and other cases, we determine analytically the maximum migration rate below which a stable fully polymorphic equilibrium exists. In this case, differentiation and local adaptation 
are maintained. Their degree is quantified by a new multilocus version of $\Fst$ and by the migration load, respectively. In addition, we investigate the invasion conditions of locally beneficial mutants and show that linkage to a locus that is already in migration-selection balance facilitates invasion. Hence, loci of much smaller effect can invade than predicted by one-locus theory if linkage is sufficiently tight. We study how this minimum amount of linkage admitting invasion depends on the migration pattern. This suggests the emergence of clusters of locally beneficial mutations, which may form `genomic islands of divergence'. Finally, the influence of linkage and two-way migration on the effective migration rate at a linked neutral locus is explored. Numerical work complements our analytical results. 

\vskip2.5cm
{\bf Key words:} Selection, Migration, Recombination, Population subdivision, Genetic architecture, Multilocus polymorphism, Fixation index

%MSC2010 database: 92D10, 92D15, 37N25

\newpage
\setcounter{section}{0}

\section{Introduction}
Migration in a geographically structured population may have opposing effects on the genetic composition of that population and, hence, on its evolutionary potential. On the one hand, gene flow caused by migration may be so strong that it not only limits but hinders local adaptation by swamping the whole population with a genotype that has high fitness in only one or a few demes. On the other hand, if migration is sufficiently weak, gene flow may replenish local populations with genetic variation and contribute to future adaptation. In this case, locally adapted genotypes may coexist in the population and maintain high levels of genetic variation as well as differentiation between subpopulations. For reviews of the corresponding, well developed one-locus theory, see Karlin (1982), Lenormand (2002), and Nagylaki and Lou (2008). 

If selection acts on more than one locus, additional questions arise immediately. For instance, what are the consequences of the genetic architecture, such as linkage between loci, relative magnitude of locus effects or epistasis, on the degree of local adaptation and of differentiation achieved for a given amount of gene flow?  What are the consequences for genetic variation at linked neutral sites? What genetic architectures can be expected to evolve under various forms of spatially heterogeneous selection?

For selection acting on multiple loci, the available theory is much less well developed than for a single locus. One of the main reasons is that the interaction of migration and selection, even if the latter is nonepistatic, leads to linkage disequilibrium (LD) between loci (Li and Nei 1974, Christiansen and Feldman 1975, Slatkin 1975, Barton 1983). LD causes substantial, often insurmountable, complications in the analysis of multilocus models. Therefore, many multilocus studies are primarily numerical and focus on quite specific situations or problems. For instance, Spichtig and Kawecki (2004) investigated numerically the influence of the number of loci and of epistasis on the degree of polymorphism if selection acts antagonistically in two demes. Yeaman and Whitlock (2011) showed that concentrated genetic architecture, i.e., clusters of linked, locally beneficial alleles, evolve if stabilizing selection acts on a trait such that the fitness optima in two demes differ. 

Linkage disequilibrium is also essential for the evolution of recombination. The evolution of recombination in heterogeneous environments has been studied by a number of authors (e.g., Charlesworth and Charlesworth 1979, Pylkov et al.\ 1998, Lenormand and Otto 2000), and the results depend strongly on the kind of variability of selection across environments, the magnitude of migration, and the sign and strength of epistasis. 

Recent years have seen some advances in developing general theory for multilocus migration-selection models. The focus of this work was on the properties of the evolutionary dynamics and the conditions for the maintenance of multilocus polymorphism in limiting or special cases, such as weak or strong migration (B\"urger, 2009a,b), or in the Levene model (Nagylaki 2009; B\"urger 2009c, 2010; Barton 2010; Chasnov 2012). This progress was facilitated by the fact that in each case, LD is weak or absent. 

Using a continent-island-model framework, B\"urger and Akerman (2011) and Bank et al.\ (2012) analyzed the effects of gene flow on local adaptation, differentiation, the emergence of Dobzhansky-Muller incompatibilities, and the maintenance of polymorphism at two linked diallelic loci. They obtained analytical characterizations of the possible equilibrium configurations and bifurcation patterns for wide ranges of parameter combinations. In these models, typically high LD is maintained. In particular, explicit formulas were derived for the maximum migration rate below which a fully polymorphic equilibrium can be maintained, as well as for the minimum migration rate above which the island is swamped by the continental haplotype.

Here, we explore the robustness of some of these results by admitting arbitrary (forward and backward) migration between two demes. This generalization leads to substantial mathematical complications, but also to new biological insight. Because our focus is on the consequences of gene flow for local adaptation and differentiation, we assume divergent selection among the demes, i.e., alleles $A_1$ and $B_1$ are favored in deme 1, and $A_2$ and $B_2$ are favored in deme 2. The loci may recombine at an arbitrary rate. By ignoring epistasis and dominance, we assume genic selection. Mutation and random drift are neglected. Because we assume evolution in continuous time, our model also describes selection on haploids.

The model is set up in Section 2. In Section 3, we derive the equilibrium and stability structure for several important special cases. These include weak, strong, highly asymmetric, and super-symmetric migration, no or weak recombination, independent or strongly recombining loci, and absence of genotype-environment interaction. In Section 4, we study the dependence of the equilibrium and stability patterns on the total migration rate while keeping the ratio of migration rates, the recombination rate, and the selection coefficients constant (but arbitrary). In particular, we derive the possible bifurcation patterns for the cases of independent loci (linkage equilibrium) and for completely linked loci. With the help of perturbation theory, we obtain the equilibrium and stability configurations for weak or strong migration, highly asymmetric migration, and weak or strong recombination. For these cases, we determine the maximum migration rate below which a stable, fully polymorphic equilibrium is maintained, and 
the minimum migration rate above which the population is monomorphic. Numerical work complements our analytical results.

The next four sections are devoted to applications of the theory developed in Sections 3 and 4. In Section 5 and 6, we use the migration load and a new, genuine multilocus, fixation index ($\Fst$), respectively, to quantify the dependence of local adaptation and of differentiation on various parameters, especially, on the migration and the recombination rate. In Section 7, we investigate the invasion conditions for a mutant of small effect ($A_1$) that is beneficial in one deme but disadvantageous in the other deme. We assume that the mutant is linked to a polymorphic locus which is in selection-migration balance. We show that linkage between the loci facilitates invasion. Therefore, in such a scenario, clusters of locally adapted alleles are expected to emerge (cf.\ Yeaman and Whitlock 2011, B\"urger and Akerman 2011). In Section 8, we study the strength of barriers to gene flow at neutral sites linked to the selected loci by deriving an explicit approximation for the effective migration rate at a linked 
neutral site. Our results are summarized and discussed in Section 9. Several purely technical proofs are relegated to the Appendix.

\section{The model}\label{sec:model}
We consider a sexually reproducing population of monoecious, diploid individuals that is subdivided into two demes connected by genotype-independent migration. 
Within each deme, there is random mating. We assume that two diallelic loci are under genic selection, 
i.e., there is no dominance or epistasis, and different alleles are favored in different demes. We assume soft selection,
i.e., population regulation occurs within each deme.
We ignore random genetic drift and mutation and employ a deterministic continuous-time model 
to describe evolution. A continuous-time model is obtained from the corresponding discrete-time model in the limit of weak  evolutionary forces (here, selection, recombination, and migration). 

We denote the rate at which individuals in deme 1 (deme 2) are replaced by immigrants from the other deme by $m_1 \geq 0$ ($m_2 \geq 0$). Then $m=m_1+m_2$ is the total migration rate. The recombination rate between the two loci is designated by $\rh \geq 0$.

Alleles at locus $\A$ are denoted by $A_1$ and $A_2$, at locus $\B$ by $B_1$ and $B_2$. We posit that
$A_1$ and $B_1$ are favored in deme 1, whereas $A_2$ and $B_2$ are favored in deme 2. In deme $k$ ($k=1,2$), we assign the Malthusian parameters $\tfrac12\a_k$ and $-\tfrac12\a_k$ to $A_1$ and $A_2$, and $\tfrac12\be_k$ and $-\tfrac12\be_k$ to $B_1$
and $B_2$. Because we assume absence of dominance and of epistasis, the resulting fitness matrix for the genotypes reads
\begin{equation}
  \label{eq:fitnessmatrix}
  \bordermatrix{%
    & B_1B_1 & B_1B_2 & B_2B_2 \cr
  A_1A_1 & \a_k+\be_k & \a_k & \a_k-\be_k \cr
  A_1A_2 & \be_k & 0 & -\be_k \cr
  A_2A_2 & -\a_k+\be_k & -\a_k & -\a_k-\be_k \cr
  }.
\end{equation}

By relabeling alleles, we can assume without loss of generality $\a_1>0>\a_2$ and $\be_1>0>\be_2$.
Hence, $A_1B_1$ and $A_2B_2$ may be called the locally adapted haplotypes in deme 1 and deme 2, respectively. By relabeling loci, we can assume $\be_1\ge\a_1$. We define
\begin{equation}\label{eq:theta}
   \th=\a_1\be_2-\a_2\be_1.
\end{equation}
By exchanging demes, i.e., by the transformation $\tilde\a_k=-\a_{k^\ast}$ and $\tilde\be_k=-\be_{k^\ast}$ (where $k^\ast$ denotes the deme $\ne k$), or by exchanging loci, i.e., by the transformation $\tilde{\a}_k=\be_k$ and $\tilde{\be}_k=\a_k$, we can further assume $\th\ge0$ without loss of generality, cf.\ Appendix \ref{sec:parameters}. 

The fitness matrix \eqref{eq:fitnessmatrix} is also obtained if the two loci contribute additively to a quantitative trait that is under linear directional selection in each deme (B\"urger 2009c). Then $\th=0$ if the genotypic values are deme independent, i.e., if there is no genotype-environment interaction on the trait level.

Because in the case $\th=0$ degenerate features can occur, it will be treated separately (Sections \ref{sec:theta=0} and \ref{sec:super_symm}). 
Therefore, unless stated otherwise, we always impose the following assumptions on our parameters:
\begin{subequations}\label{eq:parameter}
\begin{equation}\label{assume}
		\a_1>0>\a_2 \text{ and } \be_1>0>\be_2,
\end{equation}
and
\begin{equation}\label{be1>a1}
	\be_1>\a_1,
\end{equation}
and
\begin{equation}\label{th>0}
		\th>0.
\end{equation}
\end{subequations}

From \eqref{assume} and \eqref{th>0}, we infer
\begin{equation}\label{be2<a2}
	\be_2<\a_2 \Rightarrow \a_1<\be_1.
\end{equation}
Therefore, locus $\A$ is under weaker selection than locus $\B$ in both demes, i.e., $|\a_k|\leq |\be_k|$ for $k=1,2$,
if and only if $\be_2<\a_2$ holds. 

The population can be described by the gamete frequencies in each of the demes. 
We denote the frequencies of the four possible gametes $A_1B_1$, $A_1B_2$, $A_2B_1$, and $A_2B_2$ in deme $k$ by $x_{1,k}$, $x_{2,k}$, $x_{3,k}$, and $x_{4,k}$. Then the state space is $S_4 \times S_4$, where \newline
$S_4=\left\{(x_1,x_2,x_3,x_4):\: x_i \geq 0\: \text{and}\: \sum_{i=1}^4x_i=1\right\}$ is the simplex. 

The following differential equations for the evolution of gamete frequencies in deme $k$ can be derived straightforwardly:
\begin{equation}\label{dynamics_gametes}
	\dot x_{i,k} = \frac{d}{dt} x_{i,k}= x_{i,k}(w_{i,k}-\bar w_k) - \et_i \rh D_k + m_k(x_{i,k^\ast}-x_{i,k}).
\end{equation}
Here the marginal fitness $w_{i,k}$ of gamete $i$ and the mean fitness $\bar w_k$ in deme $k$ are calculated from 
\eqref{eq:fitnessmatrix}, $\et_1=\et_4=-\et_2=-\et_3=1$, and $D_{k}=x_{1,k}x_{4,k}-x_{2,k}x_{3,k}$ is the linkage-disequilibrium (LD) measure. We note that $D_k>0$ corresponds to an excess of the locally adapted haplotypes in deme $k$.
The equations \eqref{dynamics_gametes} also describe the dynamics of a haploid population if in deme $k$ we assign the fitnesses 
$\a_k$, $-\a_k$, $\be_k$, $-\be_k$ to the alleles $A_1$, $A_2$, $B_1$, $B_2$, respectively.

Instead of gamete frequencies it is often more convenient to work with allele frequencies and the LD measures $D_k$. 
We write $p_k=x_{1,k}+x_{2,k}$ and $q_k=x_{1,k}+x_{3,k}$ for the frequencies of $A_1$ and $B_1$ in deme $k$. Then the gamete frequencies $x_{i,k}$ are calculated from the $p_k$, $q_k$, and $D_k$ by
\begin{subequations}
\begin{alignat}{2}
		x_{1,k} &=p_k q_k+D_k, \quad 		&	x_{2,k} &=p_k(1-q_k)-D_k,\\
		x_{3,k} &=(1-p_k)q_k-D_k,\quad	&	x_{4,k}	&=(1-p_k)(1-q_k)+D_k\,.
\end{alignat}
\end{subequations}
The constraints $x_{i,k} \ge 0$ and $\sum_{i=1}^4 x_{i,k}=1$ for $i=1,2,3,4,$ and $k=1,2$ transform into 
$0\le p_k, q_k \le 1$ and $-\min \left\{p_k q_k,(1-p_k)(1-q_k)\right\}\le D_k \le \min\left\{p_k(1-q_k),(1-p_k)q_k\right\}$. It follows that $p_k$, $q_k$, and $D_k$ evolve according to
\begin{subequations}\label{eq:dynamics}
\begin{align}
    \dot{p}_k&=\a_k p_k(1-p_k)+\be_k D_k + m_k(p_{k^\ast}-p_k),\\
    \dot{q}_k&=\be_k q_k(1-q_k)+\a_k D_k + m_k(q_{k^\ast}-q_k),\\
    \dot{D}_k&=\left[\a_k(1-2p_k)+\be_k(1-2q_k)-\rh\right]D_k \nonumber\\
	     			 &\quad+ m_k \left[(D_{k^\ast}-D_k)+(p_{k^\ast}-p_k)(q_{k^\ast}-q_k)\right].
\end{align}
\end{subequations}
We emphasize that, because we are treating a continuous-time model, the parameters $\rh$, $m_k$, $\a_k$, and $\be_k$ are rates (of recombination, migration, growth), whence they can be arbitrarily large. Their magnitude is determined by the time scale. By rescaling time, for instance to units of $\rh$ or $m$, the number of independent parameters could be reduced by one without changing the equilibrium properties.

\section{Equilibria and their stability}
  \label{sec:equilibria}
  
We distinguish three types of equilibria: (i) monomorphic equilibria (ME), (ii) single-locus polymorphisms (SLPs), and (iii) full 
(two-locus) polymorphisms (FPs). The first two types are boundary equilibria, whereas FPs are internal equilibria (except when $\rh=0$). 
The stability properties of the ME and the coordinates and conditions for admissibility of the SLPs can be derived explicitly.
However, the stability conditions for the SLPs and the conditions for existence or stability of FPs could be derived only for a number of limiting cases. These include strong recombination, weak or no recombination, weak, strong, or highly asymmetric migration.

\subsection{Existence of boundary equilibria}\label{sec:boundary_equil}
The four ME, corresponding to fixation of one of the gametes, exist always. Their coordinates are as follows:
\begin{align*}
		&\M_1 \; (A_1B_1 \text{ fixed}): \quad \hat{p}_k=1,\; \hat{q}_k=1,\; \hat{D}_k=0 \;\text{ for $k=1,2$}, \\
		&\M_2 \; (A_1B_2 \text{ fixed}): \quad \hat{p}_k=1,\; \hat{q}_k=0,\; \hat{D}_k=0 \;\text{ for $k=1,2$}, \\
		&\M_3 \; (A_2B_1 \text{ fixed}): \quad \hat{p}_k=0,\; \hat{q}_k=1,\; \hat{D}_k=0 \;\text{ for $k=1,2$}, \\
		&\M_4 \; (A_2B_2 \text{ fixed}): \quad \hat{p}_k=0,\; \hat{q}_k=0,\; \hat{D}_k=0 \;\text{ for $k=1,2$}, 
\end{align*}
where a $\hat{\phantom{x}}$ signifies an equilibrium. There are up to four SLPs, one in each marginal one-locus system.
We denote the SLPs where $B_1$ or $B_2$ is fixed by $\PP_{\A,1}$ or $\PP_{\A,2}$, respectively, 
and the SLPs where $A_1$ or $A_2$ is fixed by $\PP_{\B,1}$ or $\PP_{\B,2}$. Their coordinates and the conditions for their admissibility can be calculated explicitly (Eyland 1971). We define
\begin{equation}\label{eq:si_ta_def}
	\si_k = \frac{m_k}{\a_k}  \quad\text{and}\quad  \ta_k = \frac{m_k}{\be_k}\,.
\end{equation}
By \eqref{assume}, we have
\begin{equation}
	\si_1>0>\si_2  \quad\text{and}\quad  \ta_1>0>\ta_2\,.
\end{equation}
In addition, it is easy to show that the assumptions \eqref{eq:parameter} imply:
\begin{subequations}\label{cond_sita}
\begin{align}
		&\si_1+\si_2\ge0 \; \Rightarrow \; \ta_1+\ta_2<\si_1+\si_2, \\
		&\si_1+\si_2<0  \; \Rightarrow \; \ta_1+\ta_2<0.	
\end{align}
\end{subequations}

If locus $\B$ is fixed (for $B_1$ or $B_2$), the equilibrium allele frequencies at locus $\A$ are
\begin{equation}\label{SLPA}
	\hat p_1^\A = \frac12\left(1-2\si_1 + \sqrt{1-4\si_1\si_2}\right)\,, \quad
	\hat p_2^\A = \frac12\left(1-2\si_2 - \sqrt{1-4\si_1\si_2}\right)\,.
\end{equation}
If locus $\A$ is fixed, the equilibrium allele frequencies at locus $\B$ are given by
\begin{equation}\label{SLPB}
	\hat q_1^\B = \frac12\left(1-2\ta_1 + \sqrt{1-4\ta_1\ta_2}\right)\,, \quad
	\hat q_2^\B = \frac12\left(1-2\ta_2 - \sqrt{1-4\ta_1\ta_2}\right)\,.
\end{equation}
Thus, the four SLPs have the following coordinates:
\begin{subequations}
\begin{align}
		&\PP_{\A,1}: \quad \hat{p}_1=\hat p_1^\A,\; \hat{p}_2=\hat p_2^\A,\; \hat{q}_1=\hat{q}_2=1,\; \hat{D}_1=\hat{D}_2=0 , \\
		&\PP_{\A,2}: \quad \hat{p}_1=\hat p_1^\A,\; \hat{p}_2=\hat p_2^\A,\; \hat{q}_1=\hat{q}_2=0,\; \hat{D}_1=\hat{D}_2=0 , \\
		&\PP_{\B,1}: \quad \hat{p}_1=\hat{p}_2=1,\; \hat{q}_1=\hat{q}_1^\B,\; \hat{q}_2=\hat{q}_2^\B,\; \hat{D}_1=\hat{D}_2=0 , \\
		&\PP_{\B,2}: \quad \hat{p}_1=\hat{p}_2=0,\; \hat{q}_1=\hat{q}_1^\B,\; \hat{q}_2=\hat{q}_2^\B,\; \hat{D}_1=\hat{D}_2=0 , 
\end{align}
\end{subequations}
The equilibria $\PP_{\A,1}$ and $\PP_{\A,2}$ are admissible if and only if
\begin{equation}\label{si}
	|\si_1+\si_2| < 1,
\end{equation}
and the equilibria $\PP_{\B,1}$ and $\PP_{\B,2}$ are admissible if and only if
\begin{equation}\label{ta}
	|\ta_1+\ta_2| < 1.
\end{equation}

The SLPs leave the state space through one of their `neighboring' ME if $|\si_1+\si_2|$ or $|\ta_1+\ta_2|$ increases above 1. In particular, we find
\begin{subequations}\label{eq:ASLPbf}
\begin{alignat}{2}
   &\si_1+\si_2 \downarrow -1 &\;\iff\; \PP_{\A,1} \to \M_1 \text{ and }   \PP_{\A,2} \to \M_2,\label{eq:ASLPbf_1}\\
   &\si_1+\si_2 \uparrow 1  &\;\iff\; \PP_{\A,1} \to \M_3 \text{ and }   \PP_{\A,2} \to \M_4, \label{eq:ASLPbf_2}\\
   &\ta_1+\ta_2 \downarrow -1 &\;\iff\; \PP_{\B,1} \to \M_1 \text{ and }   \PP_{\B,2} \to \M_3,\\
   &\ta_1+\ta_2 \uparrow 1  &\;\iff\; \PP_{\B,1} \to \M_2 \text{ and }   \PP_{\B,2} \to \M_4. \label{eq:BSLPbf_2}
\end{alignat}
\end{subequations}
Throughout, we use $\downarrow$ to indicated convergence from above and $\uparrow$ to indicate convergence from below.
Figure 1 illustrates the location of the possible equilibria.

\begin{figure}
 \centering
 \includegraphics[scale=1.2,keepaspectratio=true]{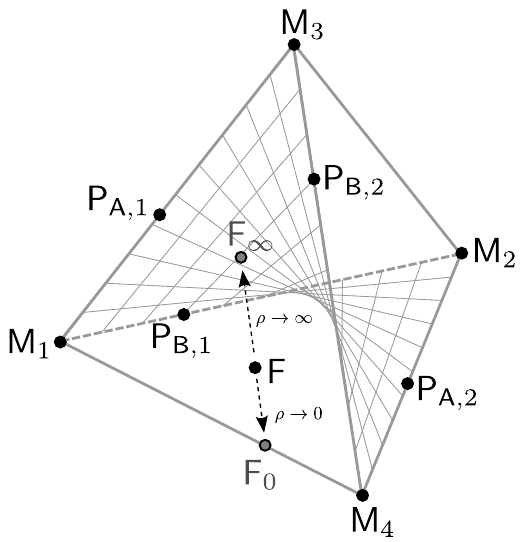}
 \caption{\small{Location of equilibria. In terms of gamete frequencies, the state space is $S_4\times S_4$, where each $S_4$ corresponds to one deme. This figure shows (schematically) the location in $S_4$ of all boundary equilibria and of the stable internal equilibrium $\F$. $\F$ converges to $\Finf$ if $\rho\to\infty$ and to $\Fnull$ if $\rho\to0$. The LE manifold is indicated by hatching.}}
 \label{fig:simplex}
\end{figure}

The SLPs are asymptotically stable within their marginal one-locus system if and only if they are admissible. 
Then they are also globally asymptotically stable within their marginal system (Eyland 1971). (We use globally stable in the sense that at least all trajectories from the interior of the designated set converge to the equilibrium.)
The reader may notice that \eqref{si} and \eqref{ta} are precisely the conditions for maintaining a protected polymorphism at locus $\A$ and $\B$, respectively.

\subsection{Stability of monomorphic equilibria}\label{sec:stab_mono}
At each monomorphic equilibrium, the characteristic polynomial factors into three quadratic polynomials. Two of them determine stability with respect to the marginal one-locus systems, whereas the third determines stability with respect to the interior
of the state space. The stability properties of the monomorphic equilibria are as follows. The proof is given in 
Appendix \ref{App_Proof_Prop_stabM1}.

\begin{proposition}\label{prop_stab_ME}
$\M_1$ is asymptotically stable if
\begin{equation}\label{M1_stab_marginal_orig}
	\si_1+\si_2<-1 \text{ and } \ta_1+\ta_2<-1
\end{equation}
and one of the following conditions hold
\begin{subequations}\label{M1_stab_rh}
\begin{equation}
	\rh \ge \min\{-\a_2, -\be_2 \}
\end{equation}
or
\begin{equation}\label{M1_r_klein,upperbound_orig}
	\rh <  \min\{-\a_2, -\be_2 \} \text{ and } m_2 > -\frac{(\a_1+\be_1+\rh+m_1)(\a_2+\be_2+\rh)}{\a_1+\be_1+\rh}.
\end{equation}
\end{subequations}

$\M_2$ is always unstable.

$\M_3$ is asymptotically stable if
\begin{equation}\label{M3_stab_marginal_orig}
	\si_1+\si_2>1 \text{ and } \ta_1+\ta_2<-1.
\end{equation}

$\M_4$ is asymptotically stable if
\begin{equation}\label{M4_stab_marginal_orig}
	\si_1+\si_2>1 \text{ and } \ta_1+\ta_2>1
\end{equation}
and one of the following conditions hold
\begin{subequations}
\begin{equation}
	\rh \ge \a_1
\end{equation}
or
\begin{equation}\label{M4_r_klein,upperbound_orig}
	\rh < \a_1 \text{ and } m_2 > \frac{(\a_1+\be_1-\rh-m_1)(\a_2+\be_2-\rh)}{\a_1+\be_1-\rh}.
\end{equation}
\end{subequations}

If one of the inequalities in \eqref{M1_stab_marginal_orig}, \eqref{M3_stab_marginal_orig}, or \eqref{M4_stab_marginal_orig},
or one of the inequalities for $m_2$ in \eqref{M1_r_klein,upperbound_orig} or \eqref{M4_r_klein,upperbound_orig} is reversed,
the respective equilibrium is unstable.
\end{proposition}

If we assumed $\th<0$, then $\M_3$ would always be unstable and $\M_2$ would be stable if $\si_1+\si_2<-1$ and $\ta_1+\ta_2>1$.

The above result shows that each of $\M_1$, $\M_3$, or $\M_4$ can be stable, but never simultaneously. 
For sufficiently loose linkage, the stability of a ME is determined solely by its stability within the two marginal one-locus systems in which it occurs. Stability of $\M_3$ is independent of the recombination rate. For given migration rates, the equilibria $\M_1$ and $\M_4$ may be stable for high recombination rates but unstable for low ones. For a low total migration rate ($m_1+m_2$), no ME is stable. For a sufficiently high total migration rate, there is a globally asymptotically stable ME (Section \ref{sec:strong_mig}).

\subsection{Stability of single-locus polymorphisms}\label{ref:stab_SLP}
As already mentioned, a single-locus polymorphism is globally attracting within its marginal one-locus system whenever it is admissible. Although the coordinates of the SLPs are given explicitly, the conditions for stability within the full, six-dimensional system on
$S_4\times S_4$ are uninformative because the four eigenvalues that determine stability transversal to the marginal one-locus system are solutions of a complicated quartic equation.

In the following we treat several limiting cases in which the conditions for stability of the SLPs and for existence and stability of FPs can be obtained explicitly.

\subsection{Weak migration}\label{ref:weak_mig}
The equilibrium and stability structure for weak migration can be deduced from the model with no migration by perturbation theory.
In the absence of migration ($m_1=m_2=0$), the two subpopulations evolve independently. Because selection is nonepistatic and there is no dominance, in each deme the fittest haplotype becomes eventually fixed. In fact, mean fitness is nondecreasing (Ewens 1969).
Our assumptions about fitness, i.e., \eqref{eq:fitnessmatrix} and \eqref{assume}, imply that
in deme 1 the equilibrium with $\hat p_1=\hat q_1=1$ and $\hat D_1=0$ is globally attracting, and in deme 2 the equilibrium with $\hat p_2=\hat q_2=0$ and $\hat D_2=0$ is globally attracting. Therefore, in the combined system, i.e., on $S_4\times S_4$, but still with $m_1=m_2=0$,
the (unique) globally attracting equilibrium is given by
\begin{equation}\label{F_nomig}
	\hat p_1=\hat q_1=1, \;\; \hat p_2=\hat q_2=0, \;\; \hat D_1=\hat D_2=0\,.
\end{equation}
All other equilibria are on the boundary and unstable. 

Because, generically, all equilibria in the system without migration are hyperbolic and it is a gradient system (Shahshahani 1979;
B\"urger 2000, p.\ 42),
Theorem 5.4 in B\"urger (2009a) applies and shows that the perturbation $\F$ of the equilibrium \eqref{F_nomig} is globally asymptotically stable for sufficiently small migration rates $m_1$ and $m_2$. Boundary equilibria remain unstable for sufficiently small migration rates. It is straightforward to calculate the coordinates of the perturbed equilibrium to leading order in $m_1$ and $m_2$. They are given by
\begin{subequations}\label{weak_mig}
\begin{alignat}{3}
		\hat p_1 &= 1-\frac{m_1}{\a_1} \frac{\a_1+\rh}{\a_1+\be_1+\rh} \,, &\quad
		\hat q_1 &= 1-\frac{m_1}{\be_1}\frac{\be_1+\rh}{\a_1+\be_1+\rh} \,, &\quad
		\hat D_1 &= \frac{m_1}{\a_1+\be_1+\rh} \,, \\
		\hat p_2 &= \frac{m_2}{-\a_2}\frac{\rh-\a_2}{\rh-\a_2-\be_2} \,, &\;
		\hat q_2 &= \frac{m_2}{-\be_2}\frac{\rh-\be_2}{\rh-\a_2-\be_2} \,, &\;
		\hat D_2 &= \frac{m_2}{\rh-\a_2-\be_2} \,.
\end{alignat}	
\end{subequations}
Therefore, we conclude

\begin{proposition}\label{prop_weak_mig}
For sufficiently weak migration, there is a unique, globally attracting, fully polymorphic equilibrium $\F$. To leading
order in $m_1$ and $m_2$, its coordinates are given by \eqref{weak_mig}.
\end{proposition}

Proposition \ref{prop_weak_mig} remains valid if the assumptions \eqref{be1>a1} and \eqref{th>0} are dropped.
Apart from the obvious fact that migration reduces differences between subpopulations,
the above approximations show that the lower the recombination rate, the smaller is this reduction. Thus, for given (small) migration rates, differentiation between subpopulations is always enhanced by reduced recombination. Linkage disequilibria within subpopulations are always positive.

\subsection{Linkage equilibrium}\label{ref:LE}
If recombination is so strong relative to selection and migration that linkage equilibrium (LE) can be assumed, i.e., if $\frac{1}{\rh}\max_{k=1,2}\{|\a_k|,|\be_k|,m_k\}\to0$, the dynamics \eqref{eq:dynamics} simplifies to
\begin{subequations}\label{dynLE}
\begin{align}
	\dot p_1 &= \a_1 p_1(1-p_1) + m_1 (p_2-p_1) \,, \label{dynLE1}\\
	\dot p_2 &=\a_2 p_2(1-p_2)  + m_2 (p_1-p_2) \,, \label{dynLE2}\\
	\dot q_1 &=\be_1 q_1(1-q_1) + m_1 (q_2-q_1) \,,\label{dynLE3}\\
	\dot q_2 &=\be_2 q_2(1-q_2) + m_2 (q_1-q_2)\,,\label{dynLE4}
\end{align}	
\end{subequations}
which is defined on $[0,1]^4$. 

In \eqref{dynLE}, the differential equations for the two loci are decoupled, i.e., 
\eqref{dynLE1} and \eqref{dynLE2} as well as  \eqref{dynLE3} and \eqref{dynLE4} form closed systems. Thus, the dynamics
of the full system is a Cartesian product of the two one-locus dynamics. Therefore,
in addition to the ME and to the SLPs determined above, the following internal equilibrium, denoted by $\Finf$,
may exist
\begin{equation}\label{Finf}
	\hat p_1^{\infty} = \hat p_1^\A, \quad
	\hat p_2^{\infty} = \hat p_2^\A, \quad
	\hat q_1^{\infty} = \hat q_1^\B, \quad
	\hat q_2^{\infty} = \hat q_2^\B,
\end{equation}
where the $\hat p_k^\A$ and $\hat q_k^\B$ are given by \eqref{SLPA} and \eqref{SLPB}, respectively. No other internal
equilibrium can exist.
This equilibrium is admissible if and only if \eqref{si} and \eqref{ta} are satisfied, i.e., if and only if all four SLPs are admissible.

Because in the one-locus model the FP is globally asymptotically stable (hence, it attracts all trajectories from the interior of the state space) whenever it is admissible (Eyland 1971; Hadeler and Glas 1983, Theorem 2; Nagylaki and Lou 2008, Section 4.3.2), and because the full dynamics is the Cartesian product of the one-locus dynamics, the fully polymorphic equilibrium
$\Finf$ is globally asymptotically stable whenever it is admissible.
Similarly, we conclude that a boundary equilibrium is globally asymptotically stable whenever it is asymptotically stable in the full system. These results in combination with those in Sections \ref{sec:boundary_equil} and \ref{sec:stab_mono} yield the following proposition.

\begin{proposition}\label{prop_LE}
Assume \eqref{dynLE}. Then a globally asymptotically stable equilibrium exists always. This equilibrium
is internal, hence equals $\Finf$ \eqref{Finf}, if and only if \eqref{si} and \eqref{ta} hold. It is a SLP if one of
\eqref{si} or \eqref{ta} is violated, and a ME if both \eqref{si}and \eqref{ta} are violated.

If, by variation of parameters, the internal equilibrium leaves (or enters) the state space, generically, it does so through one of the SLPs. The precise conditions are:
\begin{subequations}\label{eq:Einfinitybf}
\begin{align}
    \Finf \to \PAo  & \;\iff\; \tau_1+\tau_2 \downarrow -1 \text{ and } |\si_1+\si_2|<1, \label{eq:Einfinitybf_a}\\
    \Finf \to \PAt  & \;\;\text{does not occur}, \label{eq:Einfinitybf_b}\\
    \Finf \to \PBo  & \;\iff\; \si_1+\si_2 \downarrow -1 \text{ and } |\tau_1+\tau_2|<1, \label{eq:Einfinitybf_c}\\
    \Finf \to \PBt  & \;\iff\; \si_1+\si_2 \uparrow 1 \text{ and } |\tau_1+\tau_2|<1. \label{eq:Einfinitybf_d}
\end{align}
\end{subequations}
When, upon leaving the state space, $\Finf$ collides with a boundary equilibrium (SLP or ME), the respective boundary equilibrium becomes globally asymptotically stable.
\end{proposition}

We note that $\Finf \to \PAt $ does not occur because it requires $\tau_1+\tau_2 \uparrow 1$ and $|\si_1+\si_2|<1$,
which is impossible by \eqref{cond_sita}.
We leave the simple determination of the conditions for bifurcations of $\Finf$ with one of the ME to the interested reader.

Proposition \ref{prop_LE} can be extended straightforwardly to an arbitrary number of loci because the dynamics at each locus is independent of that at the other loci. This decoupling of loci occurs because there is no epistasis.

\subsection{Strong recombination: quasi-linkage equilibrium}\label{sec:QLE}
If recombination is strong, a regular perturbation analysis of the internal equilibrium $\Finf$ of \eqref{dynLE} can be performed.
The allele frequencies and linkage disequilibria can be calculated to order $1/\rh$. Formally, we set
\begin{equation}
	m_k = \mu_k/\rh \quad (k=1,2),
\end{equation}
keep $\si_k$ and $\ta_k$ constant, and let $\rh\to\infty$. Then, we obtain
\begin{subequations}\label{F_largerho}
\begin{align}
	\hat p_1 &= \hat p_1^{\infty} + \frac{\si_1}{\rh}	
	\frac{\si_2(\be_1-\be_2)+\be_1\sqrt{1-4\si_1\si_2}}{\sqrt{1-4\si_1\si_2}} \left(\hat q_1^{\infty}-\hat q_2^{\infty} \right) 
			+ O(\rh^{-2})\,,\\
	\hat q_1 &= \hat q_1^{\infty} +\frac{\ta_1}{\rh}
			\frac{\ta_2(\a_1-\a_2)+\a_1\sqrt{1-4\ta_1\ta_2}}{\sqrt{1-4\ta_1\ta_2}} \left(\hat p_1^{\infty}-\hat p_2^{\infty} \right)
			+ O(\rh^{-2})\,,\\
	\hat D_1 &= \frac{m_1}{\rh}\left(\hat p_1^{\infty}-\hat p_2^{\infty}\right) 
						\left(\hat q_1^{\infty}-\hat q_2^{\infty} \right) + O(\rh^{-2})\,,
\end{align}
\end{subequations}
and analogous formulas hold for the second deme. Because LD is of order $1/\rh$, this approximation may be called the quasi-linkage equilibrium approximation of the fully polymorphic equilibrium (Kimura 1965, Turelli and Barton 1990, Nagylaki et al.\ 1999). We note that LD is positive in both demes and increases with increasing differentiation between the demes, increasing migration, or decreasing recombination.

Proposition 5.1 in B\"urger (2009a) shows that in every small neighborhood of an equilibrium of the model with LE \eqref{dynLE}, there is one equilibrium of the perturbed system, and it has the same stability properties as the unperturbed equilibrium. Because of the simple structure of \eqref{dynLE}, a stronger result can be obtained.
In an isolated one-locus system on $[0,1]^2$ (e.g., \eqref{dynLE1} and \eqref{dynLE2}), every trajectory from the interior converges to the unique asymptotically stable equilibrium (Section \ref{ref:LE}), and the chain-recurrent points (Conley 1978) are the equilibria.
Therefore, the same holds for the LE dynamics \eqref{dynLE}, and the regular global perturbation result of Nagylaki et al.\ (1999) (the proof of their Theorem 2.3) applies for large $\rh$. Hence the dynamical behavior with strong recombination is qualitatively the same as that under LE. We conclude that for sufficiently strong recombination every asymptotically stable equilibrium is globally asymptotically stable.

\subsection{No recombination}\label{sec:no_rec}
Let recombination be absent, i.e., $\rh=0$. 
Then, effectively, we have a one-locus model in which the four alleles correspond to the four gametes
$A_1B_1$, $A_1B_2$, $A_2B_1$, $A_2B_2$.
In deme $k$, they have the selection coefficients $\tfrac12(\a_k+\be_k)$, $\tfrac12(\a_k-\be_k)$,
$\tfrac12(-\a_k+\be_k)$, $-\tfrac12(\a_k+\be_k)$, respectively.
According to Theorem 2.4 of Nagylaki and Lou (2001), generically, no more than two gametes can be present at an equilibrium. We will prove a stronger result and characterize all possible equilibria and their local stability. 

Because $\rh=0$, there may be a polymorphic equilibrium at which only the gametes $A_1B_1$ and $A_2B_2$ are present. We call
it $\Fnull$ and set 
\begin{equation}\label{eq:kappa}
	\ka_k = \frac{m_k}{\a_k+\be_k} \quad (k=1,2)\,.
\end{equation}
Then one-locus theory (Section \ref{sec:boundary_equil}) informs us that $\Fnull$ is admissible if and only if
\begin{equation}\label{eq:Fo_admissible}
	|\ka_1+\ka_2| < 1.
\end{equation}
Its coordinates are given by
\begin{subequations}\label{coord_Fnull}
\begin{align}
		\hat p_1^0 &= \hat q_1^0 = \frac12\left(1-2\ka_1 + \sqrt{1-4\ka_1\ka_2}\right)\,, \\
		\hat p_2^0 &= \hat q_2^0 = \frac12\left(1-2\ka_2 - \sqrt{1-4\ka_1\ka_2}\right)\,, \\
		\hat D_k^0 &= \hat p_k(1-\hat p_k) \quad (k=1,2)\,, 
\end{align}
\end{subequations}
where $\hat p_k^0=\hat q_k^0=\hat x_{1,k}^0$, $\hat x_{2,k}^0=\hat x_{3,k}^0=0$, and $\hat x_{4,k}^0=1-\hat x_{1,k}^0$ ($k=1,2$). Within the subsystem in which only the gametes $A_1B_1$ and $A_2B_2$ are present, $\Fnull$ is asymptotically
stable whenever it is admissible. One-locus theory implies that
\begin{subequations} \label{eq:E0leavesstatespace}
\begin{align}
    \kappa_1+\kappa_2\downarrow -1 \;\iff\; \Fnull \to \M_1,\\
    \kappa_1+\kappa_2\uparrow 1 \;\iff\; \Fnull \to \M_4.
\end{align}
\end{subequations}

A simple application of Corollary 3.9 of Nagylaki and Lou (2007) shows that the
gamete $A_1B_2$ will always be lost (to apply their result, recall assumptions \eqref{eq:parameter} and use $\ga_{22}=\ga_{23}=0$,
$\frac{\a_1}{\a_1+\be_1}<\ga_{21}<\frac{\a_2}{\a_2+\be_2}$, $\ga_{24}=1-\ga_{21}$). This strengthens the result in Section
\ref{sec:stab_mono} that $\M_2$ is always unstable. Thus, we are left with the
analysis of the tri-gametic system consisting of $A_1B_1$, $A_2B_1$, and $A_2B_2$. (If $\th<0$, then gamete $A_2B_1$ is lost.)

In Appendix \ref{int_equi_rho0} it is proved that $\Fnull$ is the only equilibrium at which both loci are polymorphic except when 
\begin{equation}\label{critmu12}
	m_1m_2 = \mtilde
\end{equation}
holds, where
\begin{equation}\label{mcrit}
	\mtilde = -\a_1\a_2\be_1\be_2(\a_1+\be_1)(\a_2+\be_2)/\th^2.
\end{equation}
If \eqref{critmu12} holds, then there is a line of internal equilibria connecting $\Fnull$ with $\PP_{\A,1}$ or $\PP_{\B,2}$ (or $\M_3$); see Appendix \ref{int_equi_rho0}.

We find that $\Fnull$ is asymptotically stable if 
\begin{equation}\label{stab_F0}
	m_1m_2 < \mtilde,
\end{equation}
and unstable if the inequality is reversed (Appendix \ref{sec:Fnull_stability}). For sufficiently small migration rates, Proposition \ref{prop_weak_mig} implies that $\Fnull=\F$ and $\Fnull$ is globally asymptotically stable. If the inequality in \eqref{stab_F0} is reversed, $\Fnull$ may or may not be admissible. 

Of course, if $\Fnull$ is asymptotically stable, then the equilibria $\M_1$ and $\M_4$ are unstable; cf.\ \eqref{eq:E0leavesstatespace}.
The following argument shows that $\M_3$ cannot be simultaneously stable with $\Fnull$. We rewrite \eqref{stab_F0} as
\begin{equation}\label{stab_F0kap}
	\ka_1\ka_2 > -\frac{\a_1\a_2\be_1\be_2}{\th^2} 
		= -\frac{\si_1\si_2\ta_1\ta_2}{(\si_1\ta_2-\si_2\ta_1)^2}\,.
\end{equation}
Because
\begin{equation}
	\ka_k^{-1} = \si_k^{-1} + \ta_k^{-1}\,,
\end{equation}
\eqref{stab_F0kap} becomes
\begin{equation}\label{stabF0_sigtau}
	(\si_1\ta_2-\si_2\ta_1)^2 + (\si_1+\ta_1)(\si_2+\ta_2) < 0\,.
\end{equation}
Since $\M_3$ is asymptotically stable if \eqref{M3_stab_marginal_orig} holds and because, as is easy to show,  
\eqref{M3_stab_marginal_orig} and \eqref{stabF0_sigtau} are incompatible, the assertion follows. It can also be shown from
\eqref{M3_stab_marginal_orig} and \eqref{stabF0_sigtau} that $\M_3$ cannot become stable when $\Fnull$ loses its 
stability except in the degenerate case when $\si_1+\si_2=1$ and $\ta_1+\ta_2=-1$.

In our tri-gametic system, $\PP_{\A,1}$ and $\PP_{\B,2}$ are the only possible SLPs. They may exist simultaneously with $\Fnull$ if \eqref{stab_F0} holds, i.e., if $\Fnull$ is stable, but not otherwise (Appendix \ref{sec:PB2_PA1_rhnull_stability}). If \eqref{stab_F0} holds, both are unstable (if admissible). $\PP_{\A,1}$ or $\PP_{\B,2}$ have an eigenvalue 0 if and only if \eqref{critmu12} holds or if they leave or enter the state space through a ME. In Appendix \ref{sec:PB2_PA1_rhnull_stability} it is shown that $\PP_{\A,1}$ is asymptotically stable if and only if
\begin{equation}\label{PA1_stab_r0}
  \ta_1+\ta_2<-1 \;\text{ and }\; |\si_1+\si_2|<1  \;\text{ and }\;  m_1m_2 > \mtilde ,
\end{equation}
and $\PP_{\B,2}$ is asymptotically stable if and only if
\begin{equation}\label{PB2_stab_r0}
  1<\si_1+\si_2 \;\text{ and }\; |\ta_1+\ta_2|<1  \;\text{ and }\;  m_1m_2 > \mtilde.
\end{equation}
Hence, if $m_1m_2$ increases above $\mtilde$, the SLP that is admissible becomes asymptotically stable. 
Upon collision of the stable SLP with one of the adjacent ME, the corresponding ME becomes stable and remains so for all higher migration rates. We summarize these findings as follows:

\begin{proposition}\label{prop_norec}
Except in the degenerate case when \eqref{critmu12} holds, only equilibria with at most two gametes present exist. 
If \eqref{eq:Fo_admissible} is satisfied, the
equilibrium $\Fnull$ given by \eqref{coord_Fnull} is admissible. If, in addition, \eqref{stab_F0} is fulfilled, then $\Fnull$ is asymptotically stable. For sufficiently small migration rates, it is globally asymptotically stable. If $\Fnull$ is unstable or not admissible, then one of the ME $(\M_1$, $\M_3$, $\M_4)$ or one of the SLPs $(\PP_{\A,1}$, $\PP_{\B,2})$
is asymptotically stable. If \eqref{critmu12} holds, then there is a line of equilibria with three gametes present.
\end{proposition}

The proposition shows that, except for the nongeneric case when \eqref{critmu12} holds, there is always precisely one stable equilibrium point. Numerical results support the conjecture that the stable equilibrium is globally asymptotically stable. Bifurcation patterns as functions of $m$ are derived in Section \ref{sec:bifs_norec}. 

In addition to $\Fnull$, there exists a second FP on the edge connecting $\M_2$ and $\M_3$. Although its coordinates can be calculated easily, it is not of interest here as it is unstable for every choice of selection and migration parameters. This unstable equilibrium leaves the state space under small perturbations, i.e., if $\rh>0$.

\subsection{Highly asymmetric migration}\label{Section_asym}
All special cases treated above suggest that there always exists a globally asymptotically stable equilibrium. This, however, is generally not true as was demonstrated by the analysis of the two-locus continent-island (CI)
model in B\"urger and Akerman (2011). There, all possible bifurcation patterns were derived and it was shown that the fully polymorphic equilibrium can be simultaneously stable with a boundary equilibrium.
For highly asymmetric migration rates, the equilibrium and stability structure can be obtained by a perturbation analysis of this CI model.

Therefore, we first summarize the most relevant features of the analysis in B\"urger and Akerman (2011). Because in that analysis the haplotype $A_2B_2$ is fixed on the continent (here, deme 2) and there is no back migration ($m_2=0$), it is sufficient to treat the dynamics on the island (here, deme 1)
where immigration of $A_2B_2$ occurs at rate $m_1$. Thus, the state space is $S_4$. 

It was shown that up to two internal (fully
polymorphic) equilibria, denoted by $\E_+$ and $\E_-$, may exist. Only one ($\E_+$) can be stable. Two SLPs, $\E_\A$ and $\E_\B$, may exist. At $\E_\A$, locus $\A$ is polymorphic and allele $B_2$ is fixed; at  $\E_\B$, locus $\B$ is polymorphic and allele $A_2$ is fixed.
$\E_\A$ ($\E_\B$) is admissible if and only if $m_1<\a_1$ ($m_1<\be_1$). $\E_\A$ is always unstable. Finally, there always exists the
monomorphic equilibrium $\E_\CC$ at which the haplotype $A_2B_2$ is fixed on the island. The equilibrium coordinates of all equilibria were obtained explicitly. In addition, it was proved (see also Bank et al.\ 2012, Supporting Information, Theorem S.4) that precisely the following two types of bifurcation patterns can occur:

Type 1. There exists a critical migration rate $m^\bullet>0$ such that:
\begin{itemize} 
\item If $0<m_1<m^\bullet$, a unique internal equilibrium,  $\E_+$, exists. It is asymptotically stable and, presumably, globally asymptotically stable.
\item At $m_1=m^\bullet$, $\E_+$ leaves the state space through a boundary equilibrium ($\E_\B$ or $\E_\CC$) by an exchange-of-stability
bifurcation. 
\item If $m_1>m^\bullet$, a boundary equilibrium ($\E_\B$ or $\E_\CC$) is asymptotically stable and, presumably, globally stable.
\end{itemize}

Type 2. There exist critical migration rates $m^\circ$ and $m^\bullet$ satisfying $m^\bullet>m^\circ>0$ such that:
\begin{itemize} 
\item If $0<m_1<m^\circ$, there is a unique internal equilibrium $(\E_+)$. It is asymptotically stable and, presumably,
globally stable.
\item At $m_1=m^\circ$, an unstable equilibrium $(\E_-)$ enters the state space by an exchange-of-stability
bifurcation with a boundary equilibrium ($\E_\B$ or $\E_\CC$).
\item If $m^\circ<m_1<m^\bullet$, there are two internal equilibria, one asymptotically stable $(\E_+)$, 
the other unstable $(\E_-)$, and one of the boundary equilibria ($\E_\B$ or $\E_\CC$) is asymptotically stable.
\item At $m_1=m^\bullet$, the two internal equilibria merge and annihilate each other by a saddle-node bifurcation.
\item If $m_1>m^\bullet$, a boundary equilibrium ($\E_\B$ or $\E_\CC$) is asymptotically stable and, presumably, globally stable.
\end{itemize}

For sufficiently large migration rates ($m>m^{\bullet\bullet}\ge m^\bullet$), $\E_\CC$ is globally asymptotically stable in both cases. Bifurcation patterns of Type 2 occur only if the recombination rate is intermediate, i.e., if $\rh$ is about as large as $\a_1$.

By imbedding the CI model into the two-deme dynamics, \eqref{dynamics_gametes} or \eqref{eq:dynamics}, perturbation theory can be applied to obtain analogous results for highly asymmetric migration, i.e., for sufficiently small $m_2/m_1$ (Karlin and McGregor 1972). This is so because all equilibria in the CI model are hyperbolic except when collisions between equilibria occur (B\"urger and Akerman 2011). Since the coordinates of the internal
equilibria $\E_+$ and $\E_-$ were derived, the perturbed equilibrium frequencies can be obtained. Because they are too complicated to be informative, we do not present them. The perturbation of $\E_+$, denoted by $\F$, is asymptotically stable. As $\E_-$ is internal, it cannot be lost by a small perturbation. Also the boundary equilibria and their stability properties are preserved under small perturbations. In particular, $\E_\CC$ gives rise to $\M_4$, and the SLPs $\E_\A$ and $\E_\B$ give rise to $\PP_{\A,2}$ and $\PP_{\B,2}$, respectively,

If recombination is intermediate, (at least) under highly asymmetric two-way migration, one stable and one unstable FP can coexist. In this case the stable FP, $\F$, is simultaneously stable with either $\M_4$ or $\PP_{\B,2}$. 
Although there is precisely one (perturbed) equilibrium in a small neighborhood of every equilibrium of the CI model, we can not exclude that other internal equilibria or limit sets are generated by perturbation.

\subsection{The case $\th=0$}\label{sec:theta=0}
The analyses in the previous sections are based on the assumptions \eqref{eq:parameter}, in particular, on $\th>0$. 
However, many of the results obtained above remain valid if $\th=0$. Here, we point out the necessary adjustments. 

Without loss of generality, we can assume 
\begin{equation}\label{eq:parameterdeg}
	|\a_k|\leq |\be_k| \text{ for } k=1,2
\end{equation}
in addition to $\th=0$ and \eqref{assume}. Then we observe that
\begin{equation}
    \si_1+\si_2=\frac{\be_2}{\a_2}(\ta_1+\ta_2)\geq \ta_1+\ta_2 \;\text{ and }\; \ka_1+\ka_2=\frac{\be_2}{\a_2+\be_2}(\ta_1+\ta_2)<\ta_1+\ta_2.
\end{equation}
Therefore, either
\begin{subequations}\label{eq:th0_order}
\begin{align}
    &0<\ka_1+\ka_2<\ta_1+\ta_2\leq \si_1+\si_2 \label{eq:th0_order_1}\\
\intertext{or}
    &0>\ka_1+\ka_2>\ta_1+\ta_2\geq \si_1+\si_2 \label{eq:th0_order_2}\\
\intertext{or}
    & \ka_1+\ka_2 = \ta_1+\ta_2 = \si_1+\si_2 = 0 \label{eq:th0_order_3}
\end{align}
\end{subequations}
applies, where equality in \eqref{eq:th0_order_1} and \eqref{eq:th0_order_2} holds if $\a_k=\be_k$ ($k=1,2$). In addition,
\begin{equation}
	\ka_1+\ka_2=0 \;\iff\; \ta_1+\ta_2=0 \;\iff\; \si_1+\si_2=0. \label{eq:all_protected}
\end{equation}
 
With these preliminaries, we can treat the changes required in the above propositions if $\th=0$.

From \eqref{eq:th0_order} we infer that, in Proposition \ref{prop_stab_ME}, not only $\M_2$ but also $\M_3$ is always unstable. In addition, if $0>\ta_1+\ta_2\geq \si_1+\si_2$, then $\M_1$ is asymptotically stable for sufficiently strong migration, whereas
$\M_4$ is stable for sufficiently strong migration if $0<\ta_1+\ta_2\leq \si_1+\si_2$ holds.

As already noted, Proposition \ref{prop_weak_mig} remains valid independently of the value of $\th$.

In Proposition \ref{prop_LE}, the only SLPs through which the internal equilibrium $\Finf$ can leave the state space are
$\PBo$ and $\PBt$; see \eqref{eq:Einfinitybf_c} and \eqref{eq:Einfinitybf_d}. The reason is that, except when $\si_1+\si_2=0$ (and \eqref{eq:all_protected} applies), $\PAo$ and $\PAt$ are only admissible if $\PBo$ and $\PBt$ are. Thus, the locus under weaker selection always becomes monomorphic at lower rates of gene flow than the locus under stronger selection.

If $\rh=0$ (Proposition \ref{prop_norec}), $\Fnull$ is asymptotically stable whenever it is admissible because $\tilde m\to\infty$
as $\th\to0$; see \eqref{stab_F0}. In addition, \eqref{eq:th0_order} implies that $\Fnull$ persists stronger gene flow than the SLPs, which are always unstable; see \eqref{PA1_stab_r0} and \eqref{PB2_stab_r0}. 

In the highly symmetric case of \eqref{eq:th0_order_3}, SLPs cannot be lost. Thus, $\Finf$ is always admissible and globally stable, cf.\ Proposition \ref{prop_LE}. If $\rh=0$, \eqref{eq:all_protected} implies that $\Fnull$ exists always (and is stable). In the next section we show that in this highly symmetric case the FP is always admissible for arbitrary recombination rates.

\subsection{The super-symmetric case}\label{sec:super_symm}
In many, especially ecological, applications highly symmetric migration-selection models are studied. Frequently made assumptions are that the migration rates between the demes are identical ($m_1=m_2$), selection in deme 2 mirrors that in deme 1 ($\a_k=-\a_{k^\ast}$), and the loci are equivalent ($\a_k=\be_k$). Thus, $\th=0$ and \eqref{eq:th0_order_3} holds, which we assume now. 

Conditions \eqref{si} and \eqref{ta} imply that all four SLPs are admissible. Hence, all monomorphisms are unstable. In addition, it can be proved that all SLPs are unstable (Appendix \ref{sec:super_symmetric}). If migration is weak, a globally asymptotically stable, fully polymorphic equilibrium ($\F$) exists (Proposition \ref{prop_weak_mig}). 

Because every boundary equilibrium is hyperbolic for every parameter choice, the index theorem of Hofbauer (1990) can be applied. Since none of the boundary equilibria is saturated, it follows that an internal equilibrium with index 1 exists. 
For small migration rates, this is $\F$ because it is unique. Since the boundary equilibria are always hyperbolic, no internal equilibrium can leave the state space through the boundary. However, we cannot exclude that the internal equilibrium undergoes a pitchfork or a Hopf bifurcation. Numerical results support the conjecture that the internal equilibrium is unique and globally attracting, independently of the strength of migration. This is a very special feature of this super-symmetric case; cf.\ Proposition \ref{prop_strong_mig}.

\subsection{General case}\label{Section_general}
Because a satisfactory analysis for general parameter choices seems out of reach, we performed extensive numerical work to determine the possible equilibrium structures. In no case did we find more complicated equilibrium structures than indicated above, i.e., apparently there are never more than two internal equilibria. If there is one internal equilibrium, it appears to be globally asymptotically stable. If there are two internal equilibria, then one is unstable and the other is simultaneously stable with one boundary equilibrium (as in the CI model). Apparently, two internal equilibria occur only for sufficiently asymmetric migration rates and only if the recombination rate is of similar magnitude as the selective coefficients. 

A glance at the dynamical equations \eqref{eq:dynamics} reveals that an internal equilibrium can be in LE only if $p_1=p_2$ or $q_1=q_2$. From \eqref{Finf}, \eqref{SLPA} and \eqref{SLPB}, we find that this can occur only if $|\si_1+\si_2|=1$ or $|\ta_1+\ta_2|=1$, i.e., for a boundary equilibrium. Thus, internal equilibria always exhibit LD. 

For low migration rates as well as for high recombination rates, there is a unique, fully polymorphic equilibrium which is globally asymptotically stable and exhibits positive LD (Sections \ref{ref:weak_mig} or \ref{sec:QLE}). We denote the (presumably unique) asymptotically stable, fully polymorphic equilibrium by $\F$. If migration is weak, or recombination is weak, or recombination is strong, we have proved that $\F$ is unique. Useful approximations are available for weak migration or strong recombination; see \eqref{weak_mig} or \eqref{F_largerho}.
Finally, for sufficiently high migration rates one of the monomorphic equilibria is globally asymptotically stable.

%Section 4, CTTIM

\section{Bifurcation patterns and maintenance of polymorphism}
\label{sec:bifs}
Here we study how genetic variation and polymorphism depend on the strength and pattern of migration. In particular, we are interested in determining how the maximum migration rate that permits genetic polymorphism depends on the other parameters. For this end, we explore properties of our model, such as the possible bifurcation patterns, as functions of the total migration rate $m$. We do this by assuming that $\a_1$, $\a_2$, $\be_1$, $\be_2$ , $\rh$, and the migration ratio 
\begin{equation}\label{phi}
  \phi=\frac{m_1}{m} ,
\end{equation}
where $m>0$ and $0\le\ph\le1$, are constant. The values $\phi=0$ and $\phi=1$ correspond to one-way migration, as in the
CI model. If $\phi=\tfrac12$, migration between the demes is symmetric, an assumption made in many studies of migration-selection models. Fixing $\ph$ and treating $m$ as the only migration parameter corresponds to the migration scheme introduced by Deakin (1966).

\subsection{Important quantities}\label{sec:important_quantities}
We define several important quantities that will be needed to describe our results and summarize the relevant relations between them. Let 
\begin{subequations}\label{eq:phAphB}
\begin{align}
   \phA &=\frac{\a_1}{\a_1-\a_2}, \label{eq:phA}\\
   \phB &=\frac{\be_1}{\be_1-\be_2}, \\
   \phFnull &=\frac{\a_1+\be_1}{\a_1+\be_1-(\a_2+\be_2)}, \\
   \phABt &= \frac{\a_1\be_1(\a_2-\be_2)}{\a_1\be_1(\a_2-\be_2)-\a_2\be_2(\a_1-\be_1)}, \\
   \phAB &= \frac{\a_1\be_1(\a_2+\be_2)}{\a_1\be_1(\a_2+\be_2)-\a_2\be_2(\a_1+\be_1)}, \\
   \phMo &= \frac{\a_1(\be_2+\rh)(\a_1+\be_1+\rh)}{\a_1(\be_2+\rh)(\a_1+\be_1+\rh)-\a_2(\be_1+\rh)(\a_2+\be_2+\rh)}, \label{eq:phMo}\\
   \phMot &= \frac{\be_1(\a_2+\rh)(\a_1+\be_1+\rh)}{\be_1(\a_2+\rh)(\a_1+\be_1+\rh)-\be_2(\a_1+\rh)(\a_2+\be_2+\rh)}, \label{eq:phMot}\\
   \phMf &= \frac{\be_1(\a_2-\rh)(\a_1+\be_1-\rh)}{\be_1(\a_2-\rh)(\a_1+\be_1-\rh)-\be_2(\a_1-\rh)(\a_2+\be_2-\rh)}, \label{eq:phMf}\\
   \phMft &= \frac{\a_1(\be_2-\rh)(\a_1+\be_1-\rh)}{\a_1(\be_2-\rh)(\a_1+\be_1-\rh)-\a_2(\be_1-\rh)(\a_2+\be_2-\rh)}, \label{eq:phMft}\\
   \phAFnull &= \frac{\a_1(\a_1+\be_1)(2\a_2+\be_2)}{\a_1(\a_1+\be_1)(2\a_2+\be_2)-\a_2(\a_2+\be_2)(2\a_1+\be_1)}, \\
   \phBFnull &= \frac{\be_1(\a_1+\be_1)(\a_2+2\be_2)}{\be_1(\a_1+\be_1)(\a_2+2\be_2)-\be_2(\a_2+\be_2)(\a_1+2\be_1)},
\end{align}
\end{subequations}
and
\begin{subequations}\label{eq:mAmB}
\begin{align}
    \mA &=\frac{\a_1\a_2}{\a_1-\phi(\a_1-\a_2)}, \label{mA}\\
    \mB &=\frac{\be_1\be_2}{\be_1-\phi(\be_1-\be_2)}, \label{mB}\\
%  \mFinf &=\min\{\mA,\mB\},\\
%  \minf &=\max\{\mA,\mB\},\\
    \mFnull &=\frac{(\a_1+\be_1)(\a_2+\be_2)}{\a_1+\be_1-\phi(\a_1+\be_1-\a_2-\be_2)},\\
    \mMo &= \frac{-(\a_1+\be_1+\rh)(\a_2+\be_2+\rh)}{\a_1+\be_1+\rh-\ph(\a_1+\be_1-\a_2-\be_2)}, \label{mMo}\\
    \mMf &= \frac{(\a_1+\be_1-\rh)(\a_2+\be_2-\rh)}{\a_1+\be_1-\rh-\ph(\a_1+\be_1-\a_2-\be_2)}, \label{mMf}\\ 
    \mast &=\frac{1}{\th}\sqrt{\frac{-\a_1\a_2\be_1\be_2(\a_1+\be_1)(\a_2+\be_2)}{\ph(1-\ph)}}.
\end{align}
\end{subequations}
We set $\mA=\infty$, $\mB=\infty$, and $\mFnull=\infty$ if $\ph=\phA$, $\ph=\phB$, and $\ph=\phFnull$, respectively. Similarly, we set $\mast=\infty$  if $\th=0$, $\ph=0$, or $\ph=1$.

The quantities $\mA$, $\mB$, and $\mFnull$ yield the bounds for the intervals of total migration rates $m$ in which the SLPs at $\A$, $\B$, and the polymorphic equilibrium $\Fnull$, respectively, are admissible:
\begin{subequations}\label{sita_bounds}
\begin{equation}
	-1 < \si_1+\si_2 < 1 \;\iff\; -1 < \frac{m}{\mA} < 1 \label{si12_mA_bound},
\end{equation}
\begin{equation}
	-1 < \ta_1+\ta_2 < 1 \;\iff\; -1 < \frac{m}{\mB} < 1 \label{ta12_mB_bound},
\end{equation}
\begin{equation}
	-1 < \ka_1+\ka_2 <1  \;\iff\; -1 < \frac{m}{\mFnull} < 1 \label{ka12_mFnull_bound}.
\end{equation}
\end{subequations}
Here, the left and the right inequalities correspond, and we have
\begin{subequations}\label{mAB_PhAB}
\begin{equation}\label{eq:mA_positiv}
	\mA>0 \;\iff\; \ph>\phA,
\end{equation}
\begin{equation}\label{eq:mB_positiv}
	\mB>0 \;\iff\; \ph>\phB,
\end{equation}
\begin{equation}\label{eq:mFnull_positiv}
	\mFnull>0 \;\iff\; \ph>\phFnull.
\end{equation}
\end{subequations}
From \eqref{eq:parameter}, we obtain 
\begin{equation}\label{eq:m_admiss}
	\mA\ne0, \quad \mB\ne0, \quad \mFnull\ne0, \quad \mast>0,
\end{equation}
\begin{equation}
	\a_2\le\mA\le\a_1, \quad \be_2\le\mB\le\be_1, \quad \a_2+\be_2\le\mFnull\le\a_1+\be_1.
\end{equation}

The quantities $\mMo$ and $\mMf$ occur in the stability conditions of the monomorphic equilibria $\M_1$ and $\M_4$ (Proposition \ref{prop_stab_ME_mphi}), and $\mast$ determines the range of stability of $\Fnull$; see \eqref{mmax_r0}. They satisfy
\begin{equation}
	-(\a_1+\be_1+\rh)\le\mMo\le-(\a_2+\be_2+\rh), \quad \a_2+\be_2-\rh\le\mMf\le\a_1+\be_1-\rh.
\end{equation}

We note that $\mA$, $\mB$, $\mFnull$, and $\mMf$ assume their minima if $\ph=0$ and their maxima if $\ph=1$, whereas $\mMo$ assumes its minimum or maximum at $\ph=1$ or $\ph=0$, respectively. $\mast$ is a convex function of $\ph$, and symmetric around its minimum $\ph=1/2$.

The definitions of (several of) the quantities $\ph^\X$ are motivated by the following relations:
\begin{subequations}\label{eq:m_equivalences}
\begin{align}
		\mA = \mB &\;\iff\;  \ph=\phABt \text{ and } \mA<0, \label{eq:mA_mB_2}\\
		\mA = -\mB &\;\iff\; \ph=\phAB \text{ and } \mA>0, \label{eq:mA_mB}\\
		\mMo = -\mA &\;\iff\;  \ph=\phMo \text{ and } \mA<0, \label{eq:ph_equal_phMo}\\
		\mMo = -\mB &\;\iff\;  \ph=\phMot \text{ and } \mB<0, \\
		\mMf = \mA &\;\iff\;  \ph=\phMft \text{ and } \mA>0, \\
		\mMf = \mB &\;\iff\;  \ph=\phMf \text{ and } \mB>0, \label{eq:ph_equal_phMf}\\
		\mFnull = -\mB &\;\iff\; \ph=\phBFnull \text{ and } \mB<0, \\
		-\mFnull = \mA &\;\iff\; \ph=\phAFnull \text{ and } \mA>0,
\end{align}
\end{subequations} 
where we have
\begin{equation} \label{eq:mFnull_mMf_mMo}
		\mFnull=\mMf= -\mMo \;\iff\; \rho=0.
\end{equation}
The following relations apply to $\mast$:
\begin{subequations}\label{eq:mast_equivalences}
\begin{align}
		\mA = -\mB = \mast &\;\iff\; \ph=\phAB, \label{eq:mast_mA_mB}\\
		\mA = \mFnull = \mast &\;\iff\; \ph=\phMf, \label{eq:mF_mA_mast}\\
		-\mB = -\mFnull =\mast &\;\iff\; \ph=\phMo,\label{eq:mF_mB_mast}
\end{align}
\end{subequations}
where we derived \eqref{eq:mF_mA_mast} and \eqref{eq:mF_mB_mast} from \eqref{eq:ph_equal_phMo} and \eqref{eq:ph_equal_phMf} using \eqref{eq:mFnull_mMf_mMo}. 

In the following, we summarize the most important inequalities between the quantities $\ph^\X$: 
\begin{equation}\label{phA<phB}
	0<\phA < \phAFnull < \phAB < \phBFnull < \phB  <1 ,
\end{equation}
\begin{equation}\label{phABt<phA}
	\be_2 < \a_2    \; \iff\; 0 < \phABt < \phA,
\end{equation}
\begin{equation}\label{phFnull_ph_A_B_Fnull}
	\phAFnull < \phFnull < \phBFnull.
\end{equation}
They can be derived straightforwardly from their definitions and our general assumption \eqref{eq:parameter}.
Finally, if $\rh=0$ and $	\tht=\a_1\a_2-\be_1\be_2$, the following relations hold:
\begin{subequations}\label{eq:ph_rel}
\begin{align}
    0<\phMo<\phA<\phAFnull<\phFnull\leq \phAB<\phBFnull<\phB<\phMf<1 & \;\iff\; \tht\leq 0,\label{eq:ph_rel_tht_leq_0}\\
    0<\phMo<\phA<\phAFnull<\phAB < \phFnull <\phBFnull<\phB<\phMf<1 & \;\iff\; \tht> 0, \label{eq:ph_rel_tht>0}
\end{align}
and
\begin{equation}
    \phABt<\phMo \;\text{ if }\;\be_2<\a_2.
\end{equation}
\end{subequations} 

Additional relations that are needed only in the proofs may be found in Appendix \ref{sec:important_quantities_cont}.

\subsection{Admissibility of SLPs}\label{sec:admissibility_SLP_m}
We begin by expressing the conditions for admissibility of the SLPs in terms of the total migration rate $m$ and
the migration ratio $\ph$. Since, by \eqref{si}, \eqref{ta}, and \eqref{sita_bounds}, every SLP is admissible if $m$ is sufficiently small and leaves the state space at a uniquely defined critical migration rate, it is sufficient to determine this critical rate and the monomorphism through which it leaves the state space. Using \eqref{mA}, \eqref{mB}, \eqref{mAB_PhAB}, and \eqref{sita_bounds}, we infer from \eqref{eq:ASLPbf} that
\begin{subequations}\label{SLPs_leave_state_space_m}
\begin{alignat}{2}
   &\ph<\phA \text{ and } m\uparrow -\mA 	 &\;\iff\; \PP_{\A,1} \to \M_1 \text{ and }   \PP_{\A,2} \to \M_2,\\
   &\ph>\phA \text{ and } m\uparrow \mA    &\;\iff\; \PP_{\A,1} \to \M_3 \text{ and }   \PP_{\A,2} \to \M_4, \\
   &\ph<\phB \text{ and } m\uparrow -\mB	 &\;\iff\; \PP_{\B,1} \to \M_1 \text{ and }   \PP_{\B,2} \to \M_3,\\
   &\ph>\phB \text{ and } m\uparrow \mB    &\;\iff\; \PP_{\B,1} \to \M_2 \text{ and }   \PP_{\B,2} \to \M_4.
\end{alignat}
\end{subequations}Ž
In particular, no SLP is admissible if
\begin{equation}\label{noSLP_admiss}
	 m>\max\{|\mA|,|\mB|\}.
\end{equation}

We observe that locus $\A$ is polymorphic and locus $\B$ is monomorphic if and only if 
\begin{equation}\label{eq:mB<m<mA}
 |\mB|<m<|\mA|.
\end{equation}
If $\be_2<\a_2$, we infer from \eqref{0<-mB<-mA} and  \eqref{0<-mB<mA} that \eqref{eq:mB<m<mA} holds if and only if
\begin{equation} \label{eq:phABt_ph_phAB}
	\phABt < \ph < \phAB.
\end{equation}
Therefore, \eqref{be2<a2} implies that if locus $\A$ is under weaker selection than locus $\B$ in both demes ($|\be_k|>|\a_k|$), then there is a range of values $\ph$  and $m$ such that locus $\A$ is polymorphic whereas $\B$ is monomorphic. This is in contrast to the CI model or highly asymmetric migration rates or $\th=0$, where it is always the locus under weaker selection that first loses its polymorphism while $m$ increases. This is a pure one-locus result and a consequence of the classical condition for a protected polymorphism, e.g., \eqref{si}. With two-way migration, a locus with alleles of small and similar (absolute) effects in the demes ($\a_1\approx-\a_2$) may be maintained polymorphic for higher migration rates than a locus with alleles of large and very different (absolute) effects.

\subsection{Stability of monomorphic equilibria}\label{sec:stab_mono_m}
Here, we reformulate the stability conditions of the ME derived in Section \ref{sec:stab_mono} in terms of $m$ and $\ph$. 

\begin{proposition}\label{prop_stab_ME_mphi}
$\M_1$ is asymptotically stable if
\begin{equation}\label{M1stab_mphi}
    \ph<\phA \text{ and } m > \max\{-\mA,-\mB , \mMo\}.
\end{equation}

$\M_2$ is always unstable.

$\M_3$ is asymptotically stable if
\begin{equation}\label{M3_stab_m_phi}
	\phA<\ph<\phB \text{ and } m>\max\{\mA,-\mB\}.
\end{equation}

$\M_4$ is asymptotically stable if
\begin{equation}\label{M4stab_mphi}
    \ph>\phB \text{ and } m > \max\{\mA,\mB ,\mMf\}.
\end{equation}

If in these conditions one inequality is reversed, the corresponding equilibrium is unstable.
\end{proposition}

\begin{proof}
We prove only that the statement about $\M_1$ is equivalent to that in Proposition \ref{prop_stab_ME}. The others follow analogously or are immediate.

From Proposition \ref{prop_stab_ME} and \eqref{mA}, \eqref{mB}, \eqref{mMo}, and \eqref{sita_bounds}, we infer immediately that 
$\M_1$ is asymptotically stable if and only if 
\begin{equation}\label{eq:M1_stability_1}
      1<\frac{m}{-\mA} \text{ and } 1<\frac{m}{-\mB}
\end{equation}
and
\begin{equation}\label{eq:M1_stability_2}
      m > \mMo.
\end{equation}
The possible inequalities between $-\mA$, $-\mB$, and $\mMo$ are given in \eqref{eq:mMo>mAmB} and \eqref{eq:mMo<mAmB}. 
By \eqref{eq:mA_positiv}, \eqref{eq:mB_positiv}, and \eqref{phA<phB}, it follows that \eqref{eq:M1_stability_1} is feasible if and only if $\ph<\phA$. Thus if $\ph\geq \phA$, $\M_1$ is unstable. Therefore, \eqref{eq:M1_stability_1} and \eqref{eq:M1_stability_2} are equivalent to \eqref{M1stab_mphi}.
\end{proof}

\begin{remark}\label{remark_mono}\rm
(i) We have $\max\{-\mA,-\mB ,\mMo\}=\mMo$ in \eqref{M1stab_mphi}
if and only if
\begin{subequations}\label{eq:M1_mMo_max}
\begin{align}
		&\ph<\min\{\phMo,\phMot\} \text{ and } \rh < -\a_2 \text{ and } \be_2<\a_2, \label{eq:M1_mMo_max1} \text{ or } \\
		&\ph<\phMo \text{ and } \rh < -\be_2 \text{ and } \be_2\ge\a_2. 
%    \ph<\min\{\phMo, \phMot\} \text{ and } \rh<\min\{-\a_2,-\be_2\}.
\end{align}
\end{subequations}

(ii) We have  $\max\{\mA,\mB,\mMf\}=\mMf$ in \eqref{M4stab_mphi} if and only if
\begin{equation}\label{eq:M4_mMf_max}
    \ph>\phMf \text{ and } \rh<\a_1.
\end{equation}

(iii) An internal equilibrium in LD can leave or enter the state space through $\M_1$ or $\M_4$ only if 
$m=\mMo$ or $m=\mMf$, respectively. If \eqref{eq:M1_mMo_max} or \eqref{eq:M4_mMf_max} holds, then $\M_1$ or $\M_4$, respectively,
become asymptotically stable by the bifurcation.
\end{remark}

\begin{proof}[Proof of Remark \ref{remark_mono}]
If $\be_2\ge\a_2$, statement (i) is an immediate consequence of \eqref{eq:-mB<-mA<mMo} and \eqref{eq:-mA<-mB<mMo} because $\ph<\phA$
implies $0<\min\{-\mA,-\mB\}$. If $\be_2<\a_2$, then \eqref{eq:-mB<-mA<mMo}, \eqref{eq:-mA<-mB<mMo}, and \eqref{eq:-mA=-mB<mMo} show that $\max\{-\mA,-\mB , \mMo\}=\mMo$ if (a) $\rh<\rhMo$ \eqref{eq:ME_critical_rho} and $\phABt<\ph<\phMo$ or (b) $\rh<\rhMo$ and $\ph\le\phABt$ or (c) $\rhMo<\rh<-\a_2$ and $\ph<\phMot$, where $\rhMo<-\a_2$ by \eqref{rhMo<-a2}. Invoking \eqref{eq:phM_interval_2}, we can combine conditions (a), (b), and (c) to obtain \eqref{eq:M1_mMo_max1}.

Statement (ii) follows directly from \eqref{eq:infeasible_m} and \eqref{eq:0<mA<mB<mMf}. 

Statement (iii) follows by observing that only internal equilibria in LD will depend on $\rh$, the factor $t_3$ \eqref{M1_charpol} in the characteristic polynomial at $\M_1$ is the only one that depends on $\rh$, and $t_3$ gives rise to an eigenvalue zero if and only if $m=\mMo$. An analogous argument holds for $\M_4$.
\end{proof}

The asymmetry between \eqref{eq:M1_mMo_max} and \eqref{eq:M4_mMf_max} results from the fact that $\a_1<\be_1$ is assumed, whereas $\be_2<\a_2$ or $\be_2\ge\a_2$ is possible.
The reader may recall the comments made below Proposition \ref{prop_stab_ME}. In addition, we note that if the fitness parameters and $\rh$ and $\ph$ are fixed, a stable ME remains stable if $m$ is increased. This is not necessarily so if $m$ and $\ph$ are varied simultaneously. For related phenomena in the one-locus case, see Karlin (1982) and Nagylaki (2012).
In Section \ref{sec:strong_mig}, we will prove global convergence to one of the asymptotically stable ME if $m$ is sufficiently large.

\subsection{Weak migration}
We recall from Proposition \eqref{prop_weak_mig} that for sufficiently weak migration, there is a fully polymorphic equilibrium, it is globally asymptotically stable, and exhibits positive LD in both demes.

\subsection{Strong migration}\label{sec:strong_mig}
\begin{proposition}\label{prop_strong_mig}
For sufficiently large $m$, one of the monomorphic equilibria $\M_1$, $\M_3$, or $\M_4$ is globally attracting.
This equilibrium is $\M_1$, $\M_3$, or $\M_4$ if $\ph<\phA$, $\phA<\ph<\phB$, or $\phB<\ph$, respectively.
\end{proposition}

\begin{proof}
The proof is based on the perturbation results about the strong-migration limit in Section 4.2 of B\"urger (2009a).
The strong-migration limit is obtained if $\max_{k=1,2}\{|\a_k|,|\be_k|,\rho \}/m \allowbreak\to 0$. In this limit, the demes become homogeneous and the system of differential equations \eqref{eq:dynamics} converges to a system, where in each deme
\begin{subequations}\label{eq:dynamics_strong_mig}
\begin{align}
    \dot{p} &=\a p(1-p)+\be D ,\\
    \dot{q} &=\be q(1-q)+\a D ,\\
    \dot{D} &=\left[\a(1-2p)+\be(1-2q)-\rh\right]D
\end{align}
\end{subequations}
holds with $p_1=p_2=p$, $q_1=q_2=q$, $D_1=D_2=D$. Here,
\begin{equation}\label{averaged fitnesses}
        \a = (1-\ph)\a_1+\ph\a_2 \;\text{ and }\; \be = (1-\ph)\be_1+\ph\be_2
\end{equation}
are the spatially averaged selection coefficients and averaging is performed with respect to the Perron-Frobenius eigenvector $(1-\ph,\ph)$ of the migration matrix (see Section 4.2 in B\"urger (2009a) for a much more general treatment starting with a multilocus model in discrete time). Therefore, Proposition 4.10 in B\"urger (2009a) applies and, provided $m$ is sufficiently large, all trajectories of \eqref{eq:dynamics} converge to a manifold on which the allele frequencies and the linkage disequilibria in both demes are nearly identical. In addition, in the neighborhood of each hyperbolic equilibrium of \eqref{eq:dynamics_strong_mig} there is exactly one equilibrium of \eqref{eq:dynamics}, and it has the same stability.

In the present case, the conclusion of Proposition 4.10 in B\"urger (2009a) can be considerably strengthened.
Because the system \eqref{eq:dynamics_strong_mig} describes evolution in an ordinary two-locus model under genic selection, the ME representing the gamete of highest fitness is globally asymptotically stable. In fact, \eqref{eq:dynamics_strong_mig} is also a generalized gradient system for which Lemma 2.2 of Nagylaki et al.\ (1999) holds. Therefore, the analog of statement (c) in Theorem 4.3 of B\"urger (2009a) applies and yields global convergence to the unique stable equilibrium.

Finally, it is an easy exercise to show that, in the strong-migration limit, i.e., with fitnesses averaged according to \eqref{averaged fitnesses}, gamete $A_1B_1$, $A_2B_1$, or $A_2B_2$ has highest fitness if
$\ph<\phA$, $\phA<\ph<\phB$, or $\ph>\phB$, respectively. Since there is no dominance, the corresponding ME is the unique
stable equilibrium.
\end{proof}

\subsection{Linkage equilibrium}
We shall establish all possible equilibrium configurations and their dependence on the parameters under LE. In Figure 2, the equilibrium configurations are displayed as schematic bifurcation diagrams with the total migration rate $m$ as the bifurcation parameter. In Theorem \ref{LE_theorem}, we assign to each diagram its pertinent parameter combinations.

In order to have only one bifurcation diagram covering cases that can be obtained from each other by simple symmetry considerations but are structurally equivalent otherwise, we use the sub- and superscripts $\X$ and $\Y$ in the labels of Figure 2. For an efficient presentation of the results, we define
\begin{equation} \label{eq:label_1a}
     \PP_\X =\PBo,\; \PP_\Y =\PAo,\; \M_\textsf{i}=\M_1,\;  m^\X=-\mB,\; m^\Y=-\mA, \tag{L1}
\end{equation}
\begin{equation} \label{eq:label_2a}
     \PP_\X =\PAo,\; \PP_\Y =\PBo,\; \M_\textsf{i}=\M_1,\; m^\X=-\mA,\; m^\Y=-\mB, \tag{L2}
\end{equation}
\begin{equation} \label{eq:label_3a}
     \PP_\X =\PAo,\; \PP_\Y =\PBt,\; \M_\textsf{i}=\M_3,\; m^\X=\mA,\; m^\Y=-\mB, \tag{L3}
\end{equation}
\begin{equation} \label{eq:label_4a}
     \PP_\X =\PBt,\; \PP_\Y =\PAo,\; \M_\textsf{i}=\M_3,\; m^\X=-\mB,\; m^\Y=\mA, \tag{L4}
\end{equation}
\begin{equation} \label{eq:label_5a}
    \PP_\X =\PBt,\; \PP_\Y =\PAt,\; \M_\textsf{i}=\M_4,\; m^\X=\mB,\; m^\Y=\mA. \tag{L5}
\end{equation}

\begin{figure}
 \centering
 \includegraphics[scale=0.8,keepaspectratio=true]{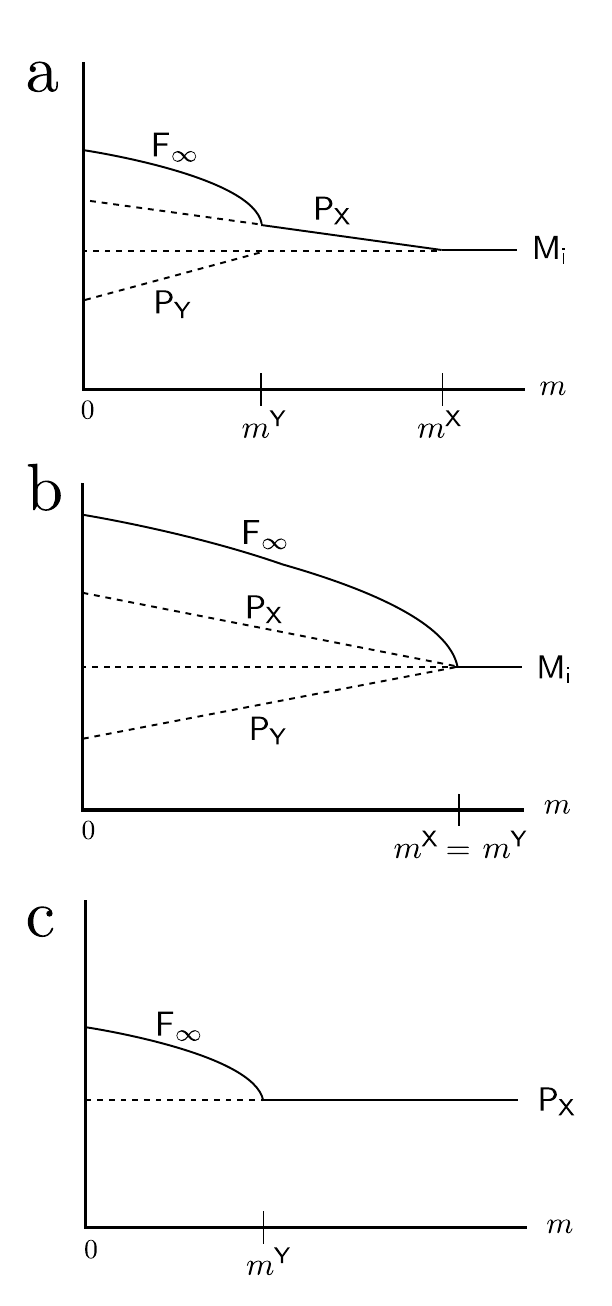}
 \caption{\small{
Bifurcation diagrams for LE. Diagrams (a)-(c) display the equilibrium configurations listed in Theorem \ref{LE_theorem}. Each line indicates one admissible equilibrium as a function of the total migration rate $m$. Only equilibria are shown that can be stable or are involved in a bifurcation with an equilibrium that can be stable. Lines are drawn such that intersections occur if and only if the corresponding equilibria collide. Solid lines represent asymptotically stable equilibria, dashed lines unstable equilibria. The meaning of the superscripts $\X$ and $\Y$ is given in \eqref{eq:label_1a} -- \eqref{eq:label_5a}}.}
 \label{fig:bf_LE}
\end{figure}

\begin{figure}
 \centering
 \includegraphics[scale=1,keepaspectratio=true,clip=true,trim=200pt 640pt 150pt 130pt]{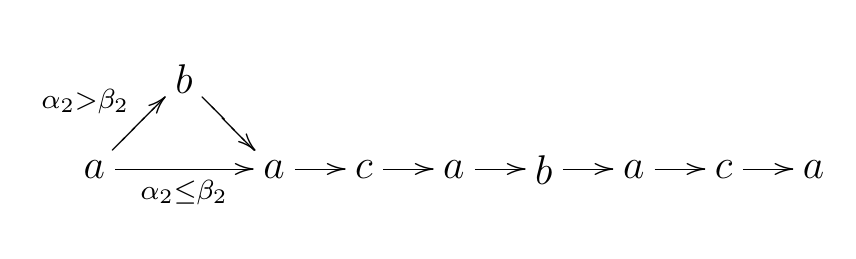}
 % CTTIM_BF_order_raw.pdf: 792x612 pixel, 72dpi, 27.94x21.59 cm, bb=0 0 792 612
 \caption{\small{Order in which the bifurcation diagrams of Figure 2 occur as $\phi$ increases from 0 to 1.}}
 \label{fig:bf_order_LE}
\end{figure}

\begin{theorem}\label{LE_theorem}
Assume LE, i.e., \eqref{dynLE}. Figure 2 shows all possible bifurcation diagrams 
that involve bifurcations with equilibria that can be stable for some $m$ given the other parameters. 

\noindent
A. Diagram (a) in Figure 2 occurs generically. It occurs if and only if one of the following cases applies:
\begin{equation}\label{eq:bf_LE_1}
    \ph<\phABt \text{ and } \be_2<\a_2 \text{ and } \eqref{eq:label_1a}
\end{equation}
or
\begin{equation}\label{eq:bf_LE_2}
    \phABt<\ph<\phA \text{ and } \be_2<\a_2 \text{ and } \eqref{eq:label_2a}
\end{equation}
or
\begin{equation}\label{eq:bf_LE_3}
   \ph<\phA \text{ and } \a_2\leq \be_2 \text{ and } \eqref{eq:label_2a}
\end{equation}
or
\begin{equation}\label{eq:bf_LE_4}
   \phA<\ph<\phAB \text{ and } \eqref{eq:label_3a}
\end{equation}
or
\begin{equation}\label{eq:bf_LE_5}
   \phAB<\ph<\phB \text{ and } \eqref{eq:label_4a}
\end{equation}
or
\begin{equation}\label{eq:bf_LE_6}
   \phB<\ph \text{ and } \eqref{eq:label_5a}.
\end{equation}

\noindent
B. The following two diagrams occur only if the parameters satisfy particular relations.

\noindent
Diagram (b) in Figure 2 applies if one of the following two cases holds:
\begin{equation}\label{eq:bf_LE_7}
    \ph=\phABt \text{ and } \be_2<\a_2 \text{ and } \eqref{eq:label_1a}
\end{equation}
or
\begin{equation}\label{eq:bf_LE_8}
    \ph=\phAB \text{ and } \eqref{eq:label_4a}.
\end{equation}

\noindent
Diagram (c) in Figure 2 applies if one of the following two cases holds:
\begin{equation}\label{eq:bf_LE_9}
    \ph=\phA \text{ and } \PP_\X=\PAo \text{ and } m^\Y=-\mB
\end{equation}
or 
\begin{equation}\label{eq:bf_LE_10}
    \ph=\phB \text{ and } \PP_\X=\PBt \text{ and } m^\Y=\mA.
\end{equation}

\noindent
C. Figure 3 shows the order in which the bifurcation diagrams of Figure 2 arise if $\ph$ is increased from 0 to 1. 
\end{theorem}

\begin{proof}
We prove parts A and B simultaneously, essentially by rewriting the conditions in Proposition \ref{prop_LE} on
admissibility and stability of the equilibria in terms of $m$, $\mA$, and $\mB$ \eqref{eq:mAmB}.

From \eqref{eq:Einfinitybf} and \eqref{sita_bounds}, we infer easily:
\begin{subequations}\label{eq:Finf_m_mA-MB}
\begin{align}
     \Finf \to \PAo & \;\iff\; m \uparrow -\mB \text{ and } 0<-\mB<|\mA| , \label{eq:Finf_m_a}\\
     \Finf \to \PAt  & \;\;\text{does not occur}, \label{eq:Finf_m_b}\\
     \Finf \to \PBo  & \;\iff\; m \uparrow -\mA \text{ and } 0<-\mA<-\mB , \label{eq:Finf_m_c}\\
     \Finf \to \PBt  & \;\iff\; m \uparrow \mA \text{ and } 0<\mA<|\mB| , \label{eq:Finf_m_d}\\
     \Finf \to \M_1  & \;\iff\; m \uparrow -\mA=-\mB \; \text{and } 0<-\mA=-\mB , \label{eq:Finf_m_e}\\
     \Finf \to \M_3  & \;\iff\; m \uparrow \mA=-\mB \;  \text{and } 0<\mA=-\mB , \label{eq:Finf_m_f}\\
     \Finf \to \M_2  & \;\;\text{ or } \Finf \to \M_4 \text{ do not occur.}
\end{align}
\end{subequations}
Invoking the relations \eqref{mAB_relations}, we can rewrite conditions \eqref{eq:Finf_m_a}, \eqref{eq:Finf_m_c}-\eqref{eq:Finf_m_f} in the form 
\begin{subequations}\label{eq:Finf_mPh}
\begin{alignat}{2}
     &\Finf \to \PAo  && \;\iff\; m \uparrow -\mB \text{ and either }  \nonumber \\
      &&&\qquad	\ph<\phAB \text{ if } \a_2\le\be_2 \;\text{ or }\; \phABt<\ph<\phAB \text{ if } \be_2<\a_2, \label{eq:Finf_mPh_a}\\
%     &\Finf \to \PAt  && \text{ does not occur}, \label{eq:Finf_mPh_b}\\
     &\Finf \to \PBo  && \;\iff\; m \uparrow -\mA \text{ and } \be_2<\a_2 \text{ and } \ph<\phABt, \label{eq:Finf_mPh_c}\\
     &\Finf \to \PBt  && \;\iff\; m \uparrow \mA \text{ and } \ph>\phAB, \label{eq:Finf_mPh_d} \\
     &\Finf \to \M_1  && \;\iff\; m \uparrow -\mA=-\mB \; \text{and } \be_2<\a_2 \text{ and } \ph=\phABt,\label{eq:Finf_mPh_e} \\
     &\Finf \to \M_3  && \;\iff\; m \uparrow \mA=-\mB \;  \text{and } \ph=\phAB. \label{eq:Finf_mPh_f} 
%\\  &\Finf \to \M_2  && \text{ or } \Finf \to \M_4 \text{ do not occur.} \label{eq:Finf_mPh_g}
 \end{alignat}
 \end{subequations}
We conclude immediately that \eqref{eq:Finf_mPh_c} applies in case \eqref{eq:bf_LE_1} (Part A), \eqref{eq:Finf_mPh_e} in case \eqref{eq:bf_LE_7} (Part B), and \eqref{eq:Finf_mPh_f} in case \eqref{eq:bf_LE_8} (Part B). From \eqref{phA<phB} and \eqref{phABt<phA} we conclude that \eqref{eq:Finf_mPh_a} applies in the following cases: \eqref{eq:bf_LE_2}-\eqref{eq:bf_LE_4} (Part A), or \eqref{eq:bf_LE_9} (Part B). Analogously we conclude that \eqref{eq:Finf_mPh_d} applies in the following cases: \eqref{eq:bf_LE_5}, \eqref{eq:bf_LE_6} (Part A), or \eqref{eq:bf_LE_10} (Part B).

From Proposition \ref{prop_LE} and \eqref{SLPs_leave_state_space_m} we obtain:
\begin{subequations}
\begin{align}
    \PAo  \text{ is globally asymptotically stable } &\;\iff\; -\mB<m<|\mA|,\\
    \PBo  \text{ is globally asymptotically stable } &\;\iff\; -\mA<m<-\mB,\\
    \PBt  \text{ is globally asymptotically stable } &\;\iff\; \mA<m<|\mB|.
\end{align}
\end{subequations}
As $m \to \max\{|\mA|,|\mB|\}$, the stable SLP leaves the state space according to \eqref{SLPs_leave_state_space_m}, which
gives precisely the cases corresponding to diagrams (a) and (c). If $\ph=\phA$ ($\mA=\infty$), $\PAo$ is always admissible, cf.\ \eqref{eq:bf_LE_9}. If $\ph=\phB$ ($\mB=\infty$), $\PBt$ is always admissible, cf.\ \eqref{eq:bf_LE_10}.
 
A ME is globally asymptotically stable  and only if 
\begin{equation}\label{mmaxMono_LE}
 	 m \ge \max\{|\mA|,|\mB|\}.
\end{equation}
By Proposition \ref{prop_stab_ME_mphi} and Remark \ref{remark_mono} this equilibrium is $\M_1$ if $\ph<\phA$ (cases \eqref{eq:bf_LE_1}-\eqref{eq:bf_LE_3}, \eqref{eq:bf_LE_7}), or $\M_3$ if $\phA<\ph<\phB$ (cases \eqref{eq:bf_LE_4}, \eqref{eq:bf_LE_5}, \eqref{eq:bf_LE_8}), or $\M_4$ if $\phB<\ph$ \eqref{eq:bf_LE_6}.
\end{proof}

The bifurcations of equilibria that cannot be stable can be derived easily from Sections \ref{sec:admissibility_SLP_m} and \ref{sec:stab_mono_m} and the above theorem by noting that these are boundary equilibria and corresponding pairs of SLPs are admissible for the same parameters; see \eqref{si} and \eqref{ta}. Inclusion of these bifurcations would require the introduction of subcases.

\begin{corollary}\label{corollary_mmaxinf}
Under the assumption of LE, the maximum migration rate, below which a stable two-locus polymorphism exists, is given by
\begin{equation}\label{minfty}
		\mmax^\infty = \min\{|\mA|,|\mB|\} .
\end{equation}
\end{corollary}
The corollary is a simple consequence of Proposition \ref{prop_LE} and \eqref{eq:Finf_mPh}.

\subsection{Strong recombination: quasi-linkage equilibrium}\label{QLE}
We recall from Section \ref{sec:QLE} that for sufficiently strong recombination, global convergence to the unique stable equilibrium occurs. From the coordinates \eqref{F_largerho} of the perturbed internal equilibrium, which is in quasi-linkage equilibrium, approximations could be derived for the critical migration rates at which the internal equilibrium collides with a boundary equilibrium and leaves the state space. It is not difficult to check with {\it Mathematica} that for large $\rh$, $\F$ collides with $\PBt$ if $m=\mmax^\rh(\PBt)+O(\rh^{-2})$, where
\begin{equation}\label{mmaxrh(PBt)}
	  \mmax^\rh(\PBt) = \mA 
	   - \frac{(\mA)^3}{\rh}\left[\frac{\be_1}{\a_1}\ph-\frac{\be_2}{\a_2}(1-\ph) \right] 
	  				\left[\frac{\ph}{\be_1}-\frac{1-\ph}{\be_2} - \sqrt{(\mA)^{-2}-\frac{4\ph(1-\ph)}{\be_1\be_2}} \right] .
\end{equation}
We note that $\mmax^\rh(\PBt)>0$ if and only if $\ph>\phAB$, as is expected from \eqref{eq:Finf_mPh_d}.
Closer examination of \eqref{mmaxrh(PBt)} reveals that both $\mmax^\rh(\PBt) > \mA$ and $\mmax^\rh(\PBt) < \mA$ may hold. 

Thus, the fully polymorphic equilibrium may be maintained for higher or lower migration rates than in the case of LE. This does not conform with the intuitive expectation that for reduced recombination, $\mmax^\rh(\PBt) > \mmax^\infty$ should hold because the locally adapted haplotypes ($A_kB_k$ in deme $k$) are less frequently broken apart. However, numerical evaluation of \eqref{mmaxrh(PBt)} shows that $\mmax^\rh(\PBt) < \mmax^\infty$ occurs only for about 3\% of the admissible parameter combinations and if it holds, $\mmax^\rh(\PBt)$ is only very slightly less than $\mmax^\infty$ (results not shown). If $\rh$ is about as large as the largest selection coefficient or smaller, $\mmax$ increases with decreasing $\rh$.
Expressions analogous to \eqref{mmaxrh(PBt)} can be obtained for collisions of $\F$ with the other equilibria.

\subsection{No recombination}\label{sec:bifs_norec}
Our aim is to establish all possible equilibrium configurations and their dependence on the parameters if recombination is absent. In Figure 4, the equilibrium configurations are displayed as schematic bifurcation diagrams with the total migration rate $m$ as the bifurcation parameter. In Theorem \ref{Mainresult_no_reco}, we assign to each diagram its pertinent parameter combinations.

\begin{figure}
 \centering
 \includegraphics[scale=0.8,keepaspectratio=true]{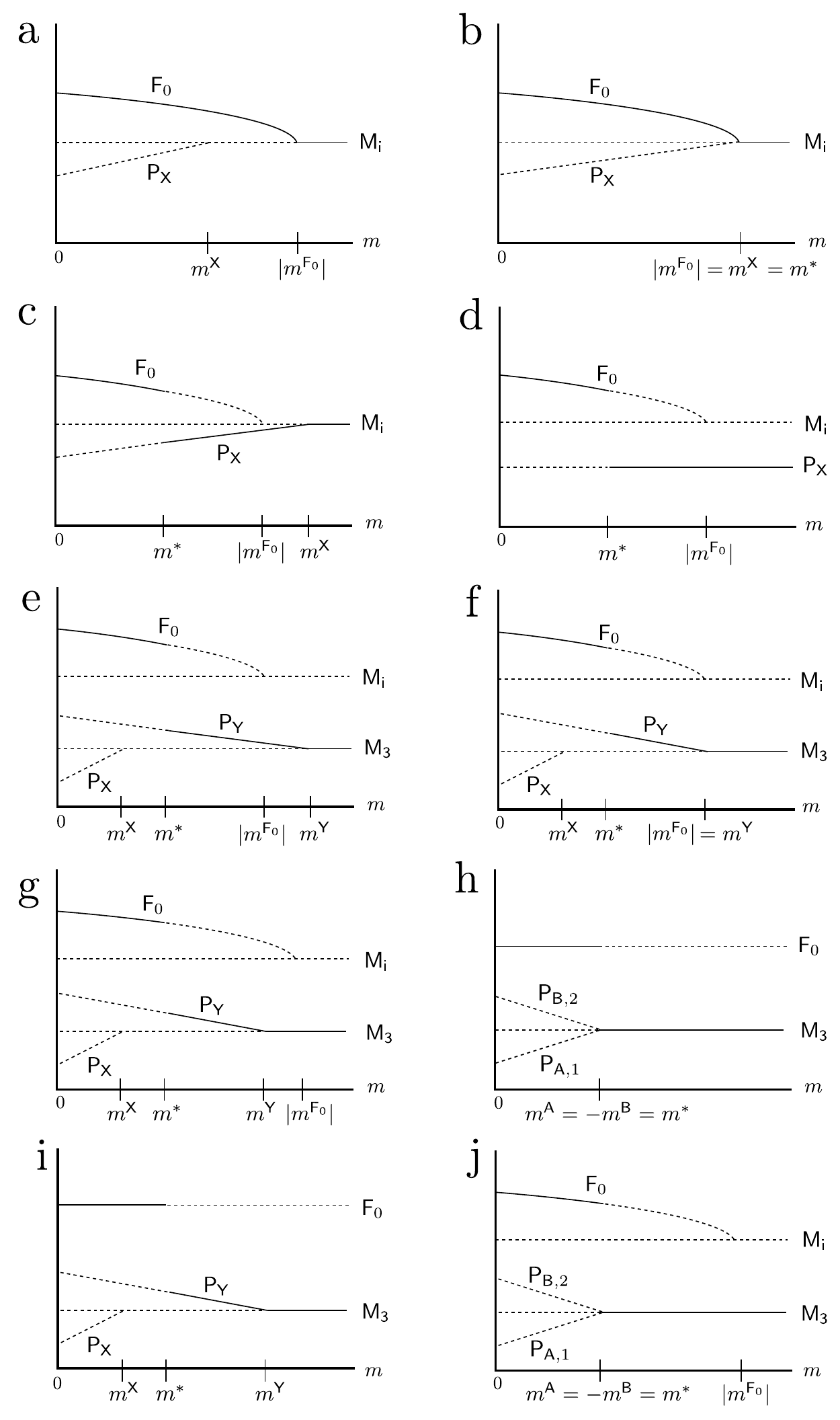}
 \caption{\small{
Bifurcation diagrams for $\rho=0$. Diagrams (a) -- (j) represent all equilibrium and stability configurations listed in Theorem \ref{Mainresult_no_reco}. Each diagram displays the possible equilibria as a function of the total migration rate $m$. Each line indicates one admissible equilibrium, drawn if and only if it is admissible. Only equilibria are shown that can be stable or are involved in a bifurcation with an equilibrium that can be stable. Lines are drawn such that intersections occur if and only if the corresponding equilibria collide. Solid lines represent asymptotically stable equilibria, dashed lines unstable equilibria.}}
 \label{fig:bf_no_reco}
\end{figure}

\begin{figure}
 \centering
 \includegraphics[scale=1,keepaspectratio=true,clip=true,trim=220pt 400pt 150pt 130pt]{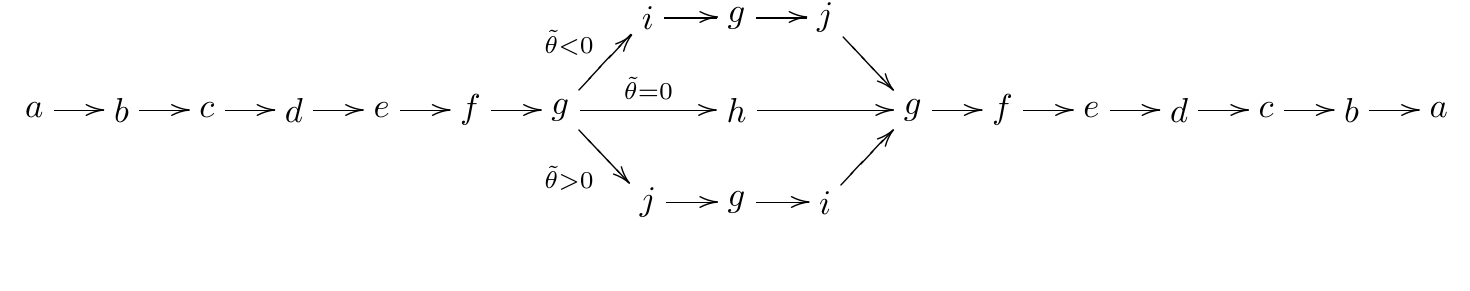}
 \caption{\small{Order in which the bifurcation diagrams of Figure 4 occur as $\phi$ increases from 0 to 1, where $\tilde{\theta}=\alpha_1\alpha_2-\beta_1\beta_2$.}}
 \label{fig:bf_order_no_reco}
\end{figure}

In order to have only one bifurcation diagram covering cases that can be obtained from each other by simple symmetry considerations but are structurally equivalent otherwise, we use the sub- and superscripts $\X$ and $\Y$ in the labels of Figure 4. For an efficient presentation of the results, we define
\begin{equation} \label{eq:label_1}
    \PP_\X =\PP_{\A,1} ,\; \PP_\Y = \PP_{\B,2},\; \M_\textsf{i} = \M_1,\; m^\X = |\mA|,\; m^\Y = |\mB|, \tag{R1}
\end{equation}
\begin{equation} \label{eq:label_2}
    \PP_\X =\PP_{\B,2} ,\; \PP_\Y = \PP_{\A,1},\; \M_\textsf{i} = \M_4,\; m^\X = |\mB|,\; m^\Y = |\mA|, \tag{R2}
\end{equation}
\begin{equation} \label{eq:label_3}
    \PP_\X =\PP_{\A,1} ,\; \PP_\Y = \PP_{\B,2},\; \M_{\textsf{i}}=\M_4, \; m^\X = |\mA|,\; m^\Y = |\mB|,  \tag{R3}
\end{equation}
\begin{equation} \label{eq:label_4}
     \PP_\X =\PP_{\B,2} ,\; \PP_\Y = \PP_{\A,1},\;\M_\textsf{i} = \M_1,\; m^\X = |\mB|,\; m^\Y = |\mA|, \tag{R4}
\end{equation}
\begin{equation} \label{eq:label_1'}
    \PP_\X =\PP_{\A,1} ,\; \M_\textsf{i} = \M_1,\; m^\X = |\mA|, \tag{R1'}
\end{equation}
\begin{equation} \label{eq:label_2'}
    \PP_\X =\PP_{\B,2} ,\; \M_\textsf{i} = \M_4,\; m^\X = |\mB|, \tag{R2'}
\end{equation}
\begin{equation} \label{eq:label_3'}
    \PP_\X =\PP_{\A,1} ,\; \PP_\Y = \PP_{\B,2},\; m^\X = |\mA|,\; m^\Y = |\mB|,  \tag{R3'}
\end{equation}
\begin{equation} \label{eq:label_4'}
     \PP_\X =\PP_{\B,2} ,\; \PP_\Y = \PP_{\A,1},\; m^\X = |\mB|,\; m^\Y = |\mA|. \tag{R4'}
\end{equation}

\begin{theorem} \label{Mainresult_no_reco}
Let $\rh=0$. Figure 4 shows all possible bifurcation diagrams 
that involve bifurcations with equilibria that can be stable for some $m$ given the other parameters.  

A. The following diagrams occur for an open set of parameters:

\begin{enumerate} 
\item
Diagram (a) in Figure 4 applies if one of the following two cases holds:
\begin{subequations}\label{eq:ph<phMo_ph>phMf}
\begin{equation}\label{eq:ph<phMo}
      0\leq \ph < \phMo  \text{ and }  \eqref{eq:label_1'}
\end{equation}
or
\begin{equation}\label{eq:phMf<ph}
      \phMf < \ph \leq 1  \text{ and }  \eqref{eq:label_2'}.
\end{equation}
\end{subequations}

\item
Diagram (c) in Figure 4 applies if one of the following two cases holds:
\begin{subequations}\label{eq:diagram_c}
\begin{equation}\label{eq:phMo<ph<phA}
    \phMo<\ph<\phA \text{ and } \eqref{eq:label_1'}
\end{equation}
or
\begin{equation}\label{eq:phB<ph<phMf}
    \phB<\ph<\phMf \text{ and } \eqref{eq:label_2'}.
\end{equation}
\end{subequations}

\item
Diagram (e) in Figure 4 applies if one of the following two cases holds:

\begin{subequations}\label{eq:diagram_e}
\begin{equation}\label{eq:phA<ph<phAFnull}
    \phA<\ph<\phAFnull \text{ and } \eqref{eq:label_4}
\end{equation}
or
\begin{equation}\label{eq:phBFnull<ph<phB}
    \phBFnull<\ph<\phB  \text{ and } \eqref{eq:label_3}.
\end{equation}
\end{subequations}

\item
Diagram (g) in Figure 4 applies if one of the following four cases holds:
\begin{subequations}\label{eq:diagram_g}
\begin{equation}\label{eq:phAFnull<ph<min_phFnull_phAB}
    \phAFnull<\ph<\min\{\phFnull,\phAB\} \text{ and } \eqref{eq:label_4}
\end{equation}
or
\begin{equation}\label{eq:phFnull<ph<phAB}
     \phFnull<\ph<\phAB \text{ and } \eqref{eq:label_2}
\end{equation}
or 
\begin{equation}\label{eq:phAB<ph<phFnull}
     \phAB<\ph<\phFnull \text{ and } \eqref{eq:label_1}
\end{equation}
or
\begin{equation}\label{eq:max_phFnull_phAB<ph<phBFnull}
     \max\{\phFnull,\phAB\}<\ph<\phBFnull \text{ and } \eqref{eq:label_3}.
\end{equation}
\end{subequations}

B. The following diagrams are degenerate, i.e., occur only if the parameters satisfy particular relations.

\item
Diagram (b) in Figure 4 applies if one of the following two cases holds:
\begin{subequations}\label{eq:ph=phMo_phMf}
\begin{equation}\label{eq:ph=phMo}
   \ph=\phMo \text{ and } \eqref{eq:label_1'}
\end{equation}
or
\begin{equation}\label{eq:ph=phMf}
   \ph=\phMf \text{ and } \eqref{eq:label_2'}.
\end{equation}
\end{subequations}

\item
Diagram (d) in Figure 4 applies if one of the following two cases holds:
\begin{subequations}\label{eq:ph=phA_phB}
\begin{equation}\label{eq:ph=phA}
   \ph=\phA \text{ and } \eqref{eq:label_1'}
\end{equation}
or
\begin{equation}\label{eq:ph=phB}
   \ph=\phB \text{ and } \eqref{eq:label_2'}.
\end{equation}
\end{subequations}

\item
Diagram (f) in Figure 4 applies if one of the following two cases holds:
\begin{subequations}\label{eq:ph=phAFnull_phBFnull}
\begin{equation}\label{eq:ph=phAFnull}
    \ph=\phAFnull \text{ and } \eqref{eq:label_2}
\end{equation}
or
\begin{equation}\label{eq:ph=phBFnull}
    \ph=\phBFnull \text{ and } \eqref{eq:label_1}.
\end{equation}
\end{subequations}

\item
Diagram (h) in Figure 4 applies if
\begin{equation}\label{eq:ph=phFnull=phAB}
    \ph=\phFnull=\phAB.
\end{equation}

\item
Diagram (i) in Figure 4 applies if one of the following two cases holds:
\begin{subequations}\label{eq:ph=phFnull}
\begin{equation}\label{eq:ph=phFnull>phAB}
    \ph=\phFnull > \phAB \text{ and } \eqref{eq:label_3'}.
\end{equation}
or
\begin{equation}\label{eq:ph=phFnull<phAB}
    \ph=\phFnull < \phAB \text{ and } \eqref{eq:label_4'}
\end{equation}
\end{subequations}

\item
Diagram (j) in Figure 4 applies if one of the following two cases holds:
\begin{subequations}\label{eq:ph=phAB}
\begin{equation}\label{eq:ph=phAB<phFnull}
    \ph=\phAB < \phFnull \text{ and } \M_\textit{i}=\M_1
\end{equation}
or
\begin{equation}\label{eq:ph=phAB>phFnull}
    \ph=\phAB > \phFnull \text{ and } \M_\textit{i}=\M_4.
\end{equation}
\end{subequations}
\end{enumerate}

C. Figure 5 shows the order in which the bifurcation diagrams of Figure 4 arise if $\ph$ is increased from 0 to 1. 
\end{theorem}

\begin{proof}
We prove parts A and B simultaneously and derive the statements about admissibility and stability of the equilibria by rewriting the conditions in Section \ref{sec:no_rec} in terms of $m$, $\mA$, $\mB$, $\mFnull$, and $\mast$ \eqref{eq:mAmB}. These critical migration rates satisfy the relations given in \eqref{eq:m_equivalences}, \eqref{eq:mast_equivalences}, \eqref{eq:mast_mFnull_relation} and \eqref{eq:crit_m_rel_r0}.

We start by treating the bifurcations and stability of $\Fnull$.
Using \eqref{eq:mFnull_positiv} and \eqref{ka12_mFnull_bound}, we infer from \eqref{eq:E0leavesstatespace}:
\begin{subequations}\label{eq:Fnull_leaves_state_space_m}
\begin{align}
    \Fnull \to \M_1 &\;\iff\; m\uparrow -\mFnull \;\text{ and }\; \ph<\phFnull, \label{eq:Fnull_leaves_state_space_m_case1}\\
    \Fnull \to \M_4 &\;\iff\; m\uparrow \mFnull \;\text{ and }\; \ph>\phFnull.  \label{eq:Fnull_leaves_state_space_m_case2}
\end{align}
\end{subequations}
From \eqref{eq:ph_rel} we conclude that \eqref{eq:Fnull_leaves_state_space_m_case1} applies precisely in the following cases:  \eqref{eq:ph<phMo}, \eqref{eq:phMo<ph<phA}, \eqref{eq:phA<ph<phAFnull}, \eqref{eq:phAFnull<ph<min_phFnull_phAB}, \eqref{eq:phAB<ph<phFnull}  (Part A), or \eqref{eq:ph=phMo}, \eqref{eq:ph=phA}, \eqref{eq:ph=phAFnull}, \eqref{eq:ph=phAB<phFnull} (Part B). Similarly, \eqref{eq:Fnull_leaves_state_space_m_case2} applies in precisely the following cases: \eqref{eq:phMf<ph}, \eqref{eq:phB<ph<phMf}, \eqref{eq:phBFnull<ph<phB}, \eqref{eq:phFnull<ph<phAB}, \eqref{eq:max_phFnull_phAB<ph<phBFnull} (Part A) or \eqref{eq:ph=phMf}, \eqref{eq:ph=phB}, \eqref{eq:ph=phBFnull}, \eqref{eq:ph=phAB>phFnull} (Part B). $\Fnull$ is admissible for every $m>0$ if and only if $\ph=\phFnull$, which corresponds to the remaining three cases \eqref{eq:ph=phFnull=phAB} and \eqref{eq:ph=phFnull<phAB}, \eqref{eq:ph=phFnull>phAB}.

Condition \eqref{critmu12}, which determines when $\Fnull$ changes stability, is equivalent to $m=\mast$. Therefore, Proposition \ref{prop_norec} and the definitions of $\mFnull$ and $\mast$ imply that $\Fnull$ is asymptotically stable if and only if either
\begin{subequations}\label{mmax_r0}
\begin{align}
 		&0 < m < |\mFnull| \le \mast  \label{eq:Fnull_bf_with_ME}\\
 		\intertext{or}
		&0 < m < \mast < |\mFnull|   \label{eq:Fnull_bf_with_SLP}
\end{align}
\end{subequations}
holds, where 
\begin{subequations}\label{eq:mmax_r0_rel_ph}
\begin{align}
	  0<|\mFnull|\le\mast & \;\iff\; \ph\le\phMo \;\text{ or }\; \ph\ge\phMf , \label{eq:Fnull_no_bf} \\
	  0<\mast<|\mFnull| & \;\iff\; \phMo<\ph<\phMf . \label{eq:Fnull_bf}
\end{align}
\end{subequations}

If \eqref{eq:Fnull_bf_with_ME} applies, according to \eqref{eq:Fnull_leaves_state_space_m}, $\Fnull$ leaves the state space at $m=-\mFnull$ or $m=\mFnull$ and exchanges stability with the respective monomorphism. By \eqref{eq:Fnull_no_bf}, this occurs in the cases \eqref{eq:ph<phMo_ph>phMf} or \eqref{eq:ph=phMo_phMf} of the theorem. 

If \eqref{eq:Fnull_bf_with_SLP} applies, $\Fnull$ loses stability at $m=\mast$ and, generically, either $\PP_{\A,1}$ or $\PP_{\B,2}$ is asymptotically stable if $m>\mast$ (see below). $\Fnull$ remains admissible up to $m=|\mFnull|$, when it collides with $\M_1$ or
$\M_4$. By \eqref{eq:ph_rel} and \eqref{eq:Fnull_bf}, this occurs in the cases \eqref{eq:diagram_c} -- \eqref{eq:diagram_g}, \eqref{eq:ph=phA_phB}, \eqref{eq:ph=phAFnull_phBFnull}, \eqref{eq:ph=phFnull}, and \eqref{eq:ph=phAB}. 

Finally, if $\ph=\phAB$ (cases \eqref{eq:ph=phFnull=phAB} and \eqref{eq:ph=phAB} in the theorem), $\M_3$ becomes stable. This follows from the statement below \eqref{stabF0_sigtau} together with \eqref{eq:mast_mA_mB}.

Next, we treat the bifurcations of the SLPs. The SLPs are admissible in intervals of the form $0<m<|\mA|$ or $0<m<|\mB|$ and leave the state space upon collision with a ME (Section \ref{sec:admissibility_SLP_m}).
From \eqref{PA1_stab_r0} we conclude by simple calculations that $\PP_{\A,1}$ is asymptotically stable if and only if
\begin{equation}
    \mast<m<|\mA| \;\text{ and }\; \phMo<\ph<\phAB,
\end{equation}
as is the case in \eqref{eq:phMo<ph<phA}, \eqref{eq:phA<ph<phAFnull}, \eqref{eq:phAFnull<ph<min_phFnull_phAB} (if $\min\{\phFnull, \phAB\}=\phAB$), and \eqref{eq:phFnull<ph<phAB}, as well as in \eqref{eq:ph=phA}, \eqref{eq:ph=phAFnull}, and \eqref{eq:ph=phFnull<phAB}. 

From \eqref{PB2_stab_r0}, we conclude that $\PP_{\B,2}$ is asymptotically stable if and only if
\begin{equation}
    \mast<m<|\mB| \;\text{ and }\; \phAB<\ph<\phMf,
\end{equation}
as is the case in \eqref{eq:phB<ph<phMf}, \eqref{eq:phBFnull<ph<phB}, \eqref{eq:phAB<ph<phFnull}, and \eqref{eq:max_phFnull_phAB<ph<phBFnull} (if $\max\{\phFnull, \phAB\}=\phAB$), as well as in \eqref{eq:ph=phB}, \eqref{eq:ph=phBFnull}, and \eqref{eq:ph=phFnull>phAB}. 

It remains to study the stability of the ME. For $\rh=0$, we infer from Section \ref{sec:important_quantities} and Proposition \ref{prop_stab_ME_mphi}:
\begin{subequations}
\begin{align}
    \M_1 \;\text{ is asymptotically stable} &\;\iff\; \begin{cases}
							  m>-\mFnull \;\text{ and }\; \ph<\phMo, \text{ or}\\
							  m>-\mA  \;\text{ and }\; \phMo\leq \ph<\phA,
                                                         \end{cases} \label{eq:M1_r0}\\
    \M_3 \;\text{ is asymptotically stable} &\;\iff\; m>\max\{|\mA|,|\mB|\} \;\text{ and }\;\phA<\ph<\phB, \label{eq:M3_r0}\\
    \M_4 \;\text{ is asymptotically stable} &\;\iff\; \begin{cases}
							  m>\mB \;\text{ and }\; \phB<\ph\leq \phMo, \text{ or}\\
							  m>\mFnull \;\text{ and }\; \phMf<\ph.
                                                         \end{cases}\label{eq:M4_r0}
\end{align}
\end{subequations}
In conjunction with the above results on $\Fnull$ and the SLPs, this shows that, except in the degenerate cases \eqref{eq:ph=phMo_phMf}, \eqref{eq:ph=phA_phB}, \eqref{eq:ph=phFnull=phAB}, and \eqref{eq:ph=phAB},
a ME becomes stable through a transcritical bifurcation with either $\Fnull$, $\PP_{\A,1}$, or $\PP_{\B,2}$.
In particular, $\M_1$ becomes asymptotically stable for large $m$ if \eqref{eq:ph<phMo}, \eqref{eq:phMo<ph<phA}, or \eqref{eq:ph=phMo} applies,  $\M_3$ becomes asymptotically stable if one of \eqref{eq:diagram_e}, \eqref{eq:diagram_g}, \eqref{eq:ph=phAFnull_phBFnull}, \eqref{eq:ph=phFnull=phAB}, \eqref{eq:ph=phFnull}, or \eqref{eq:ph=phAB} applies, and $\M_4$ becomes asymptotically stable if \eqref{eq:phMf<ph}, \eqref{eq:phB<ph<phMf}, or \eqref{eq:ph=phMf} applies. 
If $\ph=\phA$ or $\ph=\phB$ \eqref{eq:ph=phA_phB}, then $\PP_{\A,1}$ or $\PP_{\B,2}$, respectively, is admissible and asymptotically stable for every $m$, and every ME is unstable. This finishes the proof of parts A and B.

Part C of Theorem \ref{Mainresult_no_reco} follows immediately from parts A and B by applying the relations in \eqref{eq:ph_rel}.
\end{proof}

This theorem demonstrates that, for given selection parameters, the equilibrium structure, hence also the evolutionary dynamics, depends strongly on the degree $\ph$ of asymmetry of the migration rates. However, it is also important to note (and maybe counter intuitive) that for symmetric migration ($\ph=1/2$) any of the ten possible bifurcation diagrams may apply, simply by choosing the selection parameters accordingly. 

The bifurcations of equilibria that cannot be stable can be derived easily from Sections \ref{sec:admissibility_SLP_m}, \ref{sec:stab_mono_m}, \ref{sec:no_rec}, and the above theorem by noting that these are boundary equilibria and corresponding pairs of SLPs are admissible for the same parameters. Inclusion of these bifurcations would require the introduction of subcases. In particular, $\PAt$, $\PBo$, and $\M_2$ are always unstable because gamete $A_1B_2$ is eventually lost. We observe from \eqref{mmaxinf<mmast} and \eqref{mmaxinf<mFnull} that at most one pair of SLPs can be admissible if $\Fnull$ is either unstable or not admissible. If this is the case, then one of these SLPs is asymptotically stable (Figure 4).

\begin{corollary}\label{corollary_mmax0}
If $\rh=0$, the maximum migration rate, below which a stable two-locus polymorphism exists, is given by
\begin{equation}\label{m0}
		\mmax^0 = \min\{|\mFnull|,\mast\} .
\end{equation}
\end{corollary}

The corollary follows from the arguments surrounding \eqref{mmax_r0} and \eqref{eq:mmax_r0_rel_ph}.

\subsection{Weak recombination}\label{sec:bifs_weakrec}
If $m\ne\mast$, a regular perturbation analysis of $\Fnull$ yields the coordinates of a fully polymorphic (internal) equilibrium to leading order in $\rh$. This equilibrium, $\F$, is asymptotically stable (Karlin and McGregor 1972). We denote the first-order approximation of $\F$ by $\Frho$. Therefore, we have 
$\F=\Frho+o(\rh)$ and $\Frho=\Fnull+O(\rh)$ as $\rh\to0$. Because the coordinates of $\Frho$ are much too complicated to be informative, we refrain from presenting them. 

For sufficiently small $\rh$, the following properties of $\Frho$ (hence, of $\F$) can be inferred from Proposition \ref{prop_stab_ME_mphi}, Remark \ref{remark_mono}, and Theorem \ref{Mainresult_no_reco}, Part A.1:
\begin{subequations}\label{eq:Frho_bf}
\begin{align}
    \Frho \to \M_1 	& \;\iff\; m\uparrow \mMo \text{ and } \ph \leq \phMo, \\
    \Frho \to \M_4	& \;\iff\; m\uparrow \mMf \text{ and } \phMf \leq \ph.
\end{align}
\end{subequations}

The above perturbation analysis can not be used to investigate the properties of the internal equilibrium $\F$ for given small positive $\rh$ when $m$ is varied in the proximity of $\mast$. Therefore, we performed numerical calculations to study the fate of $\F$ when $\rh$ is small and fixed, and $m$ increases. It suggests the following:
\begin{subequations}\label{eq:Frho_bf2}
\begin{align}
    \F \to \PP_{\A,1} & \;\iff\; m\uparrow \mast_\A \text{ and } \phMo<\ph<\phAB, \\
    \F \to \M_3	& \;\iff\; m\uparrow \mast=\mA=-\mB \text{ and } \ph=\phAB, \\
    \F \to \PP_{\B,2} & \;\iff\; m\uparrow \mast_\B \text{ and } \phAB<\ph<\phMf,
\end{align}
\end{subequations}
where $\mast_\A$ and $\mast_\B$ are close to $\mast$. Thus, if $\rh$ is small, $\F$ stays close to $\Fnull$ as $m$ increases from 0 until a value close to $\mast$ is reached. Then, within a very short interval of $m$, $\F$ moves `quickly' along the manifold given by \eqref{manif_p2q2} and \eqref{manif} to one of the boundary equilibria ($\PAo$, $\PBt$, or $\M_3$) on the `opposite' side of the state space, where it exchanges stability upon collision with the respective equilibrium (at $\mast_\A$, $\mast_\B$, or $\mast$). $\F$ appears to be asymptotically stable whenever it is
admissible.
 
If one of the cases in \eqref{eq:Frho_bf} applies, then $\Fnull$ can be maintained for higher migration rates than $\F$ because $\mMo$ and $\mMf$ are decreasing functions in $\rh$. Numerical investigations support the conjecture that $\Fnull$ can be maintained for higher migration rates than $\F$ whenever recombination is weak but positive. Thus, when recombination is weak, decreasing $\rh$ increases the maximum migration rate below which a stable, fully polymorphic equilibrium can be maintained.

\subsection{Highly asymmetric migration}\label{sec:highly_asym}
As already discussed in Section \ref{Section_asym}, by introducing weak back migration (i.e., $\ph$ close to 0 or 1) to the CI model,  every equilibrium in the CI model gives rise to a unique equilibrium in a small neighborhood. This (perturbed) equilibrium has the same stability as the unperturbed. For weak or strong recombination, we can strengthen this conclusion. Because the CI model with $\rh=0$ is a generalized gradient system (B\"urger and Akerman 2011, Section 3.4.4) and the LE dynamics \eqref{dynLE} has a globally asymptotically stable equilibrium (Theorem \ref{LE_theorem}), the proof of Theorem 2.3 of Nagylaki et al.\ (1999) applies and shows that in both cases the global dynamics remains qualitatively unchanged under small perturbations. In particular, no new equilibria or limit sets are generated by a small perturbation. 

Therefore, if $\rh$ is sufficiently small and $\ph$ is sufficiently close to 0 or 1, we infer from Section \ref{Section_asym} and Theorem 2 in B\"urger and Akerman (2011) that the following bifurcation pattern applies (where $i=1$ or 4):
\begin{itemize} 
\item If $0<m<m^{\M_i}$, a unique internal equilibrium,  $\F$, exists. It is globally asymptotically stable.
\item At $m=m^{\M_i}$, $\F$ leaves the state space through the ME $\M_i$ by an exchange-of-stability bifurcation. 
\item If $m>m^{\M_i}$, $\M_i$ is globally asymptotically stable.
\end{itemize}
This pattern is displayed in diagram (a) of Figure 4, where $\Fnull$ needs to be substituted by $\F$. 
We conjecture that it applies whenever $\rh$ is sufficiently small and either $\ph<\phMo$ or $\ph>\phMf$ holds.
The bounds $\phMo$ and $\phMf$ follow from Remark \ref{remark_mono} because $\phMot$ is not needed if $\rh$ is sufficiently small; see \eqref{eq:-mA<-mB<mMo}. However, the upper bounds for $\rh$ given in Remark \ref{remark_mono} are, in general, too large to guarantee the above bifurcation pattern. This is known from the CI model in which the monomorphic equilibrium ($\M_i$) may be simultaneously stable with the internal equilibrium $\F$ because an unstable internal equilibrium enters the state space at $m=m^{\M_i}$ through $\M_i$. If $\ph=1$, this may occur if $\tfrac13(\a_1+\be_1) < \rh < 3\a_1-\be_1$, cf.\ \eqref{mmax_CIb}.
For $\ph\ne0$ or $\ph\ne1$, we have not been able to determine the upper bound for $\rh$ below which $\F$ indeed leaves the state space through $\M_i$.

Now we treat large $\rh$. Proposition \ref{prop_stab_ME_mphi} and Remark \ref{remark_mono} show that if
$\rh>\max\{-\a_2,-\be_2\}$, then $\M_1$ is asymptotically stable if and only if $\ph<\phA$ and $m>\max\{-\mA,-\mB\}$, and if $\rh>\a_1$, then $\M_4$ is asymptotically stable if and only if $\ph>\phB$ and $m>\max\{\mA,\mB\}$. 

If, in addition to $\rh$ being sufficiently large, $\ph$ is small or large, then Theorem \ref{LE_theorem} implies that the internal equilibrium ($\F$) leaves the state space through $\PAo$, $\PBo$, or $\PBt$. The respective conditions are small perturbations of those
given in \eqref{eq:Finf_mPh_a}, \eqref{eq:Finf_mPh_c}, of \eqref{eq:Finf_mPh_d}, respectively. Combining theses conditions with those for the stability of the ME and observing \eqref{phA<phB} and \eqref{phABt<phA}, we conclude that the following bifurcation 
pattern applies if one of the conditions (a) $\a_2\le\be_2$ and $\ph<\phA$, or (b) $\be_2<\a_2$ and $\ph<\phABt$, or (c) $\ph>\phB$ holds approximately:
\begin{itemize} 
\item If $0<m<m^\bullet$, a unique internal equilibrium, $\F$, exists. It is asymptotically stable.
\item At $m=m^\bullet$, $\F$ leaves the state space through a SLP by an exchange-of-stability bifurcation. 
\item If $m^\bullet<m<m^{\bullet\bullet}$, this SLP is asymptotically stable.
\item If $m\ge m^{\bullet\bullet}$, then a ME is asymptotically stable.
\end{itemize}
If (a) holds, then $m^\bullet\approx-\mB$ and the SLP and the ME are $\PAo$ and $\M_1$, respectively; if (b) holds, then $m^\bullet\approx-\mA$ and the SLP and the ME are $\PBo$ and $\M_1$, respectively;
if (c) holds, then $m^\bullet\approx\mA$ and the SLP and the ME are $\PBt$ and $\M_4$, respectively. Finally, $m^{\bullet\bullet}=\max\{-\mA,-\mB\}$ in (a) and (b), and $m^{\bullet\bullet}=\max\{\mA,\mB\}$ in (c).

\subsection{Maintenance of polymorphism}\label{sec:general_case}
As already noted in Section \ref{Section_general}, for general parameters the equilibrium configurations could not be determined analytically. To explore the potential of spatially heterogeneous selection in maintaining genetic variation in the presence of gene flow, we investigate the maximum total migration rate, $\mmax$, that admits a stable, fully polymorphic equilibrium. We have already shown that $\mmax=\mmax^\infty$ holds in the LE approximation (Corollary \ref{corollary_mmaxinf}), and $\mmax=\mmax^0$ holds if $\rh=0$ (Corollary \ref{corollary_mmax0}).
%Because $(\si_1+\si_2)^2<1$ and $(\ta_1+\ta_2)^2<1$ imply $(\ka_1+\ka_2)^2<1$, application of \eqref{sita_bounds} shows that
%$m<|\mA|$ and $m<|\mB|$ imply $m<|\mFnull|$, cf.\ \eqref{eq:crit_m_rel_r0}. Therefore, 
From \eqref{mmaxinf<mmast} and \eqref{mmaxinf<mFnull} we conclude that
\begin{equation}\label{minf<m0}
	\mmax^\infty \le \mmax^0,
\end{equation}
where, as is not difficult to show, equality holds if and only if $\ph=\phAB$.

For the CI model with $\ph=1$, Proposition 1 in B\"urger and Akerman (2011) yields
\begin{subequations}\label{mmax_CI} 
\begin{numcases}
 {\mmax=}
			\a_1 + \be_1 - \rh 		&if  $0< \rh \le \min\{\a_1,\tfrac13(\a_1+\be_1)\}$, \label{mmax_CIa}\\
			\dfrac{(\a_1+\be_1+\rh)^2}{8\rh} &if  $\tfrac13(\a_1+\be_1) < \rh \le 3\a_1-\be_1$, \label{mmax_CIb}\\
			\a_1\left(1+\dfrac{\be_1-\a_1}{\rh}\right) &if  $\max\{\a_1,3\a_1-\be_1\} < \rh$.  \label{mmax_CIc}
\end{numcases}
\end{subequations}
In this case, the fully polymorphic equilibrium is globally asymptotically stable if \eqref{mmax_CIa} or \eqref{mmax_CIc} apply, but only locally stable if \eqref{mmax_CIb} and $m$ is close to $\mmax$. A formula analogous to \eqref{mmax_CI}, but with $-\a_2$ and $-\be_2$ instead of $\a_1$ and $\be_1$, holds if $\ph=0$. 

In general, we have no explicit formula for $\mmax$. However, extensive numerical work, as well as \eqref{mmax_CI} and the considerations in Section \ref{sec:bifs_weakrec} suggest that 
\begin{equation}\label{mmax_bounds}
		\mmax \le \mmax^0
\end{equation}
holds always. This is illustrated by Figure 6, which displays the dependence of $\mmax$ on the migration ratio $\ph$ (Figures 6a and 6c) and on the recombination rate $\rh$ (Figures 6b and 6d) for two selection regimes. In Figures 6a and 6b, locus $\B$ is under stronger selection in both demes. In Figures 6c and 6d, each locus is under stronger selection in one deme. 

\begin{figure}[t]
 \centering
 \includegraphics[scale=0.55,keepaspectratio=true]{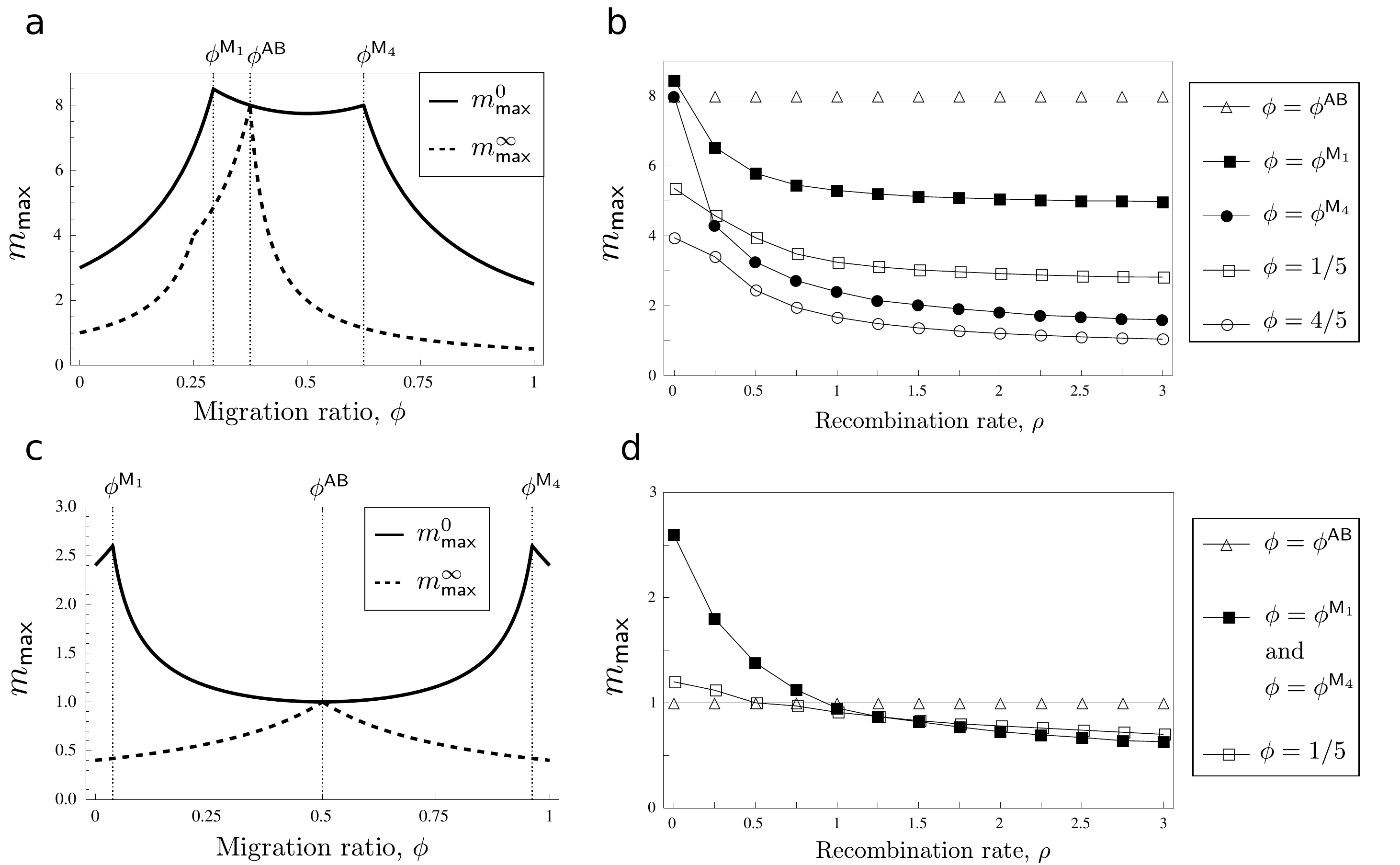}
 \caption{\small{The maximum amount of gene flow, $\mmax$, admitting an asymptotically stable two-locus polymorphism as a function of $\ph$ or $\rh$. In panels a and b, locus $\B$ is under stronger selection than locus $\A$ in both demes ($\alpha_2=-2\alpha_1=-1$, $\beta_1=-\beta_2=2$, $\theta=1$). In c and d, different loci are under stronger selection in the two demes ($\alpha_1=-\beta_2=0.4$, $\beta_1=-\alpha_2=2$, $\theta=3.84$). Panels a and c show $\mmax$ as a function of $\phi$ for complete linkage ($\mmax^0$, \eqref{m0}) and under linkage equilibrium ($\mmax^\infty$, \eqref{minfty}). Panels b and d display $\mmax$ for the indicated values of $\phi$ as a function of $\rho$. Here, $\mmax$ is obtained by determining numerically the critical migration rate when the stable internal equilibrium hits the boundary. This is done by computing when the leading eigenvalue at the boundary equilibrium is zero and by calculating the coordinates of the fully polymorphic equilibrium in a small neighborhood. 
In a and b, we have $\phABt=\tfrac14$ (indicated by the kink in the dashed line in a), $\phMo=\tfrac{5}{17}$, $\phA =\tfrac13$,
$\phAB=\tfrac38$, $\phB=\tfrac12$, $\phMf=\tfrac58$. In c and d, we have $\phMo=\tfrac{1}{26}$, $\phA=\tfrac16$, $\phAB=\tfrac12$, $\phB=\tfrac56$, $\phMf=\tfrac{25}{26}$. }}
 \label{fig:max_gene_flow}
\end{figure}
In Figures 6a and 6c, $\mmax^\infty$ and $\mmax^0$ are shown as functions of $\ph$. The inequality \eqref{minf<m0}
is a conspicuous feature in both cases. Also the shapes of $\mmax^\infty$ and $\mmax^0$ are conspicuous. The following properties are easy to prove: $\mmax^\infty$ is not differentiable at $\ph=\phAB$ and $\ph=\phABt$, and $\mmax^0$ is not differentiable at $\ph=\phMo$ and $\ph=\phMf$. $\mmax^\infty$ and $\mmax^0$ are piecewise convex functions in $\ph$.
If $\ph<\phMo$, $\mmax^0$ increases in $\ph$; if $\phMf<\ph$, $\mmax^0$ decreases in $\ph$; if $\phMo<\ph<\phMf$, $\mmax^0$ assumes its minimum at $\tfrac12$ provided $\phMo\le\tfrac12\le\phMf$. Therefore, $\mmax^0$ attains its maximum at $\phMo$ or $\phMf$. $\mmax^\infty$ increases if $\ph<\phAB$ and decreases if $\ph>\phAB$. It assumes its maximum at $\phAB$. 

Notably, $\mmax^\infty=\mmax^0$ holds if $\ph=\phAB$. Numerical work suggests that indeed $\mmax=\mmax^\infty=\mmax^0$ holds independently of $\rh$ if $\ph=\phAB$. If $\th=0$, then $\phA=\phB=\phAB=\phFnull$ and $\mmax^\infty=\mmax^0$ if $\ph=\phAB$. The latter condition is equivalent to \eqref{eq:th0_order_3}. Therefore, the analysis in Section \ref{sec:super_symm} applies and shows that an internal equilibrium, which presumably is globally asymptotically stable, exists always.

Figures 6b and 6d illustrate the effect of recombination on $\mmax$ for different values of $\ph$. In all cases investigated,
$\mmax$ decreased monotonically with increasing $\rh$. These findings support the conjecture that \eqref{mmax_bounds} is always valid. Therefore, $\mmax^0-\mmax^\infty$ serves as a useful estimate for the sensitivity of $\mmax$ to variation in $\rh$.  
We can prove that $\mmax^0-\mmax^\infty$ is maximized at $\ph=\phMo$, or at $\ph=\phMf$, or at $\ph=0$ if $\be_2<\a_2$.

Although we proved that $\mmax<\mmax^\infty$ can occur (Section \ref{QLE}), all numerical examples showed that $\mmax$ is only very slightly smaller than $\mmax^\infty$ in this case (results not shown). 
Therefore, our results suggest that the cases of LE (infinitely strong recombination) and of no recombination `essentially' bracket the range of parameters for which both loci can be maintained polymorphic. 

As Figures 6a and 6c show, the range of values $\ph$ for which the equilibrium structure can be expected to be similar to the CI model, i.e., $\ph<\min\{\phMo,\phABt\}$ or $\ph>\phMf$ (Section \ref{sec:highly_asym}), can vary considerably. 

Finally, we infer from Proposition \ref{prop_stab_ME_mphi} that none of the ME is stable if $m<\max\{|\mA|,|\mB|\}$. Hence, in this case at least one locus is maintained polymorphic. By contrast, we have shown in Section \ref{sec:admissibility_SLP_m} that no SLP is admissible if $m>\max\{|\mA|,|\mB|\}$. However, as demonstrated by our results for $\rh=0$, an internal equilibrium may be asymptotically stable if $\max\{|\mA|,|\mB|\}<m<\mmax^0$. These results suggest that no genetic variability can be maintained if
\begin{equation}
	m > \max\{|\mA|,|\mB|,\mmax^0\}.
\end{equation}
This bound is best possible if $\rh=0$. For sufficiently large $\rh$, the corresponding bound is $\max\{|\mA|,|\mB|\}$.

%Sections 5+6, CTTIM

\section{Migration load and local adaptation}\label{sec:load}
Here, we briefly investigate some properties of the migration load of the subpopulations and of the total population. We use these migration loads as simple measures for local adaptation (but see Blanquart et al.\ 2012). 
Mean fitness in deme $k$ is given by $\bar{w}_k=\a_k(2p_k-1)+\be_k(2q_k-1)$, with its maximum at $\a_k+\be_k$. Therefore, the migration loads in demes 1 and 2, defined as the deviation of $\bar w_k$ from its maximum, are given by
\begin{equation}
	L_1 = 2(\a_1(1-p_1)+\be_1(1-q_1)) \quad\text{and}\quad L_2 = 2(-\a_2p_2-\be_2q_2).
\end{equation}
Assuming that the subpopulations are of equal size, we define the load of the total population by $L=\tfrac12(L_1+L_2)$.

If migration is weak, we can calculate the migration load in each deme at the fully polymorphic equilibrium $\F$ (Proposition \ref{prop_weak_mig}) to leading order in $m_1$ and $m_2$. For deme 1, we obtain
\begin{equation}\label{eq:meanfit_weakmig}
  L_1 \approx 2m_1\frac{\a_1+\be_1+2\rh}{\a_1+\be_1+\rh},
\end{equation}
and an analogous formula holds for deme 2. Obviously, the migration load increases with increasing migration rates $m_1$ or $m_2$, hence with $m$, in each of the demes and in the total population. Simple calculations show that each of the loads also increases with increasing recombination rate $\rho$ if migration is weak.

In general, however, the load in each deme does not always increase with increasing $m$. The reason is that for sufficiently strong migration, generically, first one locus, then one of the haplotypes becomes fixed (Proposition \ref{prop_strong_mig}). If this is either $A_1B_1$ or $A_2B_2$, then the load in the corresponding deme will vanish for high migration rates, whereas that in the other deme will be very high. In such a case, the load of the total population may also decrease with increasing $m$. This occurs for large migration rates (not far below $\mmax$) and it can occur for completely linked loci as well as for loci in LE. In the CI model, the load always increases with the migration rate (B\"urger and Akerman 2011)

Finally, although $L$ is increasing in $\rh$ if migration is weak, this is not necessarily so if migration is strong. By using a grid of parameter combinations, we showed numerically that in about 0.34\% of more than $10^6$ combinations of $\a_1,\a_2,\be_1,\be_2,m$, and $\ph$, the total load $L$ at the equilibrium $\Finf$ is lower than that at $\Fnull$ (results not shown). Again, this occurs for high migration rates, not far below the value $\mmax^\infty$ at which $\Finf$ leaves the state space. Then a population maintained fully polymorphic by tight linkage may have a higher total load than a population in which fixation of a locus or a haplotype is facilitated by high recombination. In all such cases, selection in one deme was (considerably) stronger than in the other, and in more than 70$\%$ of the cases, a specialist haplotype became fixed at very high migration rates. In summary, under a wide range of conditions in this model, reduced recombination is favored, but there are instances where increased 
recombination is favored (cf.\ Pylkov et al.\ 1998; Lenormand and Otto 2000).

\section{$\Fst$ and differentiation}
The most commonly used measure for quantifying differentiation in spatially structured populations is $\Fst$. For diallelic loci, $\Fst$ can be defined as $\Fst=\frac{\text{Var}(p)}{\bar{p}(1-\bar{p})}$,
where $\text{Var}(p)$ is the variance of the allele frequencies in the total population and $\bar{p}$ is the allele frequency averaged over the demes. 
Estimators of multilocus $\Fst$ are usually defined as weighted averages of one-locus $\Fst$ estimators (e.g., Weir and Cockerham 1984, Leviyang and Hamilton, 2011).
Here, we extend Nagylaki's (1998) approach and define a genuine multilocus version of $\Fst$ that measures the covariance of the frequencies of (multilocus) haplotypes. We restrict attention to the diallelic two-locus case, but the extension to multiple multiallelic loci is evident. A general multilocus theory of fixation indices will be developed elsewhere.

Let $c_k$ denote the proportion of the population in deme $k$, so that $\sum_k c_k=1$. Then the frequency of haplotype $i$ in the entire population is $\bar x_i=\sum_k c_k x_{i,k}$. Because our subpopulations are randomly mating, the frequency of genotype $ij$ in the entire population is given by $\overline{x_ix_j}=\sum_k c_k x_{i,k}x_{j,k}$.
Following eqs.\ (6a) and (6b) in Nagylaki (1998), we define $\Fstij$ as a standardized measure of the covariance of the frequencies of haplotypes $i$ and $j$:
\begin{subequations}
\begin{align}
	\overline{x_i^2}  &= \bar x_i^2 + \Fstii \bar x_i(1-\bar x_i), \\
	\overline{x_ix_j} &= (1-\Fstij)\bar x_i \bar x_j.	
\end{align}
\end{subequations}
The multilocus, or haplotype, heterozygosity in the entire population can be defined as
\begin{equation}
	\bar h_S = \sum_{i,j:i\ne j}\overline{x_ix_j} = \sum_i(\bar x_i-\overline{x_i^2}),
\end{equation}
where $\sum_i$ runs over all haplotypes.
If the entire population were panmictic, its multilocus heterozygosity would be
\begin{equation}
	h_T = \sum_i \bar x_i(1-\bar x_i).
\end{equation}
Thus, $1-h_T$ is the probability that two gametes chosen at random from the entire population are the same haplotype.

Following eq.\ (32) in Nagylaki (1998), we define $\Fst$ by
\begin{equation}
	\Fst = \frac{1}{h_T}\sum_i \bar x_i(1-\bar x_i) \Fstii.
\end{equation}
Then $\Fst$ can be written as
\begin{equation}\label{FST}
	\Fst = 1-\frac{\bar h_S}{h_T} = \frac{\sum_i \text{Var}(x_i)}{\sum_i \bar x_i(1-\bar x_i)},
\end{equation}
in direct generalization of the classical formula given above.

We focus on the dependence of the equilibrium value of $\Fst$ on the migration parameters $m$ and $\ph$ and on the recombination rate $\rh$. Because we obtained the coordinates of the stable, fully polymorphic equilibrium equilibrium $\F$ explicitly only in special or limiting cases, explicit formulas for $\Fst$ can be derived only in these cases. For instance, if migration is weak, we obtain from \eqref{weak_mig} that, to leading order in $m$, 
\begin{equation}\label{eq:FST_weak_mig}
    \Fst = 1 - m\left[\frac{\ph}{c_2}\,\frac{\a_1\be_1+(\a_1+\be_1)\rh}{\a_1\be_1(\a_1+\be_1+\rh)} - 
    		\frac{1-\ph}{c_1} \frac{\a_2\be_2-(\a_2+\be_2)\rh}{\a_2\be_2(\a_2+\be_2-\rh)}\right].
\end{equation}
Here, $\Fst$ increases with decreasing $\rh$, and decreases with increasing $m$. Thus, stronger linkage leads to increased differentiation if migration is weak. 

\begin{figure}[t]
 \centering
 \includegraphics[scale=0.60,keepaspectratio=true]{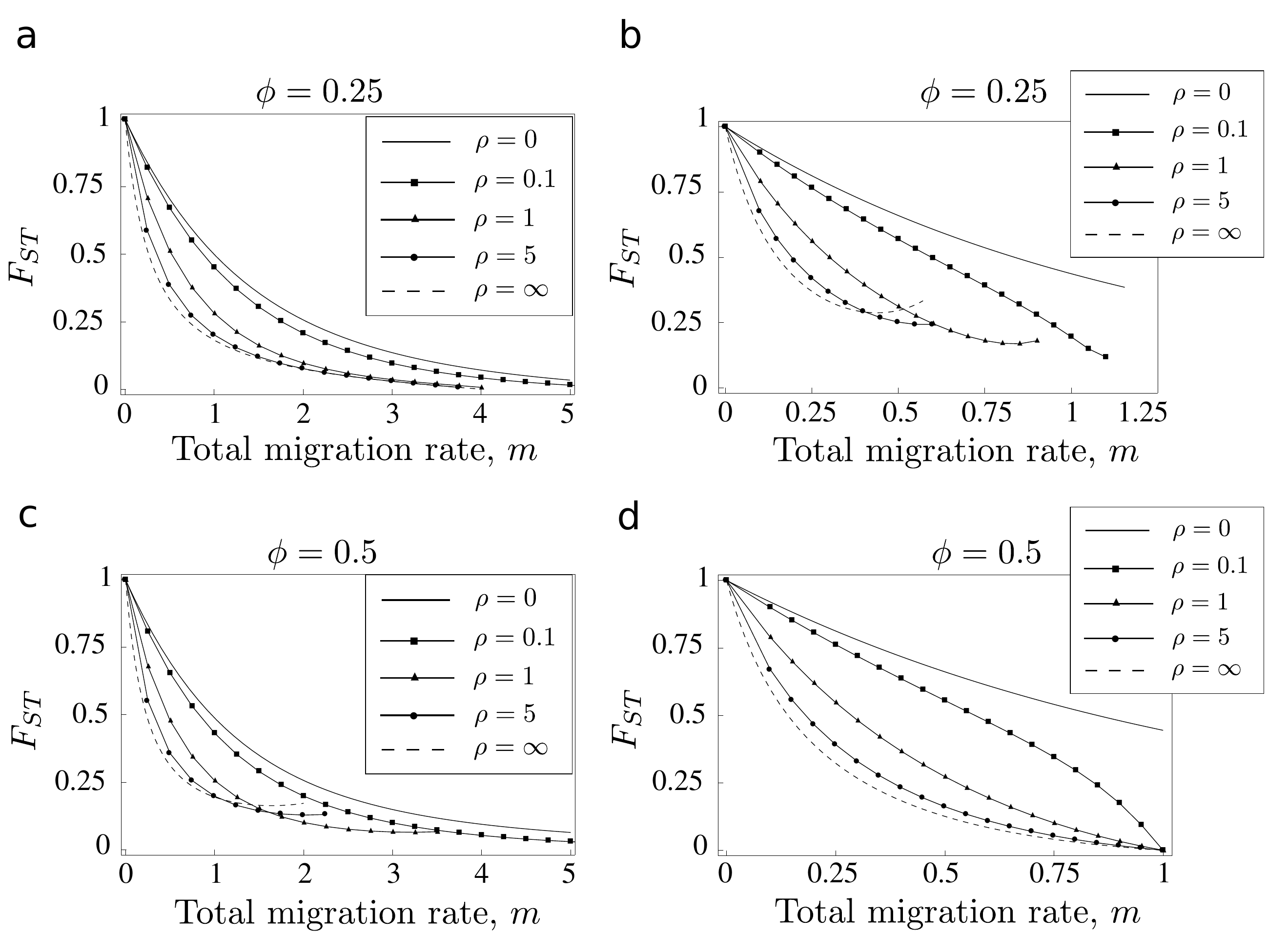}
 \caption{\small{$\Fst$ as a function of the total migration rate $m$. In panels a and c, locus $\B$ is under stronger selection in both demes ($\alpha_1=\tfrac12$, $\alpha_2=-1$, $\beta_1=-\beta_2=2$, $\theta=1$). In panels b and d, locus $\A$ is under stronger selection than $\B$ in deme 2, and locus $\B$ is under stronger selection than $\A$ in deme 1 ($\alpha_1=-\beta_2=0.4$, $\beta_1=-\alpha_2=2$, $\theta=3.84$). Note that in all cases, $\Fst$ is also monotone decreasing in $\rh$. For $\rh=0$ and $\rh=\infty$ (LE), the lines are from numerical evaluation of \eqref{FST} by substitution of the coordinates of $\Fnull$ \eqref{coord_Fnull} or $\Finf$ \eqref{Finf}. For other values of $\rh$, the numerically determined coordinates of the internal equilibrium are used.}}
 \label{fig:FST1}
\end{figure}

Figure 7 illustrates for two selection scenarios how $\Fst$, evaluated at the stable, fully polymorphic equilibrium $\F$, depends on the total migration rate $m$ and the recombination rate $\rho$. In diagrams (a) and (c) of Figure 7, it is assumed that locus $\B$ is under stronger selection than locus $\A$ in both demes.
It shows that $\Fst$ usually declines with increasing migration rate. However, there are a few instances, where $\Fst$ increases if $m$ is slightly below the migration rate at which the fully polymorphic equilibrium loses admissibility. 
In diagrams (a) and (c) of Figure 7, differentiation between the populations experiences the fastest decline for weak migration (relative to the selection parameters), whereas this is not necessarily so in diagrams (b) and (d). There, $\Fst$ may experience its strongest decrease if migration is strong.

Figure 7 also shows that at large migration rates, $\Fst$ may increase if the recombination rate increases, i.e., $\Fst$ is not minimized under linkage equilibrium. However, this occurs only for large recombination rates, i.e., larger than the largest selection coefficient. This is compatible with the finding in Section \ref{sec:general_case} that at high recombination rates, $\mmax$ may (slightly) increase in $\rh$, and the finding in Section \ref{sec:load} that the load $L$ may decrease with increasing $m$. We note that this `aberrant' behavior of $\mmax$, $L$, and $\Fst$ does not necessarily occur for the same parameter combinations.
Among more than $10^6$ parameter combinations of $\a_1,\a_2,\be_1,\be_2,m$, and $\ph$, we found no instance where $\Fst$ evaluated at the equilibrium $\Finf$ was higher than that at $\Fnull$ (results not shown).  
Importantly, if recombination is weak or migration is weak then $\Fst$ apparently always increases with tighter linkage.

Comparison of our multilocus $\Fst$ with averages of single-locus $\Fst$ values showed that the multilocus $\Fst$ declines somwehat faster at small migration rates than the 
averaged single-locus $\Fst$. For large parameter regions, the qualitative behavior of these measures of differentiation is the same. Differences occur only for a subset of selection coefficients at high migration rates and high recombination rates. Finally, we mention that our multilocus $\Fst$ is a sensitive measure of differentiation only if the effective number of haplotypes is low. This parallels the well known fact that the classical $\Fst$ is a sensitive measure of differentiation only if the effective number of alleles is low (e.g., Nagylaki 1998, 2011). Thus, our multilocus $\Fst$ may be most useful if applied to short sequences of DNA. A thorough and more general study is in preparation.

\section{Invasion of a locally beneficial mutant}
Differentiation between subpopulations can be increased by the invasion of mutants that establish a stable polymorphism at their locus.
Therefore, we consider a locus ($\A$) at which a new mutant $A_1$ arises that is advantageous relative to the wild type $A_2$ in deme 1, but disadvantageous in deme 2. In terms of our model, we assume $\a_1>0>\a_2$. If locus $\A$ is isolated, this mutant can invade and become established in a stable polymorphism if and only if $|\si_1+\si_2|<1$; cf.\ \eqref{si} and \eqref{eq:ASLPbf}. Using $m$ and $\ph$, this condition can be rewritten as
\begin{equation}
	m < |\mA|,
\end{equation}
see \eqref{mA} and \eqref{si12_mA_bound}, or
\begin{equation}
	\frac{m+\a_2}{m}\,\phA < \ph < \frac{m-\a_2}{m}\,\phA = \phinv.
\end{equation}
We restrict attention to the case $\ph>\phA$ \eqref{eq:phA} when the influx of the deleterious allele $A_2$ into deme 1 is sufficiently strong such that $\A_2$ is protected. (The case $\ph<\phA$ is symmetric and more suitable to study invasion of $\A_2$ under influx into deme 2 of $\A_1$ which is deleterious there.) Then the mutant $A_1$ can invade if any of the following equivalent conditions hold:
\begin{subequations}\label{A1_invade}
\begin{equation}\label{A1_m_invade}
 	m < \mA,
\end{equation}
\begin{equation}\label{A1_a1_invade}
	\a_1 > \frac{\a_2m\ph}{\a_2-m(1-\ph)} = \frac{m_1}{1-m_2/\a_2},
\end{equation}
or
\begin{equation}\label{A1_ph_invade}
	\phA < \ph < \phinv ,
\end{equation}
\end{subequations}
where $\phinv>1$ if and only if $m<\a_1$. Thus, $A_1$ can always invade if $m<\a_1$. For the CI model ($\ph=1$), each of the conditions in \eqref{A1_invade} simplifies to the well known invasion condition $m<\a_1$ (Haldane 1930). The conditions \eqref{A1_invade} show that invasion is facilitated whenever back migration is increased, either by keeping $m_1$ constant and increasing $m_2$, or by fixing $m$ and decreasing $\ph$.

For the CI model it was proved that invasion of a locally beneficial mutant is always facilitated by increased linkage to a locus in migration-selection balance (B\"urger and Akerman 2011). In fact, mutants of arbitrarily small effect can invade provided they are sufficiently tightly linked to this polymorphic locus which may be considered as the background in which the new mutant appears.  

Here, we investigate whether this is also the case with two-way migration. Thus, we assume that locus $\B$ is in migration-selection balance (which requires that analogs of \eqref{A1_invade} are satisfied for $\be_1$ and $\be_2$) and a locally beneficial mutant $A_1$ arises at the linked locus $\A$. Hence, the model in Section 2 applies and we assume \eqref{eq:parameter}.

\begin{figure}[t]
 \centering
 \includegraphics[scale=.8,keepaspectratio=true]{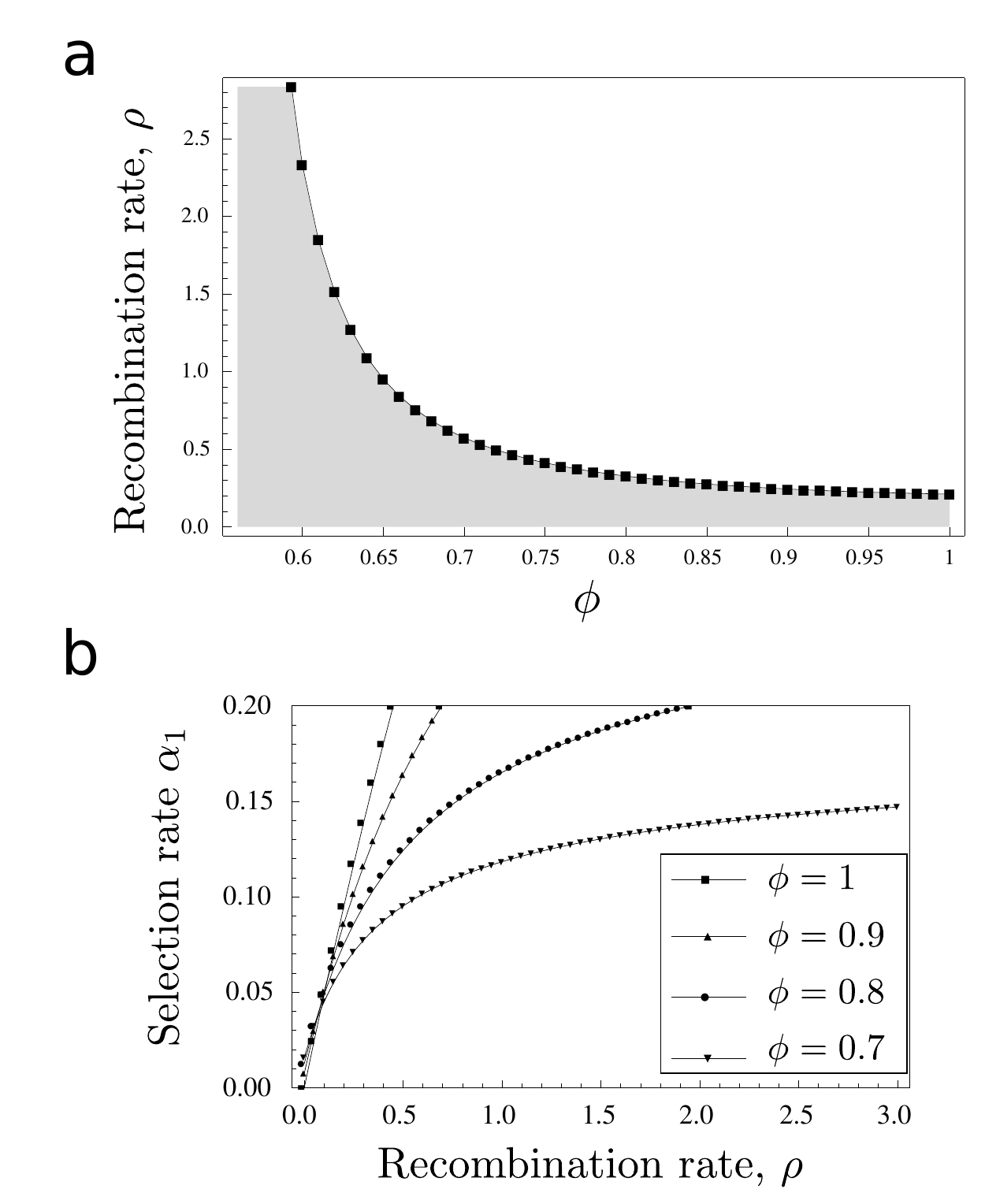}
 \caption{\small{Invasion properties of locally beneficial alleles. In a, the maximum recombination rate between loci $\A$ and $\B$, below which invasion of $\A_1$ can occur, is displayed as a function of $\ph$. The parameters $\a_1=-\a_2=0.1$, $\be_1=-2\be_2=2$, and $m=1$ are fixed. Therefore, $\phA=\tfrac12$ and $\phinv=0.55$. In b, the minimum selective advantage $\a_1$ required for invasion of $\A_1$ is shown as a function of $\rh$ for different values of $\ph$. The parameters $\a_2=-0.1$, $\be_1=-2\be_2=2$, and $m=1$ are fixed.}}
 \label{fig:invasion_reduced}
\end{figure}

Because we are mainly interested in the invasion properties of mutants of small effect, we assume that locus $\B$ is under stronger selection than $\A$, i.e., $|\a_k|<|\be_k|$ in deme $k=1,2$. Before the mutant $A_1$ arises, the population is at the equilibrium $\PBt$ (where $|\ta_1+\ta_2|<1$ must hold for admissibility; see Section 3). $A_1$ can invade if $\PBt$ is unstable. Since the eigenvalues determining external stability are zeros of a complicated quartic equations, the stability of $\PBt$ cannot be determined analytically. We expect that the new stable equilibrium that will be reached is the fully polymorphic equilibrium $\F$. For the CI model, this was be proved in (B\"urger and Akerman 2011). For the case of LE, it follows from Theorem \ref{LE_theorem}.

Figure 8 displays typical results about the invasion of the mutant $A_1$. In Figure 8a, the maximum recombination rate admitting invasion, denoted by $\rh_\text{max}$, is shown as a function of $\ph$.  In the shaded region, $A_1$ can invade. If $\ph\le\phinv=0.55$, \eqref{A1_ph_invade} implies that $\A_1$ can always invade. If $\ph>\phinv$, there exists $\rh_\text{max}<\infty$, such that $A_1$ can invade only if $\rh<\rh_\text{max}$, i.e., if $\A_1$ is sufficiently tightly linked to locus $\B$. 
In Figure 8b, the minimum selection coefficient $\a_1$ necessary for invasion of $A_1$ is shown as a function of $\rh/m$ for various values of $\ph$. These values are obtained by computing when the leading eigenvalue that determines external stability of $\PBt$ equals zero. 

We conclude that, as in the CI model, mutants of arbitrarily small effect can invade provided they are sufficiently tightly linked to a locus that is already maintained in migration-selection balance. In addition, as shown by both panels in Figure 8, increasingly symmetric migration facilitates the invasion and establishment of locally beneficial alleles.

%Section 6, CTTIM

\section{The effective migration rate at a linked neutral site}\label{sec:neutral_locus}
Linkage to loci under selection may impede or enhance gene flow at a neutral marker locus. In the first case, linkage may act as a barrier to gene flow. This was shown by the work of Petry (1983), Bengtsson (1985), Barton and Bengtsson (1986), and Charlesworth et al.\ (1997), who developed and studied the concept of the effective migration rate as a measure of the `effective' gene flow at a neutral site. More recently, the effective migration rate was studied for CI models with selection on a single locus in a class-structured population (Kobayashi et al.\ 2008) or with selection on two linked loci (B\"urger and Akerman 2011). Fusco and Uyenoyama (2011) investigated the consequences of a selectively maintained polymorphism on the rate of introgression at a linked neutral site under symmetric migration between two demes.

Here, we derive an explicit expression for the effective migration rate at a neutral locus ($\N$) that is located between the two selected loci, $\A$ and $\B$.
Recombination between locus $\A$ ($\B$) and the neutral locus occurs with rate $\rho_\AN$ ($\rho_\NB$) such that $\rho=\rho_\AN+\rho_\NB$. Thus, only one crossover event occurs in a sufficiently small time interval. We assume that $\rho_{\AN}$ and $\rho_{\NB}$ are positive, i.e., the neutral locus is not completely linked to a selected site. We consider two variants at the neutral locus, $N_1$ and $N_2$, each with arbitrary, positive initial frequency in at least one deme. The frequency of $N_1$ in deme $k(=1,2)$ is denoted by $n_k$. We model evolution at the three loci by a system of $7\times 2$ ordinary differential equations for
the allele frequencies and linkage disequilibria ($p_1$, $p_2$, $q_1$, $q_2$, $D^\AB_1$, $D^\AB_2$, $n_1$, $n_2$, $D^{\AN}_1$, $D^{\AN}_2$, $D^{\NB}_1$, $D^{\NB}_2$, $D^{\ANB}_1$, $D^{\ANB}_2$). We refrain from presenting the equations for the allele frequencies at the neutral locus and the associated  linkage disequilibria because they are a straightforward extension of those in Section 4.6 of B\"urger and Akerman (2011).

Obviously, the equilibrium allele frequencies at the neutral locus are the same in each deme and given by the initial allele frequencies averaged over the two demes:
  \begin{equation}
   \hat{n}_1=\hat{n}_2=\hat{n}=\frac{m_2 n_{1}(0)}{m}+\frac{m_1 n_{2}(0)}{m}.
    \label{eq:neutralequilibrium}
  \end{equation}
%Here, $n_{k}(0)$ is the initial frequency of $N_1$ in deme $k$. 
The equilibrium frequencies at the two selected loci are independent of the neutral locus andŽ, thus, the same as in the two-locus model treated above. The linkage disequilibria involving the neutral locus ($D^{\AN}_k$, $D^{\NB}_k$, and $D^{\ANB}_k$) are zero at equilibrium. By \eqref{eq:neutralequilibrium}, there is a one-dimensional manifold of equilibria resulting from the absence of selection at the neutral locus.

\begin{figure}[t]
\centering
\includegraphics[scale=0.8,keepaspectratio=true]{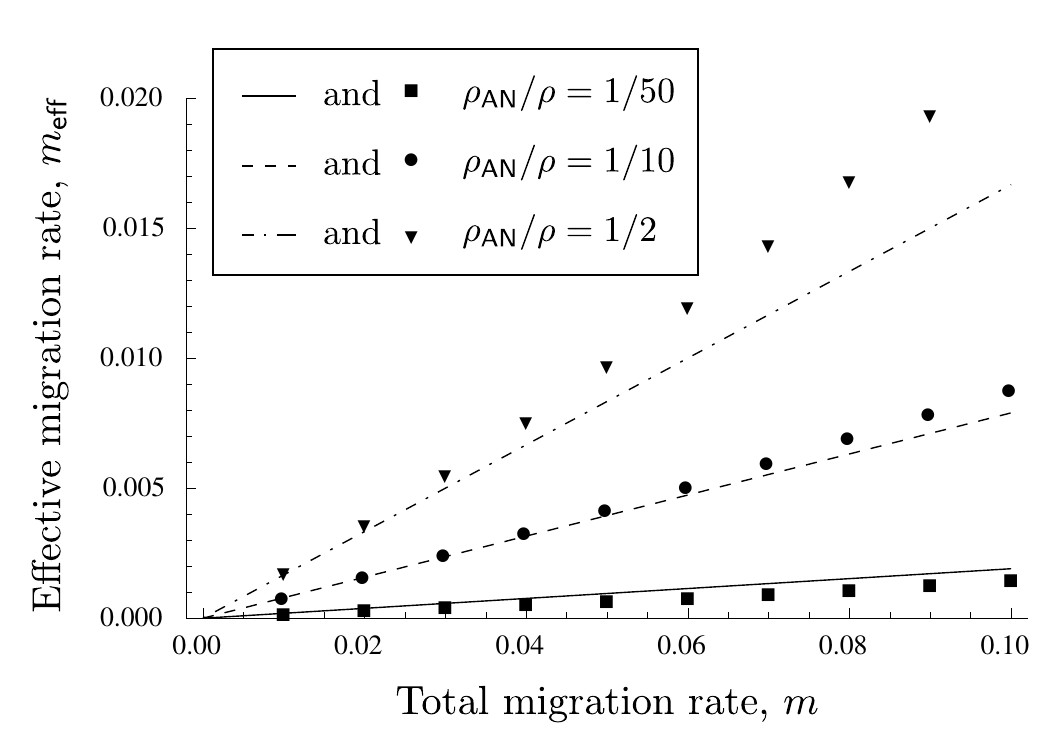}
\caption{\small{The effective migration rate $\meff$ as a function of $m$ for $\alpha_1=-\alpha_2=0.1$, $\beta_1=-\beta_2=0.2$, $\ph=\tfrac12$, and $\rho=0.2$. Recall that $\rh=\rh_\AN+\rh_\NB$. Lines show the weak-migration approximation of $\meff$ \eqref{eq:effectivemigrationrate}. Symbols give the exact numerical value of $\meff=-\la_N$.}}
 \label{fig:meff}
\end{figure}

We assume that parameters are such that the fully polymorphic equilibrium $\F$ is admissible and globally asymptotically stable.
Using the above order for the allele frequencies and linkage equilibria, the Jacobian at the equilibrium $\F$ has block structure,
\begin{equation}
   J= \begin{pmatrix} J_S & 0 \\ 0 & J_N  \end{pmatrix},
\end{equation}
where $J_S$ is the Jacobian describing convergence of $(p_1,p_2,q_1,q_2,D^\AB_1,D^\AB_2)$ to $\F$, and $J_N$ is the Jacobian describing convergence of $(n_1,n_2,D^{\AN}_1, D^{\AN}_2, D^{\NB}_1, D^{\NB}_2, D^{\ANB}_1, D^{\ANB}_2)$ to \newline $(\hat{n},\hat{n},0,0,0,0,0,0)$.

Because zero is the leading eigenvalue of $J_N$, the rate of convergence to equilibrium at the neutral locus is determined by the second largest eigenvalue of $J_N$, which we denote by $\la_N$. We define the effective (total) migration rate by $\meff=-\la_N$ (Bengtsson 1986, Kobayashi et al.\ 2008, B\"urger and Akerman 2011). It can be checked that under weak migration, i.e., to leading order in $m_1$ and $m_2$, one obtains
\begin{equation}\label{eq:effectivemigrationrate}
   \meff = -\la_N = m_1\frac{\rho_\AN\rho_\NB}{(\rho_\AN+\alpha_1)(\rho_\NB+\beta_1)} 
   					 + m_2\frac{\rho_\AN\rho_\NB}{(\rho_\AN-\alpha_2)(\rho_\NB-\beta_2)}
\end{equation}
(a \textit{Mathematica} notebook is available on request). If the neutral site is linked only to one selected locus (e.g., because $\be_1=\be_2=0$), then
\begin{equation}
   \meff = m_1\frac{\rho_\AN}{\rho_\AN+\alpha_1} + m_2\frac{\rho_\AN}{\rho_\AN-\alpha_2}
\end{equation} 
is obtained. Thus, two linked selected loci act as a much stronger barrier to gene flow than a single selected locus, especially
if the recombination rate between the two loci is not much larger than the selective coefficients.
In Figure 9, the approximation \eqref{eq:effectivemigrationrate} of the effective migration rate $\meff$ is displayed as a function of $m$ for various parameter combinations and compared with the exact value obtained by numerical evaluation of $\la_N$.

We note that $\meff$ is (approximately) the sum of the two effective one-way migration rates (B\"urger and Akerman 2011) and closely related to Kobayashi and Telschow's (2011) effective recombination rate. Our result complements their explicit
example on two-locus incompatibilities. We refer to their paper for the discussion of the relation of this concept of an effective migration rate to that of Bengtsson (1985) and for applications in the context of speciation theory.

%Section 7 (Discussion), CTTIM

\section{Discussion}
The purpose of this investigation was to improve our understanding of how genetic architecture, in particular recombination and locus effects, as well as the pattern and amount of migration determine polymorphism, local adaptation, and differentiation in a subdivided population inhabiting a heterogeneous environment. For simplicity, we restricted attention to two linked, diallelic loci and to migration between two demes. The study of diversifying selection in just two demes may also shape our intuition about clinal variation if the two subpopulations are from different ends of the cline. If alleles are beneficial in only one environment and detrimental in the other, local adaptation of subpopulations and differentiation between them can be obtained only if a (multilocus) polymorphism is maintained. Therefore, most of our mathematical results focus on existence and stability of polymorphic equilibria and on the dependence of the equilibrium configurations on the model parameters (migration rates, selection 
coefficients, recombination rate).

The model is introduced in Section 2. Sections 3 and 4 are devoted to the derivation of the possible equilibrium configurations and bifurcation patterns. They contain our main mathematical results. Explicit analytical results about existence and stability of equilibria were obtained for several limiting or special cases and are complemented by numerical work. 

The conditions for admissibility of all single-locus polymorphisms (SLPs) are given in Section \ref{sec:boundary_equil}, those for asymptotic stability of the monomorphic equilibria (ME) in Proposition \ref{prop_stab_ME} in Section \ref{sec:stab_mono}. The stability of SLPs could not generally be determined (Section \ref{ref:stab_SLP}). Weak migration is treated by perturbation methods in Section \ref{ref:weak_mig}. For sufficiently weak migration, there exists a globally attracting fully polymorphic equilibrium, $\F$ (Proposition \ref{prop_weak_mig}). Its approximate coordinates are given by \eqref{weak_mig}. 

The complete equilibrium and stability structure could be derived under the assumption of linkage equilibrium (Section \ref{ref:LE}). The unique, fully polymorphic equilibrium $\F=\Finf$ is admissible and globally attracting if and only if all four SLPs are admissible. Otherwise, one boundary equilibrium (SLP or ME) is globally asymptotically stable (Proposition \ref{prop_LE}). These results extend straightforwardly to an arbitrary number of diallelic loci. Based on these results, nonlinear perturbation theory establishes the existence of a globally stable, fully polymorphic equilibrium in a perturbed parameter range if recombination is sufficiently strong (Section \ref{sec:QLE}). This equilibrium is in quasi-linkage equilibrium and given by \eqref{F_largerho}.

Also for completely linked loci all equilibria and their local stability properties could be derived (Section \ref{sec:no_rec}). In this case, the fully polymorphic equilibrium $\Fnull$ \eqref{coord_Fnull} may lose stability while it is admissible \eqref{stab_F0}.  At this threshold a boundary equilibrium becomes stable by a `jump bifurcation' (Proposition \ref{prop_norec}). In general, however, more complicated equilibrium patterns than determined by Propositions \ref{prop_LE} and \ref{prop_norec} can occur, in particular, multiple stable equilibria.

In Section \ref{Section_asym}, we apply perturbation theory to infer the equilibrium properties under highly asymmetric migration from those derived for the continent-island model in B\"urger and Akerman (2011) and Bank et al.\ (2012). There, a stable ($\F$) and an unstable fully polymorphic equilibrium may exist if recombination is intermediate, and $\F$ is simultaneously stable with a boundary equilibrium. In general (Section \ref{Section_general}), we cannot exclude the existence of more than two internal equilibria or complicated dynamical behavior. Numerical searches produced no such instances. What can be shown easily is that, if $\rh<\infty$, any fully polymorphic equilibrium exhibits LD. In all cases, where an internal equilibrium was calculated (numerically or analytically), it exhibited positive LD.  

In the super-symmetric case, in which selection in deme 2 mirrors that in deme 1 and migration is symmetric, an assumption made in several applications, a fully polymorphic equilibrium exists always and, presumably, is stable (Section \ref{sec:super_symm}). This is a highly degenerate situation because if $\th\ne0$, only a monomorphic equilibrium can be stable for sufficiently large migration rates (Proposition \ref{prop_strong_mig}). If $\th=0$ (Section \ref{sec:theta=0}), then a fully polymorphic equilibrium can exist for arbitrarily large migration rates if $\ph=\phAB$ (see also Section \ref{sec:general_case}).

Whereas in Section 3 the focus was on the efficient presentation of the existence and stability results of equilibria, in Section 4 these results are used to derive the possible bifurcation patterns with the total migration rate $m$ as the bifurcation parameter.
All possible bifurcation patterns could be derived under the assumption of LE (Theorem \ref{LE_theorem}, Figures 2 and 3), and under the assumption of complete linkage (Theorem \ref{Mainresult_no_reco}, Figures 4 and 5). The latter case is considerably more complex. Interestingly, in each case, every bifurcation pattern can occur for every ratio $\ph=m_1/m$ of migration rates by choosing the selection coefficients appropriately. Hence, the assumption of symmetric migration does not yield simpler equilibrium configurations than general migration if arbitrary selection coefficients are admitted. 

In each of these cases (LE or $\rh=0$), we determined the maximum migration rate $\mmax$ admitting an asymptotically stable, fully polymorphic equilibrium (Corollaries \ref{corollary_mmaxinf} and \ref{corollary_mmax0}). The maximum migration rate $\mmax^0$ for $\rh=0$ always exceeds or equals that ($\mmax^\infty$) for LE, i.e., $\mmax^\infty\le\mmax^0$. Although for strong recombination, $\mmax$ can be very slightly smaller than $\mmax^\infty$ (Section \ref{QLE}), in the vast majority of investigated cases, $\mmax$ is bracketed by $\mmax^\infty$ and $\mmax^0$ (Figure 6, Section \ref{sec:general_case}).

Proposition \ref{prop_strong_mig} demonstrates that a ME is globally attracting if migration is sufficiently strong (except in the degenerate case noted above). If we interpret the equilibria $\M_2$ and $\M_3$ as fixation of a generalist ($A_1B_2$ and $A_2B_1$ are haplotypes of intermediate fitness), and $\M_1$ and $\M_4$ as fixation of a specialist ($A_1B_1$ and $A_2B_2$ are the locally adapted haplotypes), then depending on the sign of $\th$ one of the generalists becomes fixed for high $m$ if $\ph$ is intermediate (i.e., $\phA<\ph<\phB$ if $\th>0$, $\phB<\ph<\phA$ if $\th<0$; but note that, depending on the selection coefficients, both $\phA$ and $\phB$ can be arbitrarily close to 0 or 1.). The critical value $m$ as well as $\phA$ and $\phB$ are independent of $\rh$. Otherwise, one of the specialists becomes fixed for large $m$. 

The fact that a generalist becomes fixed for strong migration is a distinct feature of (balanced) two-way migration: in the CI model or if migration is sufficiently asymmetric ($\ph<\phA$ or $\ph>\phB$ if $\th>0$), one of the specialist haplotypes swamps the populations and becomes fixed. Another difference between highly asymmetric and more symmetric migration patterns is that in the first case, it is always the locus under weaker selection that first loses its polymorphism while $m$ increases, whereas this not necessarily so in the latter case (see Section \ref{sec:admissibility_SLP_m} and Theorem \ref{Mainresult_no_reco}, cases A3 and A4).

In summary, we determined quantitatively when the following three evolutionarily stable states discussed by Kawecki and Ebert (2004) occur: (i) existence of a single specialist optimally adapted to one deme and poorly to the other, (ii) existence of a single generalist type which has higher average fitness in the whole population than than any of the specialists, and (iii) existence of a set of specialists each adapted to its deme, i.e., coexistence in a polymorphism. Local adaptation and differentiation occur only in case (iii).

In Section 5, we used the migration load in each deme to quantify the degree of local adaptation. In Section 6 we introduced a new multilocus version of $\Fst$ to measure differentiation. If migration is weak, then local adaptation and differentiation decrease with increasing migration rate and increase with increasing linkage between the loci (Figure 7). In particular, for given (small) migration rate, local adaptation and differentiation are maximized if the fitness effects are concentrated on a single locus (corresponding to $\rh=0$ in our model). However, as discussed in Section 5, for high migration rates, the migration load of the total population can decrease with increasing recombination or migration rate. Similarly, at high recombination and migration rates, $\Fst$ can increase with increasing migration or recombination rate. Thus, for given, relatively high migration rate, $\Fst$ may be minimized at intermediate recombination rates. Apparently, it is always maximized in the absence of recombination.

In Section 7, we investigated the conditions for invasion of locally beneficial mutants. At an isolated locus, such a mutant can invade and become established in a migration-selection equilibrium if and only if its advantage exceeds a threshold that increases with the immigration rate of the wild type; see \eqref{A1_a1_invade}. If, however, this mutant occurs at a locus that is linked to a locus that is already in migration-selection balance, then its invasion is facilitated, i.e., its local selective advantage can be smaller (Figure 8b). Equivalently, for given selection coefficients and total migration rate, the minimum recombination rate needed for invasion increases if $\ph$, or the influx of the (deleterious) wild type relative to the efflux of the new mutant, increases (Figure 8a). For the extreme case of one-way migration from a `continental' population to an `island' population that is adapting to a new environment, B\"urger and Akerman (2011) proved that invasion of a locally beneficial mutant is 
always facilitated by increased linkage to a locus in migration-selection balance. 

Thus, our results complement the numerical finding by Yeaman and Whitlock (2011) for a multilocus quantitative-genetic model that clusters of locally adaptive mutations, or concentrated genetic architectures, build up in spatially structured populations with opposing selection pressures in two demes. Because tighter linkage is required for invasion under increasingly asymmetric migration rates, more concentrated architectures and a greater advantage for recombination-reducing mechanisms (such as chromosome inversions) should be expected for highly asymmetric migration. 
In finite populations, invasion of new mutants occurs only with a certain probability, and genetic drift may erase polymorphism. Numerical work, supported by analytical methods, has already shed some light on the dependence of the probability of establishment of new, locally adaptive mutations on the recombination rate and other factors (Yeaman and Otto 2011, Feder et al.\ 2012). Analytical work on the role of genetic drift and finite population size on these issues is in progress. 

Our results also show that, in the absence of epistasis and under the present form of balancing selection, reduced recombination between selected loci is favored, except when migration rates are sufficiently symmetric and high (Section 5). Selection inducing certain forms of epistasis may favor high recombination in structured populations more easily 
(Pylkov et al.\ 1998; Lenormand and Otto, 2000; Bank et al.\ 2012). Therefore, general predictions about the emergence of clusters of locally adaptive mutations in regions of reduced recombination, or of genomic islands of speciation (Wu and Ting 2004) or of differentiation (Feder et al.\ 2012), can not be made in the absence of detailed information about epistasis and the spatial pattern of selection and migration. At least in the absence of epistasis, the most favorable situation for the emergence of such clusters should occur in populations that are adapting to a new environment, still receiving maladaptive gene flow but sending out only very few or no migrants (corresponding to a continent-island model). 

In Section 8, we derived the approximation \eqref{eq:effectivemigrationrate} for the effective migration rate at a linked neutral locus that is located between the selected loci. This approximation is simply the sum of the two effective migration rates under one-way migration (B\"urger and Akerman, 2011). Because in the present model, polymorphism at the selected loci is maintained by balancing selection, the effective migration rate may be greatly reduced compared with the actual migration rate (see Figure 9). Thus, strong barriers against gene flow may build up at such neutral sites and enhance (neutral) differentiation (see Charlesworth and Charlesworth 2010, Chap.\ 8.3). Future work will have to study the actual amount and pattern of neutral diversity at such sites in finite populations.

\section*{Acknowledgments}
We are grateful for useful comments on the manuscript by two anonymous referees. One of them  inspired us to introduce the new multilocus fixation index. This work was supported by grants P21305 and P25188 of the Austrian Science Fund FWF. Support by the Vienna Graduate School of Population Genetics funded by the Austrian Science Fund (FWF, W1225) is also acknowledged.

\vspace{1cm}

\begin{appendix}
 \numberwithin{equation}{section}

\section{Appendix}\label{sec:appendix}

%\subsection{Parameter space of selection coefficients}\label{sec:parameters}
\subsection{Sufficiency of the assumptions \eqref{eq:parameter}}\label{sec:parameters}
By relabeling alleles, we can assume without loss of generality \eqref{assume}. Generically, one of the following nine parameter sets applies:
\begin{subequations}
\begin{align}
  \th &> 0, \;\;\a_1 < \be_1, \;\text{ and }\; \a_2 \ge \be_2 , \label{eq:ca1}\\ 
  \th &> 0, \;\;\a_1 < \be_1, \;\text{ and }\; \a_2 < \be_2 ,   \label{eq:ca2}\\
  \th &> 0, \;\;\a_1 \ge \be_1, \;\text{ and }\; \a_2 < \be_2 , \label{eq:ca3}\\
  \th &< 0, \;\;\a_1 > \be_1, \;\text{ and }\; \a_2 \le \be_2 , \label{eq:ca4}\\
  \th &< 0, \;\;\a_1 > \be_1, \;\text{ and }\; \a_2 > \be_2 ,   \label{eq:ca5}\\
  \th &< 0, \;\;\a_1 \le \be_1, \;\text{ and }\; \a_2 > \be_2 . \label{eq:ca6}
\end{align}
In addition, there are the following three parameter sets:
\begin{align}
  \th &= 0, \;\;\a_1 < \be_1 \;\text{ and }\; \a_2 > \be_2 ,\label{eq:ca7}\\
  \th &= 0, \;\;\a_1 = \be_1 \;\text{ and }\; \a_2 = \be_2 ,\label{eq:ca8}\\
  \th &= 0, \;\;\a_1 > \be_1 \;\text{ and }\; \a_2 < \be_2 .\label{eq:ca9}
\end{align}
\end{subequations}
The sets \eqref{eq:ca1} -- \eqref{eq:ca9} yield the complete parameter space of the selection coefficients. 

We show that the parameter sets \eqref{eq:ca3} -- \eqref{eq:ca6} can be derived from \eqref{eq:ca1} and \eqref{eq:ca2} by simple transformations. Let $f$ denote the exchange of loci, i.e., $f(\a_k)=\be_k$ and $f(\be_k)=\a_k$, and $g$ the exchange of demes, i.e.,  $g(\a_k)=-\a_{k^\ast}$ and $g(\be_k)=-\be_{k^\ast}$. We observe that $\text{sign} (f(\th))=\text{sign}(g(\th))=-\text{sign}(\th)$ and
\begin{subequations}
\begin{align}
  \eqref{eq:ca1} &\stackrel{f}{\rightarrow} \eqref{eq:ca4} \stackrel{g}{\rightarrow} \eqref{eq:ca3} \stackrel{f}{\rightarrow} \eqref{eq:ca6},\\
  \eqref{eq:ca2} &\stackrel{g}{\rightarrow} \eqref{eq:ca5}
\end{align}
\end{subequations}
hold. 
Therefore, \eqref{eq:parameter} is sufficient to describe the (generic) parameter region where $\th\ne0$. Since 
\begin{align}
  \eqref{eq:ca7} \stackrel{f}{\rightarrow} \eqref{eq:ca9},
\end{align}
\eqref{eq:parameterdeg} is sufficient to describe the degenerate cases $\th=0$.

\subsection{Proof of Proposition \ref{prop_stab_ME}}\label{App_Proof_Prop_stabM1}
At each monomorphic equilibrium, the characteristic polynomial factors into three quadratic polynomials, $P(\la) = t_1(\la)t_2(\la)t_3(\la)$. Two of them, $t_1(\la)$ and $t_2(\la)$, determine stability with respect to the adjacent marginal one-locus systems. The corresponding conditions are already known from one-locus theory. The third, $t_3(\la)$, determines stability with respect to the interior of the state space. 

In the following, we derive the stability conditions \eqref{M1_stab_marginal_orig} and \eqref{M1_stab_rh} for $\M_1$. 
Those for $\M_4$ can be deduced analogously or by symmetry considerations by taking into account that \eqref{be1>a1} implies
$\min\{\a_1,\be_1\}=\a_1$. The stability analysis of $\M_2$ and $\M_3$ is much simpler and left to the reader. 

For $\M_1$, it is straightforward to show that
\begin{subequations}\label{M1_charpol}
\begin{align}
  t_1(\la)&= \la^2 + [\a_1(1+\si_1)+\a_2(1+\si_2)]\la + \a_1\a_2(1+\si_1+\si_2), \\
  t_2(\la)&= \la^2 + [\be_1(1+\ta_1)+\be_2(1+\ta_2)]\la + \be_1\be_2(1+\ta_1+\ta_2), \\
  t_3(\la)&= \la^2+(\a_1+\a_2+\be_1+\be_2+2\rh+m_1+m_2)\la \nonumber \\
  				&\quad+ (\a_1+\be_1+m_1+\rh)(\a_2+\be_2+m_2+\rh)-m_1m_2.
\end{align}
\end{subequations}
Because $t_1''(\la)>0$ for every $\la$, $t_1'(0)=\a_1(1+\si_1)+\a_2(1+\si_2)>0$ if $\si_2<-1$, $t_1(0)>0$ if and only if
$\si_1+\si_2<-1$, and $\min_\la \{t_1(\la)\}<0$, we conclude that the two eigenvalues emanating from $t_1$ are negative if and only if 
\begin{subequations}\label{eq:M1_stability_Appendix}
\begin{equation}\label{si1+si2<-1}
	\si_1+\si_2 < -1.
\end{equation}
Analogously, the two eigenvalues emanating from $t_2$ are negative if and only if 
\begin{equation}\label{ta1+ta2<-1}
	\ta_1+\ta_2 < -1,
\end{equation}
and those originating from $t_3$ are negative if and only if 
\begin{equation}\label{eq:M1_lambda3_negative}
	m_2 > -\frac{(\a_1+\be_1+m_1+\rh)(\a_2+\be_2+\rh)}{\a_1+\be_1+\rh}.
\end{equation}
\end{subequations}
Conditions \eqref{si1+si2<-1} and \eqref{ta1+ta2<-1} yield \eqref{M1_stab_marginal_orig}.

Concerning \eqref{eq:M1_lambda3_negative}, we observe that it is always satisfied if $\rh>-(\a_2+\be_2)$ because then the right-hand side is negative. Next we show, that \eqref{eq:M1_lambda3_negative} is also satisfied if $\rh>-\a_2$. Because the right-hand side of
\eqref{eq:M1_lambda3_negative} is strictly monotone decreasing in $\rh$, it is sufficient to prove that \eqref{eq:M1_lambda3_negative}
holds if $\rh=-\a_2$. Then simple rearrangement of \eqref{eq:M1_lambda3_negative} leads to the condition
\begin{equation}
	\frac{m_2(\a_1+\be_1-\a_2)}{\be_1\be_2} + \frac{\a_1+\be_1-\a_2+m_1}{\be_1} <0,
\end{equation}
which can be rewritten as
\begin{equation}
	\ta_1+\ta_2+1 + \frac{\a_1-\a_2}{\be_1}(1+\ta_2)<0.
\end{equation}
This is satisfied if \eqref{ta1+ta2<-1} holds because this also implies $1+\ta_2<0$. One shows similarly that \eqref{eq:M1_lambda3_negative} is satisfied if $\rh\ge-\be_2$. 
Therefore, we have proved that $\M_1$ is asymptotically stable if \eqref{M1_stab_marginal_orig} and \eqref{M1_stab_rh} hold.

\subsection{Calculation of equilibria with two polymorphic loci if $\rh=0$}\label{int_equi_rho0}
As shown in the main text, by Corollary 3.9 of Nagylaki and Lou (2007) it is sufficient to assume that $A_1B_2$ is absent, which implies $D_k=p_k(1-q_k)$ and $p_k\le q_k$. Setting $\rh=0$, we find from the equations $\dot p_1=0$ and $\dot q_1=0$ \eqref{eq:dynamics} that
\begin{subequations}\label{manif_p2q2}
\begin{align}
	p_2 &= p_1[m_1-\a_1(1-p_1)-\be_1(1-q_1)]/m_1,\\
	q_2 &= 1 - (1-q_1)(m_1+\a_1p_1+\be_1q_1)/m_1
\end{align}
\end{subequations}
holds at equilibrium.
Substituting \eqref{manif_p2q2} into $\dot p_2$ and $\dot q_2$, we obtain at equilibrium,
\begin{subequations}
\begin{align}
	0 &= p_1[g_1(p_1,q_1)-\a_1^2\a_2p_1^3-\be_1^2\be_2q_1^3]/m_1^2, \\
	0 &= (1-q_1)[g_2(p_1,q_1)-\a_1^2\a_2p_1^3-\be_1^2\be_2q_1^3]/m_1^2,
\end{align}
\end{subequations}
where $g_1$ and $g_2$ are quadratic polynomials in $(p_1,q_1)$. The obvious substitution results in the equilibrium condition
\begin{subequations}\label{int_equi_cond_rho0}
\begin{align}
	0 &=  \a_1\a_2(\a_1+\be_1)p_1^2 + \be_1\be_2(\a_1+\be_1)q_1^2 + (\a_1+\be_1)(\a_1\be_2+\a_2\be_1)p_1q_1 \notag \\
	 &\qquad + [m_1(\a_1(2\a_2+\be_2)+\a_2\be_1)-(\a_1+\be_1)(\a_2\be_1+\a_1(\a_2+\be_2))]p_1 \notag \\
	 &\qquad+ [m_1(\be_2(\a_1+2\be_2)+\a_2\be_1)-\be_1\be_2(\a_1+\be_1)]q_1 \notag \\
	 &\qquad+ m_1[m_1(\a_2+\be_2)+m_2(\a_1+\be_1)-(\a_1+\be_1)(\a_2+\be_2)].
\end{align}
\end{subequations}
It is easy to check that $\Fnull$ always fulfills this condition and it is the only solution satisfying $0\le p_1=q_1\le 1$.
Hence, unless there is curve $(p_1,q_1)$ of solutions of \eqref{int_equi_cond_rho0} that passes through $\Fnull$ and through
either a point on $p_1=0$ with $0< q_1\le1$ or on $q_1=1$ with $0\le p_1<1$, $\Fnull$ is the unique admissible solution of \eqref{int_equi_cond_rho0}.

Because $\Fnull$ has an eigenvalue 0 only if either \eqref{critmu12} is satisfied or if $|\ka_1+\ka_2|=1$ (which occurs if and only if $\Fnull$ collides with either $\M_1$ or $\M_4$), $\Fnull$ is the only equilibrium with both loci polymorphic, except when \eqref{critmu12} is satisfied. In the latter case, a line of equilibria exists, as we show now.

We calculate $m_2$ from \eqref{critmu12} and substitute into \eqref{int_equi_cond_rho0}. The right-hand side factorizes into two linear terms. Only one of them gives rise to admissible equilibria and, in fact, yields the manifold:
\begin{equation}\label{manif}
	p_1 = \frac{\th[\be_1(1-q_1)-m_1]-\a_1\be_1(\a_2+\be_2)}{\a_1\th},\\
\end{equation}
where $0\le q_1\le 1$. The allele frequencies in the other deme are obtained from \eqref{manif_p2q2}. It is straightforward to check that not only $\Fnull$, but also the equilibria $\PAo$ and $\PBt$ lie on this manifold. In terms of the gamete frequencies, this manifold is a straight line.

\subsection{Stability of $\Fnull$}\label{sec:Fnull_stability}
In this section we derive the stability of $\Fnull$.

As $A_1B_2$ is lost if $\rh=0$ and \eqref{eq:parameter} hold, it is sufficient to consider the dynamics \eqref{dynamics_gametes} in $S_3\times S_3$. In this case, the characteristic polynomial at $\Fnull$ factors into two quadratic polynomials, $P(\la)=t_1(\la)t_2(\la)$. These are given by
\begin{subequations}
\begin{align}
    t_1(\la)=&\la^2+\left[(\a_1+\be_1-\a_2-\be_2)\sqrt{1-\ka_1\ka_2}-(m_1+m_2) \right]\la \nonumber\\
	    &\quad +(\a_1+\be_1)(\a_2+\be_2)\left[1+(\ka_1-\ka_2)\sqrt{1-\ka_1\ka_2}\right],\\
    t_2(\la)=&\la^2+\frac{1}{2}\left[\a_1+\a_2-\be_1-\be_2+(\a_1-\a_2+\be_1-\be_2)\sqrt{1-\ka_1\ka_2}\right]\la \nonumber\\
	    &\quad +\frac{1}{2}\left[-\a_1\be_2(1+\sqrt{1-\ka_1\ka_2})-\a_2\be_1(1-\sqrt{1-\ka_1\ka_2})\right].
\end{align}
\end{subequations}
The polynomial $t_1$ determines the stability with respect to the (effectively one-locus) system where only 'alleles' $A_1B_1$ and $A_2B_2$ are present. It is convex with $t_1(0)\geq 0$ if and only if $|\ka_1+\ka_2|\leq 1$ (where the equalities correspond), i.e., whenever $\Fnull$ is admissible, cf.\ \eqref{eq:Fo_admissible}. If $|\ka_1+\ka_2|<1$, $t^{'}_1(0)>0$ and $t_1$ attains a negative value at its minimum (as can be shown easily). Therefore, all eigenvalues emanating from $t_1$ are real and negative whenever $\Fnull$ is admissible.

The polynomial $t_2$ determines stability with respect to the interior of $S_3\times S_3$. It is convex and attains its minimum at
\begin{equation}
  \la_\text{min}=\frac{1}{4}\left[(\a_1+\be_2-\a_2-\be_2)(1-\sqrt{1-\ka_1\ka_2})\right]
\end{equation}
where $\la_\text{min}<0$ by \eqref{assume} and \eqref{eq:kappa}. As $t_2(\la_\text{min})<0$, the eigenvalues emanating from $t_2$ are real. As
\begin{equation}
  t_2(0)\geq 0 \iff m_1m_2\leq \mtilde,
\end{equation}
where the equalities correspond and $\mtilde$ is defined in \eqref{mcrit}, and because $t^{'}_2(0)>0$, we conclude that the two eigenvalues emanating from $t_2$ are negative if and only if \eqref{stab_F0} holds.

\subsection{Stability of SLPs under $\rh=0$}\label{sec:PB2_PA1_rhnull_stability}
For $\rh=0$ it is sufficient to study the dynamics \eqref{dynamics_gametes} in $S_3\times S_3$. SLPs where $\hat{x}_{k,2}>0$ ($k=1,2$), i.e., $\PAt$ and $\PBo$, are unstable. It remains to study the stability of $\PAo$ and $\PBt$. 

We present the analysis for $\PAo$ in detail, as results for $\PBt$ follow analogously. 

At $\PAo$ the characteristic polynomial factors into two quadratic polynomials, $P(\la)=t_1(\la)t_2(\la)$, given by
\begin{subequations}
\begin{align}
    t_1(\la)=&\la^2+\left[\a_1(\sqrt{1-4\si_1\si_2}-\si_1)-\a_2(\sqrt{1-4\si_1\si_2}+\si_2)\right]\la \nonumber\\
	      & +\a_1\a_2\left[\right(\si_1-\si_2)\sqrt{1-4\si_1\si_2}-(1-4\si_1\si_2)],\\
    t_2(\la)=&\la^2+\frac{1}{2}\left[2\be_1+2\be_2+\a_1(1-\sqrt{1-4\si_1\si_2})+\a_1(1+\sqrt{1-4\si_1\si_2})\right]\la \nonumber\\
	      & +\frac{1}{2}\left[\be_1(\be_2+\a_2)+\be_2(\a_1+\be_1)+\theta \sqrt{1-\si_1\si_2}\right].
\end{align}
\end{subequations}
$t_1$ determines stability with respect to the one-locus system where $B_1$ is fixed. $t_1(0)=0$ if and only if $|\si_1+\si_2|=1$, i.e., whenever $\PAo$ collides with a ME according to \eqref{eq:ASLPbf_1} and \eqref{eq:ASLPbf_2}. Whenever $|\si_1+\si_2|<1$, i.e., $\PAo$ is admissible \eqref{si}, $t_1(0)>0$ and $t_1^{'}(0)>0$. As $t_1^{''}(\la)>0$ for every $\la$, $t_1$ attains a minimum, where it is straightforward to show that $t_1$ takes a negative value at its minimum. Thus, all eigenvalues emanating from $t_1$ are real and negative whenever $\PAo$ is admissible.

$t_2$ determines stability with respect to the interior of $S_3\times S_3$. $t_2(0)\geq 0$, if and only if $m_1m_2\geq \mtilde$, cf.\ \eqref{mcrit}, where the equalities correspond. Whenever $m_1m_2>\mtilde$, $t_2^{'}(0)>0$. As $t_2^{''}(\la)>0$ for every $\la$, $t_2$ attains a minimum, where it is straightforward to show that $t_2$ takes a negative value at its minimum. Thus, all eigenvalues emanating from $t_1$ are real and negative whenever $m_1m_2>\mtilde$ holds. Otherwise, at least one eigenvalue is positive.

Combining the results obtained for $t_1$ and $t_2$ it follows that $\PAo$ is asymptotically stable if and only if
\begin{equation}\label{eq:PAo_EV_negative}
    -1<\si_1+\si_2<1 \text{ and } m_1m_2>\mtilde
\end{equation}
hold. We note that $m_1m_2>\mtilde$ is equivalent to $(\si_1\ta_2-\si_2\ta_1)^2<-(\si_1+\ta_1)(\si_2+\ta_2)$, and our general assumption \eqref{eq:parameter} implies $\ta_1<\si_1$ and $\si_1\ta_2-\si_2\ta_1<0$. Using these relations we can show with the help of \textit{Mathematica} that \eqref{eq:PAo_EV_negative} is incompatible with $-1<\ta_1+\ta_2$. Consequently, $\PBt$ is not admissible if $\PAo$ is asymptotically stable.

\subsection{The super-symmetric case}\label{sec:super_symmetric}
We prove that in the super-symmetric case of Section \ref{sec:super_symm}, all SLPs are unstable. 

We assume symmetric migration rates ($m_1=m_2=m$), equivalent loci ($\a_k=\be_k=a$), and selection in deme 2 mirrors that in deme 1 ($\a_k=-\a_{k^\ast}$). Thus, $\th=0$. Equilibria may collide (thus leave or enter the state space) if and only if at least one of their eigenvalues is zero. Eigenvalues are zeros of the characteristic polynomial, which has the form $P(\la) = c_6\la^6 + \dots + c_1 \la + c_0$. If zero is an eigenvalue at an equilibrium, i.e., $P(0)=0$, the constant term $c_0$ must vanish. In the super-symmetric case every characteristic polynomials at an SLP has the same constant term 
\begin{equation}
 c_0=-a^2 \rh \left(2a^2 \sqrt{a^2+m^2}+ m (3m-\rh) \sqrt{a^2+m^2} - (a^2+m^2)(3m-\rh) \right).
\end{equation}
One can show that $c_0=0$ is impossible if $m > 0$.

\subsection{Important quantities and relations}\label{sec:important_quantities_cont}
The following section complements Section \ref{sec:important_quantities}. Here, we derive all relations of $\ph^\X$ \eqref{eq:phAphB} and $m^\X$ \eqref{eq:mAmB} needed in Sections \ref{sec:admissibility_SLP_m} to \ref{sec:bifs_norec} and in the proofs of the theorems there.

Using \eqref{eq:mA_mB_2}, \eqref{eq:mA_mB}, \eqref{eq:mA_positiv}, \eqref{eq:mB_positiv}, \eqref{phA<phB} and \eqref{phABt<phA}, we derive all possible inequalities between $\mA$ and $\mB$:
\begin{subequations}\label{mAB_relations}
\begin{align}
		0<\mA<\mB &\;\iff\; \phB \le \ph, \label{0<mA<mB}\\
		0<-\mA <-\mB &\;\iff\; \a_2 > \be_2 \text{ and } \ph<\phABt, \label{0<-mA<-mB} \\
		0<-\mB <-\mA &\;\iff\; 
				\begin{cases} \a_2 \le \be_2 \text{ and } \ph<\phA, \text{ or } \\
											\a_2 > \be_2 \text{ and } \phABt<\ph<\phA,
				\end{cases} 	 \label{0<-mB<-mA} \\				
		0 < -\mB < \mA &\;\iff\; \phA \le \ph<\phAB, \label{0<-mB<mA}\\
		0 < \mA < -\mB &\;\iff\; \phAB<\ph<\phB, \label{0<mA<-mB}
\end{align}
where
\begin{equation}\label{eq:infeasible_m}
			0<\mB\leq \mA,\; 0 < -\mA \leq \mB, \text{ and } 0 < \mB \leq -\mA \;\text{are infeasible}.
\end{equation}
\end{subequations}
Using \eqref{eq:m_admiss}, \eqref{mAB_PhAB}, and \eqref{phA<phB}-\eqref{eq:ph_rel} we obtain the following inequalities for $\mast$:
\begin{subequations}\label{eq:mast_mFnull_relation}
\begin{align}
	  0<\mast<-\mFnull & \;\iff\; \phMo(\rh=0)<\ph<\phFnull, \label{eq:mast<-mFnull}\\
	  0<-\mFnull<\mast & \;\iff\; \ph<\phMo(\rh=0) , \label{eq:mast>-mFnull} \\
	  0<\mast<\mFnull  & \;\iff\; \phFnull<\ph<\phMf(\rh=0) , \label{eq:mast<mFnull}\\
	  0<\mFnull<\mast  & \;\iff\; \phMf(\rh=0)<\ph , \label{eq:mast>mFnull} \\
	  0<\mast<-\mA & \;\iff\; \phMo(\rh=0)<\ph<\phA, \label{eq:mast<-mA}\\
	  0<\mast<-\mB & \;\iff\; \phAB<\ph<\phB, \label{eq:mast<-mB}\\
	  0<\mast<\mA  & \;\iff\; \phA<\ph<\phAB, \label{eq:mast<mA}\\
	  0<\mast<\mB  & \;\iff\; \phB<\ph<\phMf(\rh=0), \label{eq:mast<mB}\\
	  0<-\mA<\mast & \;\iff\; \ph<\phMf(\rh=0), \\
	  0<-\mB<\mast & \;\iff\; \ph<\phAB, \\
	  0<\mA<\mast  & \;\iff\; \phAB<\ph, \\
	  0<\mB<\mast  & \;\iff\; \phMf(\rh=0)<\ph \label{eq:mB<mast}.
\end{align}
\end{subequations}
From  \eqref{phA<phB}, \eqref{phABt<phA}, \eqref{mAB_relations}, and \eqref{eq:mast<-mA} -- \eqref{eq:mB<mast} we infer
\begin{equation}\label{mmaxinf<mmast}
	\min\{|\mA|,|\mB|\} \le \mast.
\end{equation}

Next, we derive the relations between $\mFnull$ and $\mA$ or $\mB$ needed in the proof of Theorem \ref{Mainresult_no_reco}. As their derivation is lengthy, the reader may wish to skip the proof and go immediately to the results given by \eqref{eq:crit_m_rel_r0} and \eqref{mmaxinf<mFnull}.

Our approach to derive the possible relations between $\mFnull$ and $\mA$ or $\mB$ is as follows: First, we derive all relevant relations of $\ph^\X$ \eqref{eq:phAphB} for arbitrary recombination $\rh$. We use these relations to determine the required relations between $\mA$, $\mB$, $\mMo$ and $\mMf$ for arbitrary $\rh$. By setting $\rh=0$ in the results obtained and by the equivalence given in \eqref{eq:mFnull_mMf_mMo}, the possible relations between $\mFnull$ and $\mA$ or $\mB$ follow immediately.

%Relations of ph
By definition, the values $\phA$, $\phB$, $\phFnull$, $\phABt$, $\phAB$, $\phAFnull$, and $\phBFnull$ \eqref{eq:phAphB} are independent of the recombination rate $\rh$. Their relations under \eqref{eq:parameter} are given in \eqref{phA<phB}, \eqref{phABt<phA}, and \eqref{phFnull_ph_A_B_Fnull}. 

The values $\phMo$, $\phMot$, $\phMf$, and $\phMft$ \eqref{eq:phAphB} depend on $\rh$, and we analyze this dependence in the following.
The conditions which determine the admissibility of $\ph^{\M_\textsf{i}}$ and $\tilde{\ph}^{\M_\textsf{i}}$ ($i=1,4$) are:
\begin{subequations}\label{eq:phMo_phMf}
\begin{align}
	0 < \phMo < 1 & \;\iff\; 
				 0\leq \rh <-\be_2 \text{ or } \rho>-\a_2-\be_2, \\
	0 < \phMot < 1 & \;\iff\; 
				0\leq \rh <-\a_2 \text{ or } \rho>-\a_2-\be_2, \\
	0 < \phMf < 1 & \;\iff\; 
				0\leq \rh <\a_1 \text{ or } \rho>\a_1+\be_1, \\
	0 < \phMft < 1 & \;\iff\; 
				0\leq \rh <\be_1 \text{ or } \rho>\a_2+\be_2,
\end{align}
\end{subequations}
with the relations
\begin{subequations}\label{eq:phM_interval_1}
\begin{align}
	0\leq \rh <-\be_2 	& \;\; \Longrightarrow \;\; 0 < \phMo < \phA, \label{eq:phMo_phA}\\
	0\leq \rh < \a_1 	& \;\; \Longrightarrow \;\; \phB< \phMf < 1. \label{eq:phB_phMf}
\end{align}
\end{subequations}
To determine further relations of $\phMo$, $\phMot$, $\phMf$, and $\phMft$, we define the following critical recombination rates:
\begin{subequations}\label{eq:ME_critical_rho}
\begin{align}
      \rhMo 	& =  \frac{\a_2\be_1(\a_2+\be_1)-\a_1\be_2(\a_1+\be_2)}{2\th} \nonumber \\
		& \quad +\frac{\sqrt{(\a_2\be_1(\a_2+\be_1)-\a_1\be_2(\a_1+\be_2))^2-4 \th^2 (\a_2\be_1+\a_1\be_2)}}{2\th}, \\
      \rhMot 	& = \frac{-\th(\a_1+\a_2)+\a_2\be_1^2-\a_1\be_2^2}{2\th} \nonumber \\
		& \quad +\frac{\sqrt{(-\th(\a_1+\a_2)+\a_2\be_1^2-\a_1\be_2^2)^2-4\th^2(\a_2\be_1+\a_1(\a_2+\be_2))}}{2\th}, \\
      \rhMf	& = \frac{\a_1\be_2(\a_1+\be_2)-\a_2\be_1(\a_2+\be_1)}{2\th} \nonumber \\
		& \quad + \frac{\sqrt{(\a_1\be_2(\a_1+\be_2)-\a_2\be_1(\a_2+\be_1))^2-4\th^2(\a_2\be_1+\a_1\be_2)}}{2\th}, \\
      \rhMft	& = \frac{\th(\be_1+\be_2)+\a_1^2\be_2-\a_2^2\be_1}{2\th} \nonumber \\
		& \quad + \frac{\sqrt{(-\th(\be_1+\be_2)+\a_2^2\be_1-\a_1^2\be_2)^2-4\th^2(\a_2\be_1+\be_2(\a_1+\be_1))}}{2\th}.
\end{align}
\end{subequations}
Next, we determine the admissibility of $\rh^\X$ and $\tilde{\rh}^X$ defined in \eqref{eq:ME_critical_rho}. Therefore, we partition the selection parameters satisfying \eqref{eq:parameter} and $\be_2>\a_2$ according to
\begin{subequations}
\begin{equation}\label{eq:be2_a2_condi1}
       \be_2 > \a_2  +\frac{\a_1\be_2(\be_2-\a_1)}{\be_1(\a_1+\be_1-\be_2)} > \a_2
\end{equation}
and
\begin{equation}\label{eq:be2_a2_condi2}
	\a_2 +\frac{\a_1\be_2(\be_2-\a_1)}{\be_1(\a_1+\be_1-\be_2)}>\be_2 > \a_2.
\end{equation}
\end{subequations}
Analogously, the selection parameters satisfying \eqref{eq:parameter} and $\be_1>\a_1$ can be partitioned according to
\begin{subequations}
\begin{equation}\label{eq:be1_a1_condi1}
      \be_1 > \a_1+\frac{\a_1\a_2(\a_1-\a_2)}{\be_2(\a_1-\a_2-\be_2)}>\a_1
\end{equation}
and
\begin{equation}\label{eq:be1_a1_condi2}
      \a_1+\frac{\a_1\a_2(\a_1-\a_2)}{\be_2(\a_1-\a_2-\be_2)}>\be_1>\a_1.
\end{equation}
\end{subequations}
Using these partitions, we obtain that $\rh^\X$ and $\tilde{\rh}^X$ satisfy the following relations (as can be checked with \textit{Mathematica}):
\begin{subequations}\label{eq:r_crit_order}
\begin{align}
       \rhMot< \rhMo < -\a_2 <-\be_2 &\;\iff\; \be_2 < \a_2, \label{rhMo<-a2} \\
      \rhMot < \rhMo = -\a_2 = -\be_2 &\;\iff\; \be_2 = \a_2,\\
      -\be_2<\rhMot< -\a_2<\rhMo<-\a_2-\be_2 &\;\iff\; \eqref{eq:be2_a2_condi1} \text{ holds },\\
      \rhMot<-\be_2< -\a_2<\rhMo<-\a_2-\be_2  &\;\iff\;
      \eqref{eq:be2_a2_condi2} \text{ holds },\\
      \a_1<\rhMft<\be_1<\rhMf<\a_1+\be_1  &\;\iff\; \eqref{eq:be1_a1_condi1} \text{ holds }, \\
      \rhMft<\a_1<\be_1<\rhMf<\a_1+\be_1 &\;\iff\;
      \eqref{eq:be1_a1_condi2}  \text{ holds }.
\end{align}
\end{subequations}

As $\rh \geq 0$ and \eqref{eq:parameter} hold, we obtain that
\begin{subequations}\label{eq:phM_equality}
\begin{align}
      \phMo = \phMot	       	& \;\iff\; \rh = -\a_2-\be_2 \text{ or } \rh = \rhMo, \\
      \phABt = \phMo	\;\iff\;  \phABt = \phMot &\;\iff\;  \rh = \rhMo, \\
      \phA =\phMot 		& \;\iff\; \rh = \rhMot,
\end{align}
\end{subequations}
and
\begin{subequations}\label{eq:phM_equality_2}
\begin{align}
      \phMf = \phMft	       	& \;\iff\; \rh = \a_1+\be_1 \text{ or } \rh = \rhMf, \\
      \phB =\phMft 		& \;\iff\; \rh = \rhMft,
      \end{align}
where
      \begin{equation}
      \phA =\phMo \text{ and } \phB =\phMf \; \text{ are infeasible}.
\end{equation}
\end{subequations}
We derive the following relations additinal to \eqref{eq:phMo_phA} and \eqref{eq:phB_phMf}, using \eqref{eq:r_crit_order} and \eqref{eq:phM_equality}:
\begin{subequations}\label{eq:phMot_phMft}
\begin{align}
      0<\phMot<\phA & \;\iff\; \rhMot <\rh<-\a_2,\\
      \phA<\phMot<1 & \;\;\Longleftarrow\; 0\leq \rh< \rhMot,
\end{align}
\end{subequations}
and
\begin{subequations}\label{eq:phMot_phMft_2}
\begin{align}
      \phB<\phMft<1 & \;\iff\; \rhMft <\rh<\be_1, \\
      0<\phMft<\phB & \;\;\Longleftarrow\; 0\leq \rh <\rhMft,
\end{align}
\end{subequations}
and by recalling \eqref{phABt<phA}:
\begin{subequations} \label{eq:phM_interval_2}
\begin{align}
	\phABt < \phMo <\phA 	& \; \iff \; \be_2 < \a_2 \;\text{and}\; 0\leq \rh < \rhMo ,\\
	0 < \phMo < \phABt 	& \; \iff \; \be_2 < \a_2 \;\text{and}\; \rhMo < \rh < -\a_2.
\end{align}
\end{subequations}
Now we derived all relations between $\ph^X$ necessary to deduce the relevant relations between $\mMo$ ($\mMf$) and $\mA$, $\mB$.

First, we note that under \eqref{eq:parameter}, $\mMo>0$ if $\ph<\phA$ and $\mMf>0$ if $\phB<\ph$. 

In the following, we derive all possible relations between $\mA$, $\mB$, and $\mMo$ where we assume that $\ph<\phA$ (otherwise $\M_1$ is unstable, cf.\ Proposition \ref{prop_stab_ME_mphi}, and $\mMo$ is not of particular interest). By \eqref{eq:mA_positiv}, \eqref{eq:mB_positiv}, \eqref{phA<phB}, \eqref{phABt<phA}, \eqref{eq:phMo_phA}, \eqref{eq:phM_equality}, \eqref{eq:phMot_phMft}, and \eqref{eq:phM_interval_2}, we obtain that
\begin{subequations} \label{eq:mMo>mAmB}
\begin{align}
& 0<-\mB <-\mA <\mMo \;\iff\;   \nonumber\\
			    & \quad \quad \quad \begin{cases}
                             \be_2<\a_2 \text{ and } \rh < \rhMo \;\text{ and }\; \phABt < \ph < \phMo, \text{ or}\\
			     \be_2\geq \a_2 \text{ and } \rh < -\be_2 \;\text{ and }\; \ph < \phMo,
                            \end{cases} 
			    \label{eq:-mB<-mA<mMo}\\
& 0<-\mA <-\mB <\mMo \;\iff\;   \nonumber \\
			    & \quad \quad \quad \be_2<\a_2 \;\text{ and }\; 
			     \begin{cases}
			     \rh <  \rhMo \;\text{ and }\; \ph<\phABt, \text{or} \\
			     \rhMo < \rh <-\a_2 \;\text{ and }\; \ph<\phMot,
			    \end{cases}  
			    \label{eq:-mA<-mB<mMo} \\
& 0<-\mA =-\mB <\mMo \;\iff\;  \be_2<\a_2 \;\text{ and }\; 
			      \rh <  \rhMo \;\text{ and }\; \ph=\phABt,
			    \label{eq:-mA=-mB<mMo}
\end{align}
\end{subequations}
and
\begin{subequations} \label{eq:mMo<mAmB}
\begin{align}
& 0<-\mB < \mMo <-\mA  \;\iff\;	\nonumber \\ 
			  & \quad \quad \quad \begin{cases}
                         \be_2 \leq \a_2 \text{ and }  \rh < \rhMot \;\text{ and }\; \phMo <\ph <\phA, \text{or} \\
			 \be_2 \leq \a_2 \text{ and }   \rhMot<\rh<\rhMo \;\text{ and }\; \phMo <\ph <\phMot, \text{or} \\
			 \eqref{eq:be2_a2_condi1} \text{ and } \rh <-\be_2 \;\text{ and }\; \phMo<\ph<\phA, \text{or} \\
			 \eqref{eq:be2_a2_condi1} \text{ and } -\be_2< \rh <\phMot \;\text{ and }\; \ph<\phA, \text{or} \\
			 \eqref{eq:be2_a2_condi1} \text{ and } \rhMot< \rh <-\a_2 \;\text{ and }\; \ph<\phMot, \text{or} \\
			 \eqref{eq:be2_a2_condi2} \text{ and } \rh<\rhMot \;\text{ and }\; \phMo<\ph<\phA, \text{or} \\
			 \eqref{eq:be2_a2_condi2} \text{ and } \rhMot<\rh<-\be_2 \;\text{ and }\; \phMo<\ph<\phMot, \text{or} \\
			 \eqref{eq:be2_a2_condi2} \text{ and } -\be_2<\rh<-\a_2 \;\text{ and }\; \ph<\phMot,
			  \end{cases} 
			  \label{eq:-mB<mMo<-mA}\\
& 0< -\mA < \mMo < -\mB \;\iff\;  \nonumber \\
			    & \quad \quad \quad \be_2<\a_2 \;\text{ and }\;
			      \begin{cases}
			      \rhMo < \rh < -\a_2 \;\text{ and }\; \phMot < \ph < \phMo , \text{or} \\
			      -\a_2 < \rh < -\be_2 \;\text{ and }\; \ph < \phMo ,
                              \end{cases} 
			    \label{eq:-mA<mMo<-mB} \\
& 0< \mMo < -\mB <-\mA \;\iff\; \nonumber \\
			    & \quad \quad \quad \begin{cases}
			     \be_2<\a_2 \text{ and } \rhMot <\rh<\rhMo \;\text{ and }\; \phMot<\ph<\phA, \text{or} \\
                             \be_2<\a_2 \text{ and }  \rhMo <\rh \;\text{ and }\; \phABt<\ph<\phA, \text{or} \\
			     \be_2\geq \a_2 \text{ and } \rhMot< \rh <-\be_2 \;\text{ and }\; \phMot<\ph<\phA, \text{or} \\
			     \be_2\geq \a_2 \text{ and } -\be_2<\rh \;\text{ and }\; \ph<\phA, \text{or}
                            \end{cases}
			      \label{eq:<mMo<-mB<-mA}\\
& 0<\mMo <-\mA <-\mB \;\iff\; \nonumber \\
			    & \quad \quad \quad \be_2<\a_2 \;\text{ and }\; 
			    \begin{cases}
			    \rhMo<\rh<-\be_2 \;\text{ and }\; \phMo<\ph<\phABt, \text{or} \\
			    -\be_2<\rh \;\text{ and }\; \ph<\phABt.
                          \end{cases}
			    \label{eq:<mMo<-mA<-mB}
\end{align}
\end{subequations}
From \eqref{eq:infeasible_m} it follows that
\begin{subequations} \label{eq:mMo_mA_mB_impossible}
\begin{align}
& 0<-\mA <-\mB <\mMo \text{ occurs never if }\; \be_2\geq \a_2, \\
& 0< -\mA < \mMo < -\mB \text{ occurs never if } \; \be_2\geq \a_2,\\
& 0<\mMo <-\mA <-\mB \text{ occurs never if } \; \be_2\geq \a_2.
\end{align}
\end{subequations}

To derive all possible relations between $\mA$, $\mB$, and $\mMf$ we assume $\ph>\phB$ (cf.\ Proposition \ref{prop_stab_ME_mphi}). By \eqref{eq:mA_positiv}, \eqref{eq:mB_positiv}, \eqref{phA<phB}, \eqref{phABt<phA}, \eqref{eq:phB_phMf}, \eqref{eq:phM_equality_2}, and \eqref{eq:phMot_phMft_2}, we obtain that
\begin{subequations} \label{eq:m_relations_mMf}
\begin{align}
& 0<\mA<\mB <\mMf  \;\iff\; \rh<\a_1 \;\text{ and }\; \phMf <\ph, \label{eq:0<mA<mB<mMf}\\
& 0<\mA<\mMf <\mB  \;\iff\; \nonumber \\
		      &\quad \quad \quad \begin{cases}
			\eqref{eq:be1_a1_condi1} \text{ and } \rh<\a_1 \;\text{ and }\; \phB<\ph<\phMf, \text{or}\\
			\eqref{eq:be1_a1_condi1} \text{ and } \a_1<\rh<\rhMft \;\text{ and }\; \phB<\ph, \text{or}\\
			\eqref{eq:be1_a1_condi1} \text{ and } \rhMft<\rh<\be_1 \;\text{ and }\; \phMft<\ph, \text{or}\\
			\eqref{eq:be1_a1_condi2} \text{ and } \rh<\rhMft \;\text{ and }\; \phB<\ph<\phMf, \text{or}\\
			\eqref{eq:be1_a1_condi2} \text{ and } \rhMft<\rh<\a_1 \;\text{ and }\; \phMft<\ph<\phMf, \text{or}\\
			\eqref{eq:be1_a1_condi2} \text{ and } \a_1<\rh<\be_1 \;\text{ and }\; \phMft<\ph,
                      \end{cases} \\
& 0<\mMf<\mA <\mB  \;\iff\; \nonumber\\
		      &\quad \quad \quad \begin{cases}
			\rhMft<\rh<\be_1 \;\text{ and }\; \phB<\ph<\phMft, \text{or} \\
			\be_1<\rh \;\text{ and }\; \phB<\ph.
                      \end{cases}
\end{align}
\end{subequations}
From \eqref{eq:infeasible_m} it follows that 
\begin{subequations}\label{eq:m_relations_mMf_impossible}
\begin{align}
    & 0<\mB<\mA <\mMf  \;\text{ occurs never if } \; \a_1<\be_1, \\
    & 0<\mB<\mMf <\mA  \;\text{ occurs never if } \; \a_1<\be_1, \\
    & 0<\mMf<\mB <\mA    \;\text{ occurs never if } \; \a_1<\be_1.
\end{align}
\end{subequations}

If $\ph<\phA$ or $\ph>\phB$, the relations involving $\mFnull$, $\mA$ and $\mB$ follow immediately by \eqref{eq:mFnull_mMf_mMo}, i.e., by setting $\rh=0$ in the relevant formulas in \eqref{eq:mMo>mAmB}-\eqref{eq:m_relations_mMf}. The remaining cases where $\phA<\ph<\phB$ can be calculated easily using $\tht=\a_1\a_2-\be_1\be_2$, \eqref{mAB_PhAB}, \eqref{eq:ph_rel}, and \eqref{mAB_relations}. Then, all admissible relations are: 
\begin{subequations}\label{eq:crit_m_rel_r0}
\begin{align}
	  0<-\mA<-\mB<-\mFnull &\;\iff\; \a_2>\be_2 \;\text{ and }\; \ph<\phABt, \\
	  0<-\mB<-\mA<-\mFnull &\;\iff\; \begin{cases}
					    \a_2>\be_2 \;\text{ and }\; \phABt<\ph<\phMo, \text{or} \\
					    \a_2\leq \be_2 \;\text{ and }\; \ph<\phMo,
	                                  \end{cases} \\
	  0<\mA<-\mB<-\mFnull &\;\iff\; \tht>0 \;\text{ and }\; \phAB<\ph<\phFnull, \\
	  0<-\mB<\mA<-\mFnull &\;\iff\; \begin{cases}
	                                  \tht<0 \;\text{ and }\; \phAFnull<\ph<\phFnull, \text{or} \\
					   \tht\geq 0 \;\text{ and }\; \phAFnull<\ph<\phAB,
	                                 \end{cases}\\
	  0<\mA<-\mB<\mFnull 	&\;\iff\; \begin{cases}
					 \tht<0 \;\text{ and }\; \phAB<\ph<\phBFnull, \text{or} \\
	                       	         \tht\geq 0 \;\text{ and }\;\phFnull<\ph<\phBFnull,				  
	                       	         \end{cases}\\
	  0<-\mB<\mA<\mFnull  &\;\iff\; \tht<0 \;\text{ and }\; \phFnull<\ph<\phAB, \\
	  0<\mA<\mB<\mFnull   &\;\iff\; \phMf(\rh=0)<\ph, \\
	  0<-\mB<-\mFnull<-\mA &\;\iff\; \phMo(\rh=0)<\ph<\phA, \\
	  0<-\mB<-\mFnull<\mA &\;\iff\; \phA<\ph<\phAFnull, \\
	  0<\mA<\mFnull<-\mB  &\;\iff\; \phBFnull<\ph<\phB, \\
	  0<\mA<\mFnull<\mB   &\;\iff\; \phB<\ph<\phMf(\rh=0).
\end{align}
\end{subequations}
Because other strict inequalities between $\mA$, $\mB$, and $\mFnull$ do not occur, we infer
\begin{equation} \label{mmaxinf<mFnull}
	\min\{|\mA|,|\mB|\} \le |\mFnull|.
\end{equation}
\end{appendix}

\vspace{1cm}

\parindent=0pt
\section*{References}

{\everypar={\hangindent=.5cm \hangafter=1}

Bank, C., B\"urger, R., Hermisson, J. 2012. The limits to parapatric speciation: Dobzhansky-Muller incompatibilities in a continent-island model. Genetics 191, 845-865.

Barton, N.H. 1983. Multilocus clines. Evolution 37, 454-471.

Barton, N.H. 2010. What role does natural selection play in speciation? Phil.\ Trans.\ R.\ Soc.\ B.\ 365, 1825-1840.

Barton, N.H., Bengtsson, B.O. 1986. The barrier to genetic exchange between hybridising populations. Heredity 56, 357-376.

Bengtsson, B.O. 1985. The flow of genes through a genetic barrier. In: Evolution Essays in honour of John Maynard
Smith, Greenwood, J.J., Harvey, P.H., and Slatkin, M. (eds), pp. 31-42. Cambridge: University Press.

Blanquart, F., Gandon, S., Nuismer, S.L. 2012. The effects of migration and drift on local adaptation to a heterogeneous environment. J.\ Evol.\ Biol.\ 25, 1351-1363.

B\"urger, R. 2000. The Mathematical Theory of Selection, Recombination, and Mutation. Wiley, Chichester.

B\"urger, R. 2009a. Multilocus selection in subdivided populations I. Convergence properties for weak or strong migration. J.\ Math.\ Biol.\ 58, 939-978.

B\"urger, R. 2009b. Multilocus selection in subdivided populations II. Maintenance of polymorphism under weak or strong migration. J.\ Math.\ Biol.\ 58, 979-997.

B\"urger, R. 2009c. Polymorphism in the two-locus Levene model with nonepistatic directional selection. Theor.\ Popul.\ Biol.\ 76,
214-228.

B\"urger, R. 2010. Evolution and polymorphism in the multilocus Levene model with no or weak epistasis. Theor.\ Popul.\ Biol.\ 78,
123-138.

B\"urger, R., Akerman, A. 2011. The effects of linkage and gene flow on local adaptation: A two-locus continent-island model. Theor.\ Popul.\ Biol.\ 80, 272-288.

Charlesworth, B., Nordborg, M., Charlesworth, D. 1997. The effects of local selection, balanced polymorphism and background selection on equilibrium patterns of genetic diversity in subdivided populations. Genetical Research 70, 155-174.

Charlesworth, B., Charlesworth, D., 2010. Elements of Evolutionary Genetics. Roberts \& Co., Greenwood Village, Colorado.

Charlesworth, D., Charlesworth, B. 1979. Selection on recombination in clines. Genetics 91, 581-589.

Chasnov, J.R. 2012. Equilibrium properties of a multi-locus, haploid-selection, symmetric-viability model. Theor.\ Popul.\ Biol.\ 
81, 119-130.

Christiansen, F.B., Feldman, M. 1975. Subdivided populations: A review of the one- and two-locus deterministic theory. Theor.\ Popul.\ Biol.\ 7, 13-38.

Conley, C. 1978. Isolated invariant sets and the Morse index. NSF CBMS Lecture Notes 38. Providence, RI.: Amer. Math. Soc.

Deakin, M.A.B. 1966. Sufficient conditions for genetic polymorphism. Amer.\ Natur.\ 100, 690-692.

Ewens, W.J. 1969. Mean fitness increases when fitnesses are additive. Nature 221, 1076.

Eyland, E.A. 1971. Moran's island migration model. Genetics 69, 399-403.

Feder, J.L., Gejji, R., Yeaman, S., Nosil, P. 2012. Establishment of new mutations under divergence and genome hitchhiking. Phil.\ Trans.\ R.\ Soc.\ B.\ 367, 461-474.

Fusco, D., Uyenoyama, M. 2011. Effects of polymorphism for locally adapted genes on rates of neutral introgression in structured populations. Theor.\ Popul.\ Biol.\ 80, 121-131.

Hadeler, K.P., Glas, D. 1983. Quasimonotone systems and convergence to equilibrium in a population genetic model. J.\ Math.\ Anal.\ Appl.\ 95, 297-303.

Haldane, J.B.S. 1930. A mathematical theory of natural and artificial selection. Part VI. Isolation. Proc.\ Camb.\ Phil.\ Soc.\ 26, 220-230.

Karlin, S. 1982. Classification of selection-migration structures and conditions for a protected polymorphism. Evol.\ Biol.\ 14, 61-204.

Karlin, S., McGregor, J. 1972. Polymorphism for genetic and ecological systems with weak coupling. Theor.\ Popul.\ Biol.\ 3, 210-238.

Kawecki, T.J., Ebert, D. 2004. Conceptual issues in local adaptation. Ecology letters 7, 1225-1241.

Kimura, M. 1965. Attainment of quasi linkage equilibrium when gene frequencies are changing by natural selection. Genetics 52, 875-890.

Kobayashi, Y., Hammerstein, P., Telschow, A. 2008. The neutral effective migration rate in a mainland-island context. Theor.\ Popul.\ Biol.\ 74, 84-92.

Kobayashi, Y., Telschow, A. 2011. The concept of effective recombination rate and its application in speciation theory. Evolution 65, 617-628.

Lenormand, T. 2002. Gene flow and the limits to natural selection. Trends Ecol.\ Evol.\ 17, 183-189.

Lenormand, T., Otto, S.P. 2000. The evolution of recombination in a heterogeneous environment. Genetics 156, 423-438.

Leviyang, S., Hamilton, M.B. 2011. Properties of Weir and Cockerham's $\Fst$ estimators and associated bootstrap confidence intervals. Theor.\ Popul.\ Biol.\ 79, 39-52.

Li, W.-H., Nei, M. 1974. Stable linkage disequilibrium without epistasis in subdivided populations. Theor.\ Popul.\ Biol.\ 6, 173-183.

Nagylkaki, T. 1998. Fixation indices in subdivided populations. Genetics 148, 1325-1332.

Nagylaki, T. 2009. Evolution under the multilocus Levene model. Theor.\ Popul.\ Biol.\ 76, 197-213. 

Nagylaki, T. 2011. The influence of partial panmixia on neutral models of spatial variation. Theor.\ Popul.\ Biol.\ 79, 19-38.

Nagylaki, T. 2012. Clines with partial panmixia in an unbounded unidimensional habitat. Theor.\ Popul.\ Biol.\ 82, 22-28.

Nagylaki, T., Hofbauer, J., Brunovsk\'y, P. 1999. Convergence of multilocus systems under weak epistasis or weak selection. J.\ Math.\ Biol.\ 38, 103-133.

Nagylaki, T., Lou, Y. 2001. Patterns of multiallelic polymorphism maintained by migration and selection. Theor.\ Popul.\ Biol.\ 59, 297-313.

Nagylaki, T., Lou, Y. 2007. Evolution under multiallelic migration-selection models. Theor.\ Popul.\ Biol.\ 72, 21-40.

Nagylaki, T., Lou, Y. 2008. The dynamics of migration-selection models. In: Friedman, A.\ (ed) Tutorials in Mathematical Biosciences IV. Lect.\ Notes Math.\ 1922, pp.\ 119 - 172. Springer,  Berlin Heidelberg New York.

Petry, D. 1983. The effect on neutral gene flow of selection at a linked locus. Theor.\ Popul.\ Biol.\ 23, 300-313.

Pylkov,K.V., Zhivotovsky, L.A., Feldman, M.W. 1998. Migration versus mutation in the evolution of recombination under multilocus selection. Genetical Research 71, 247-256.

Shahshahani, S. 1979. A new mathematical framework for the study of linkage and selection. Memoirs Amer. Math. Soc. 211. Providence, R.I.: Amer. Math. Soc.

Slatkin, M. 1975. Gene flow and selection in two-locus systems. Genetics 81, 787-802.

Spichtig, M., Kawecki, T.J. 2004. The maintenance (or not) of polygenic variation by soft selection in heterogeneous environments. Amer.\ Natur.\ 164, 70-84.

Turelli, M., Barton, N.H. 1990. Dynamics of polygenic characters under selection. Theor.\ Pop.\ Biol.\ 38, 1-57.

Weir, B., Cockerham, C. 1984. Estimating $F$ statistics for the analysis of population structure. Evolution 38, 1358-1370.

Wu, C.-I., Ting, C.-T. 2004. Genes and speciation. Nat.\ Rev.\ Genet.\ 5, 114-122.

Yeaman, S., Otto, S.P. 2011. Establishment and maintenance of adaptive genetic divergence under migration, selection, and drift. Evolution 65, 2123-2129.

Yeaman, S., Whitlock, M. 2011. The genetic architecture of adaptation under migration-selection balance. Evolution 65, 1897-1911.

}

\end{document}